\documentclass[10pt]{article}
\usepackage{amsmath}
\usepackage{amssymb}
\usepackage{amsfonts}
\usepackage{amsthm}
\usepackage{mathtools}
\usepackage{bbm}
\usepackage{mathrsfs}
\usepackage{stmaryrd}
\usepackage[parfill]{parskip}
\usepackage{enumitem}
\usepackage[hidelinks, colorlinks=true]{hyperref}
\usepackage[capitalize]{cleveref}
\usepackage[noadjust]{cite}
\usepackage{thmtools}
\usepackage{graphicx}
\usepackage{subcaption}
\usepackage{tikz-cd}
\usepackage{booktabs}
\usepackage{xcolor}
\usepackage{caption}
\usepackage{subcaption}
\usepackage{float}
\usepackage{cancel}
\usepackage{titlesec}
\usepackage[left=2.00cm, right=2.00cm, top=2.00cm, bottom=2.00cm]{geometry}

\let\origcedilla\c

\def\Ab{{\mathbb{A}}}

\def\Cb{{\mathbb{C}}}
\def\Db{{\mathbb{D}}}
\def\Eb{{\mathbb{E}}}

\def\Nb{{\mathbb{N}}}

\def\Rb{{\mathbb{R}}}

\def\Tb{{\mathbb{T}}}

\def\Zb{{\mathbb{Z}}}
\def\Ac{{\mathcal{A}}}

\def\Cc{{\mathcal{C}}}
\def\Dc{{\mathcal{D}}}
\def\Ec{{\mathcal{E}}}
\def\Fc{{\mathcal{F}}}

\def\Hc{{\mathcal{H}}}

\def\Dbf{{\mathbf{D}}}

\def\Xbf{{\mathbf{X}}}

\def\bfr{{\mathfrak{b}}}

\def\gfr{{\mathfrak{g}}}

\def\sfr{{\mathfrak{s}}}
\def\tfr{{\mathfrak{t}}}

\def\cbf{{\mathbf{c}}}
\def\dbf{{\mathbf{d}}}

\def\kbf{{\mathbf{k}}}

\def\mbf{{\mathbf{m}}}
\def\nbf{{\mathbf{n}}}

\def\vbf{{\mathbf{v}}}

\def\xbf{{\mathbf{x}}}

\def\zbf{{\mathbf{z}}}
\def\alphabf{{\boldsymbol{\alpha}}}

\def\deltabf{{\boldsymbol{\delta}}}

\def\zetabf{{\boldsymbol{\zeta}}}
\def\etabf{{\boldsymbol{\eta}}}
\def\thetabf{{\boldsymbol{\theta}}}

\def\sigmabf{{\boldsymbol{\sigma}}}
\def\varsigmabf{{\boldsymbol{\varsigma}}}

\def\varphibf{{\boldsymbol{\varphi}}}

\def\d{{\mathrm{d}}} 
\def\c{{\mathrm{c}}}

\def\i{{\mathrm{i}}}
\def\id{{\mathrm{id}}}

\def\im{\operatorname{im}} 

\def\supp{\operatorname{supp}}

\def\loc{\mathrm{loc}}

\def\XXint#1#2#3{{\setbox0=\hbox{$#1{#2#3}{\int}$ }
	\vcenter{\hbox{$#2#3$ }}\kern-.6\wd0}}

\def\vol{\operatorname{vol}}

\def\law{\stackrel{\mathrm{law}}{=}}

\def\la{\langle} 
\def\ra{\rangle}
\def\fr{\forall\,}

\def\pa{\partial}

\newcommand{\wt}[1]{\widetilde{#1}}

\newcommand{\ol}[1]{\overline{#1}}

\newtheorem{thm}{Theorem}
\numberwithin{thm}{section}
\crefname{thm}{Theorem}{Theorems}
\newtheorem{prop}[thm]{Proposition}
\crefname{prop}{Proposition}{Propositions}
\newtheorem{cor}[thm]{Corollary}
\crefname{cor}{Corollary}{Corollaries}
\newtheorem{lem}[thm]{Lemma}
\crefname{lem}{Lemma}{Lemmas}
\theoremstyle{definition}
\newtheorem{defn}[thm]{Definition}
\crefname{defn}{Definition}{Definitions}
\newtheorem*{notation}{Notation}

\crefname{eg}{Example}{Examples}
\newtheorem*{eg*}{Example}
\theoremstyle{remark}

\crefname{rem}{Remark}{Remarks}
\newtheorem*{rem*}{Remark}

\def\N{{\mathrm{N}}}
\def\D{{\mathrm{D}}}
\def\M{{\mathrm{M}}}
\def\MN{{\mathrm{M_N}}}
\def\MD{{\mathrm{M_D}}}
\def\paN{{\wt{\pa}_\N}}
\def\paD{{\wt{\pa}_\D}}
\def\g{\gfr}
\def\b{\bfr}
\def\bN{\b_\N}
\def\bD{\b_\D}
\def\bM{\b_\M}
\def\bMD{\b_\MD}
\def\btilde{\tilde{\b}}
\def\s{\sfr}
\def\t{\tfr}
\def\h{{\mathrm{h}}}
\def\L{{\mathrm{L}}}

\def\reg{{\mathrm{reg}}}
\def\const{{\mathrm{const}}}
\def\Fr{{\mathrm{Fr}}}

\def\paS{\zetabf^*\paD\Sigma}
\def\bt{{\wt{\varphibf}}}
\def\pc{c^\mathrm{P}}
\def\TR{\Tb_R}

\def\aux{u}

\newcommand{\blue}[1]{\textcolor{black}{#1}}

\title{Boundary Compactified Imaginary Liouville Theory}
\author{Yang Xiao and Yuxiao Xie}
\date{}

\begin{document}
\maketitle

\begin{abstract}
\blue{We generalize the construction of Compactified Imaginary Liouville Theory (CILT), a non-unitary logarithmic Conformal Field Theory (CFT) defined on closed surfaces, to surfaces with boundary. Starting from a compactified Gaussian Free Field (GFF) with Neumann boundary condition, we perturb it by adding in curvature terms and exponential potentials on both the bulk and the boundary. In physics, this theory is conjectured to describe the scaling limit of loop models such as the Potts and $O(n)$ models. To define it mathematically, the curvature terms require a detailed analysis of the topology, and the potential terms are defined using the imaginary Guassian Multiplicative Chaos (GMC).  We prove that the resulting probabilistic path integral satisfies the axioms of CFT, including Segal’s gluing axioms. This work provides the foundation for future studies of boundary CILT and will also help with the understanding of CILT.}
\end{abstract}

\setcounter{tocdepth}{1}
\tableofcontents

\bigskip

\section{Introduction}\label{introduction}

As a special case of Quantum Field Theories (QFTs), Conformal Field Theory (CFT) is a general physical framework that is expected to characterize the small-scale behavior of QFTs. It also plays a central role in understanding critical phenomena in statistical physics, and serves as the building block in string theory, among other applications. Thanks to the rich conformal symmetry in dimension $2$, the groundbreaking work of Belavin--Polyakov--Zamolodchikov \cite{BPZ} successfully solved a class of CFTs (known as the minimal models) under the conformal bootstrap hypothesis. With the help of the Virasoro algebra and the Ward identities, they were able to express the correlation functions as sums of certain special functions (known as the conformal blocks) multiplied by the structure constants. Their work has since greatly inspired mathematicians and led to the application of CFT to modular forms, representation theories of infinite-dimensional Lie algebras and vertex algebras, monstrous moonshine, geometric Langlands theory, and knot theory, to mention a few. On the other hand, several rigorous formulations of CFT have been proposed in the mathematical community: the vertex algebra approach by Kac \cite{Kac98} grew out of the representation-theoretical structure of CFTs; Frenkel and Ben-Zvi \cite{FBZ04} formalized CFTs in the context of algebraic geometry, and this was further expanded by Beilinson and Drinfeld \cite{BD04}; the lecture notes \cite{Ga99} emphasized the probabilistic approach based on the Feynman path integral; Segal \cite{Segal04} proposed an axiomatic formulation inspired by the path integral approach and designed to capture the geometric nature of the conformal bootstrap.

Yet it remains challenging to construct a nontrivial CFT satisfying any given set of axioms. \textcolor{black}{It was not until recently that a mathematical construction of Liouville CFT was brought within reach by the probabilistic approach \cite{LQGH,Pol2d}.} This then led to great success in the mathematical study of this theory, opening the door to the first proof of the DOZZ formula for the structure constant \cite{kupiainen2020integrability}, Segal's axioms \cite{Segal}, the conformal bootstrap \cite{Bootstrap}, and many more results to come.

\textcolor{black}{In addition to the standard Liouville theory, the physics community is also looking for an imaginary version of Liouville CFT that can be used to describe the scaling limit of loop models, a major class of models in statistical physics (see \cite{Jacobsen2009} for instance; for progress in the physics literature, we refer to the review \cite{Ribault:2024rvk} and references therein)}. More recently, a new CFT has been constructed and is expected to serve as a candidate: the \textbf{Compactified Imaginary Liouville Theory} \cite{CILT}, or \textbf{CILT} for short. Specifically, on a closed Riemann surface $\Sigma$ with conformal metric $g$, the \textit{imaginary} Liouville action is defined by
\[S_\L(\phi)=
\frac{1}{4\pi}\int_{\Sigma}|d\phi|_g^2\,\d v_g
+\frac{\i Q}{4\pi}\int_\Sigma K_g\phi\,\d v_g
+\mu\int_\Sigma e^{\i \beta\phi}\,\d v_g,\]
where $v_g$ is the volume, $K_g$ is the scalar curvature, and the parameters are $\mu\in\mathbb C$, $\beta\in\Rb$, $Q=\beta/2-2/\beta$. Here the Liouville field $\phi$ is a map from the surface $\Sigma$ to the circle $\Rb/2\pi R\Zb$ with radius $R>0$. Formally, the path integral of a test functional $F(\phi)$ is
\[\int F(\phi)\,e^{-S_\L(\phi)}\,\D\phi,\]
where $\D\phi$ is the formal ``Lebesgue'' measure on the infinite-dimensional space of maps from $\Sigma$ to $\mathbb R/2\pi R\mathbb Z$. \textcolor{black}{When $\mu=Q=0$, this reduces to the compactified Gaussian Free Field (GFF) (we refer to \cite[Lecture 1.4]{Ga99} and \cite[Subsections 6.3.5 and 10.4.1]{Duplantier:2014daa} for the physics viewpoint, and also \cite[Subsection 2.1.3]{dubedat2015dimers} for a mathematical description)}. \textcolor{black}{To date, it remains a challenge to give a mathematical construction of the \textit{non-compactified} imaginary Liouville theory, which would get rid of the rationality conditions on the parameters.
A possible approach is proposed in \cite{Usciati:2025cdn} with simulations on the DOZZ formula and rigorous heuristics when $\Sigma$ is a circle. Let us also mention another rigorous construction of the timelike Liouville theory \cite{Chatterjee:2025yzo} which produces the DOZZ formula and studies the semiclassical limits under neutrality conditions.} 

In contrast to the real Liouville CFT, the action of CILT is complex-valued, and the field is also compactified. These differences lead to completely new algebraic, geometric, and probabilistic phenomena in CILT. The CFT structure of CILT is also drastically distinct from the real theory. Specifically, the Hamiltonian (the infinitesimal generator of dilations) acting on the canonical Hilbert space is not self-adjoint, which means that CILT is a \textbf{non-unitary} CFT. One also finds that its spectrum is countably infinite (thus \textbf{non-rational}), and that the chiral and anti-chiral parts have different representations (thus \textbf{non-diagonal}). Moreover, the Hamiltonian is \textit{not} diagonalizable, and it is possible to construct explicit Jordan blocks. This suggests that the correlation functions of CILT may have logarithmic singularities, indicating that CILT is a \textbf{logarithmic CFT}, an \blue{active topic of research} in physics \cite{logcft_Feigin:2006iv,logcft_Fjelstad:2002ei,logcft_Gurarie:1993xq,logcft_Kytola:2009ax,logcft_Mathieu:2007pe,logcft_nivesvivat2021logarithmic,logcft_Pearce:2006sz,read2007associative,logcft_Ruelle:2013eda,logcft_Santachiara:2013gna,Cardy:2013logcft} and mathematics \cite{math_logcft_Bakalov:2021bao,math_logcft_kytola2009sle,math_logcft_Liu:2024ybm}. \textcolor{black}{We stress that results on the Hamiltonian and its spectrum have been obtained and will appear in a future work. The Jordan cell structure and the Virasoro representation are works in progress. The logarithmic behavior is observed at least when the correlation function involves a Jordan partner (see also \cite{Cardy:2013logcft,logcft_Santachiara:2013gna}).}

As demonstrated in \cite{BPZ}, under the assumption that the scaling limit of the discrete statistical physics model corresponds to a CFT, they can predict critical exponents of the model with the help of Virasoro algebra. \textcolor{black}{Thus physicists have been working on the classification of CFTs, and also discovering the CFT aspects in the scaling limits of all kinds of discrete models \cite{BPZ,Cardy:1989da}.} Loop models, for example, have been tackled with the so-called Coulomb gas formalism (see \cite{Cardy:1989da,diFrancesco:1987ses,nienhuis1987coulomb} for instance). Roughly speaking, it maps a loop configuration to a height function by, for example, assigning orientations to each loop, and the height function is generally believed to be some variant of a GFF. In addition, the mapping between the loop and the height function eventually requires the field to be compactified. The Coulomb gas representation also introduces a background charge $Q$ (see the curvature term in the action of CILT), which shifts the central charge 1 (the one for GFF) to $1-6Q^2$. Finally, as argued in \cite{cilt_exp_Jacobsen:1998lbo,cilt_exp_Kondev:1997dy,cilt_exp_kondev1996operator} (and the review \cite{cilt_exp_review_rushkin2007critical}, also some criticism in \cite{cilt_exp_critic_gorbenko2018walking}), an imaginary Liouville potential should be also included in the action, hence physicists reach the path integral used for their studies of loop models. 

In this paper, following and extending the work of \cite{CILT}, we construct the \textbf{Boundary Compactified Imaginary Liouville Theory} (\textbf{BCILT}) and prove Segal's axioms for it. As stressed in \cite{BCFT}, the study of CFTs with boundary is important in understanding the CFT structure and for describing critical phenomena in the presence of a conformal boundary. In the case of the real Liouville theory, the boundary CFT has been constructed \cite{Huang:2018ncv,Remy:2017tws,wu2022liouville}, its structure constants are computed \cite{Remy:2017phd,Remy:2020suk,FZZ,allbdystruc}, and the conformal bootstrap program is initiated in \cite{BLCFT}. This has been shown to help in the study of the theory without boundary. For example, \blue{using the conformal bootstrap for the boundary Liouville theory,} \cite{ghosal2022analyticity} proved the connection between the boundary three-point function and the fusion kernel as proposed in \cite{Ponsot:1999uf}. Similar results have been obtained in physics for rational CFTs \cite{Runkel:1998he,behrend2000boundary,fuchs2005tft}. Therefore, we expect that the same can be achieved for CILT, and our construction of BCILT lays the foundation for further studies in this direction.

Let $\Sigma$ be a Riemann surface with boundary $\pa\Sigma$ and $g$ a conformal metric on $\Sigma$. We shall construct the CFT obtained by quantizing the imaginary Liouville action
\begin{equation}\label{eq_intro_action}
    S_\L(\phi)=
\frac{1}{4\pi}\int_{\Sigma}|d\phi|_g^2\,\d v_g
+\frac{\i Q}{4\pi}\int_\Sigma K_g\phi\,\d v_g+\frac{\i Q}{2\pi}\int_{\pa\Sigma}k_g\phi\,\d\ell_g
+\mu\int_\Sigma e^{\i \beta\phi}\,\d v_g+\int_{\pa\Sigma}\mu_\pa\,e^{\i\frac{\beta}{2}\phi}\,\d\ell_g,
\end{equation}
where $Q,\beta\in\mathbb R$, $\mu\in\mathbb C$, $\mu_\partial:\pa\Sigma\to\Cb$ is a piecewise constant function. Here $K_g$, $v_g$, $k_g$, $\ell_g$ denote the scalar curvature, volume, geodesic curvature, and length with respect to $g$, and the Liouville field $\phi:\Sigma\to\mathbb R/ 2\pi R\mathbb Z$ is a map from the surface to the circle with radius $R>0$. For a test functional $F(\phi)$, we would like to construct the path integral
\[\int F(\phi)\,e^{-S_\L(\phi)}\,\D\phi,\]
where $\D\phi$ is the formal ``Lebesgue'' measure on the space of maps $\phi$ from $\Sigma$ to $\mathbb R/2\pi R\mathbb Z$ with Neumann boundary condition at $\pa\Sigma$. When $\mu=0$, $\mu_\partial=0$, this becomes the Coulomb gas formalism with boundary. Otherwise we fix $Q$ by
\[Q=\frac\beta2-\frac2\beta.\]
The well-definedness of the Liouville action requires that $\beta R, QR\in2\Zb$. Because of the boundary, these conditions are slightly stronger than in CILT (where one requires only that $\beta R,QR\in\mathbb Z$). Nevertheless, they can be satisfied by choosing $R$ appropriately, and we are still in the rational regime ($\frac1\beta\mathbb Z\cap\frac1Q\mathbb Z\not=\{0\}$) studied in \cite{CILT}.

\subsection{Sketch of the construction}
The construction of the path integral begins with the observation that the first term $\int_{\Sigma}|d\phi|_g^2\,\d v_g$ in the Liouville action \eqref{eq_intro_action} is the action of a compactified Gaussian Free Field (GFF). Suppose $\phi$ is a smooth map from $\Sigma$ to $\mathbb R/2\pi R\mathbb Z$, so that $d\phi$ is a well-defined closed $1$-form with coefficients in $2\pi R\Zb$. By the Hodge decomposition, there exist a unique smooth function $f:\Sigma\to\mathbb R$ with mean zero and a unique harmonic $1$-form $\omega$ such that $d\phi=d f+\omega$. Then $\int_{\Sigma}|d\phi|_g^2\,\d v_g=\int_{\Sigma}|df|_g^2\,\d v_g+\int_{\Sigma}|\omega|_g^2\,\d v_g$. We interpret the path integral for the compactified GFF as integrating over the product measure
\[e^{-\frac{1}{4\pi}\int_\Sigma|d\phi|^2\,\d v_g}\,\D\phi := e^{-\frac1{4\pi}\int_{\Sigma}|df|_g^2\,\d v_g}\,\d f \otimes e^{-\frac1{4\pi}\int_{\Sigma}|\omega|_g^2\,\d v_g}\,\d [\omega]\otimes\d c,\]
where the first part is given by the standard GFF with mean zero and Neumann boundary condition at $\pa\Sigma$, $\d [\omega]$ is the counting measure on the first cohomology group $H^1(\Sigma;2\pi R\Zb)$ with coefficients in $2\pi R\Zb$, and $\d c$ is the Lebesgue measure on $\mathbb R/2\pi R\mathbb Z$. The compactified field $\phi$ is then
\[\phi(z) = c+f(z)+\textstyle\int_{x_0}^z\omega,\]
sampled according to the measure above. Here we fix a basepoint $x_0\in\Sigma$ and integrate over any path from $x_0$ to $z$. Since $\phi$ takes values in $\mathbb R/2\pi R\mathbb Z$, this is well-defined. The details of this construction are carried out in \cref{sec_gff}, see in particular \cref{measure1,measure2}.

We interpret the other terms in the Liouville action as functionals on this measure space. One has to be careful since $\phi$ is now multi-valued when viewed as a map to $\Rb$. As a first observation, for the Liouville action to be well-defined, these terms should be invariant under the translation $\phi\mapsto\phi+2\pi R$, which forces
\begin{equation*}
    \frac{QR}{4\pi}\left(\int_\Sigma K_g\,\d v_g+2\int_{\pa\Sigma}k_g\,\d\ell_g\right)\in\mathbb Z
,\qquad\text{and}\qquad e^{\i \beta R\cdot2\pi}=e^{\i\frac{\beta}{2}R\cdot2\pi}=1.
\end{equation*}
By the Gauss--Bonnet formula, these conditions amount to $QR\in\Zb$ and $\beta R\in2\Zb$.

To make sense of the curvature terms $\int_\Sigma K_g\phi\,\d v_g$ and $\int_{\pa\Sigma}k_g\phi\,\d\ell_g$, the idea is to integrate over a domain of full measure on which $\phi$ is single-valued. This is done by cutting the surface along a family of curves called a separating family (defined in \cref{sec_sep_fam}), leaving a null-homologous domain. With appropriate correction terms that depend on the cutting curves, the resulting integral can be shown to be invariant, except that changing the homotopy classes of the cutting curves introduces an integral anomaly, which is why we actually need $QR\in2\Zb$ for a general surface. \cref{curv-term-section} is devoted to this.

Finally, the Liouville potential terms $\int_\Sigma e^{\i \beta\phi}\,\d v_g$ and $\int_{\pa\Sigma}e^{\i\frac{\beta}{2}\phi}\,\d\ell_g$ have been studied as a probabilistic object in \cite{lacoin2015complex,lacoin2023probabilistic}, called the imaginary Gaussian Multiplicative Chaos (GMC). The parameter $\beta$ is restricted to $(0,\sqrt 2)$, and important estimates of the bulk imaginary GMC $\int_\Sigma e^{\i \beta\phi}\,\d v_g$ (with Dirichlet boundaries actually, see Segal's axioms below) are made in \cite{CILT} for the construction of CILT. In Section \ref{GMC}, we apply the techniques used in \cite{CILT} to the terms $\int_\Sigma e^{\i \beta\phi}\,\d v_g$ and $\int_{\pa\Sigma}e^{\i\frac{\beta}{2}\phi}\,\d\ell_g$ on surfaces with Neumann boundary.

As a concrete example, we specialize to the case of the disk. Let $g$ be a Neumann extendible (\blue{meaning that it extends to a symmetric metric on the doubled surface, see \cref{doubling} for the precise definition}) conformal metric on the unit disk $\mathbb D$. We fix a compactification radius $R>0$. As discussed above, we require $\beta\in(0,\sqrt 2)$, $QR\in\Zb$, and $\beta R\in2\Zb$. The first cohomology is trivial, so the curvature terms need no special treatment. Let $X_g$ be the Neumann GFF on $(\mathbb D,g)$. We define the bulk and boundary GMCs $M_{g}(X_{g},\d v_{g})$ and $L_{g}(X_{g},\d \ell_{g})$ by
\begin{equation*}
    M_g(X_g,\d v_g):=\lim_{\varepsilon\to0}\varepsilon^{-\beta^2/2}e^{\i\beta X_{g,\varepsilon}(x)}\, \d v_g(x) ,\ \ \  L_g(X_g,\d \ell_g):=\lim_{\varepsilon\to0}\varepsilon^{-\beta^2/4}e^{\i\beta /2X_{g,\varepsilon}(x)}\, \d \ell_g(x),
\end{equation*}
where $X_{g,\varepsilon}$ is a suitable regularization of the random distribution $X_g$ at scale $\varepsilon$ with respect to the metric $g$. The convergence is nontrivial, and we refer to \cite{lacoin2015complex} for the details. Writing $\phi=c+X_g$, the path integral for a test functional $F(\phi)$ is
\begin{equation*}
    \langle F\rangle_{\mathbb D,g}:=C(g)\int_{\mathbb R/2\pi R\mathbb Z}\mathbb E\left[F(\phi)\,e^{-\frac{\i Q}{4\pi}\int_\Sigma K_g\phi\,\d v_g-\frac{\i Q}{2\pi}\int_{\pa\Sigma}k_g\phi\,\d\ell_g-\mu M_g(\phi,\mathbb D)-L_g(\phi, \mu_\pa1_{\pa\mathbb D})}\right]\d c,
\end{equation*}
where $C(g)$ is an explicit constant depending only on the metric $g$ and the expectation is over the GFF $X_g$.

Let us introduce an important type of functionals called \textbf{electric operators}. For $\alpha,\eta\in\Rb$ such that $\alpha R\in \mathbb Z$, $\eta R\in 2\mathbb Z$, we define formally the electric operators $V_{\alpha}(z)$, $V_{\eta}(x)$ at $z\in\mathbb D$, $x\in\pa\mathbb D$ as the $\varepsilon\to0$ limit of
\begin{equation*}
    V_{\alpha,g,\varepsilon}(z):=\varepsilon^{-\frac{\alpha^2}2}e^{\i\alpha \phi_{g,\varepsilon}(z)},V_{\eta,g,\varepsilon}(x):=\varepsilon^{-\frac{\eta^2}4}e^{\i\frac\eta2 \phi_{g,\varepsilon}(x)},
\end{equation*}
where $\phi_{g,\varepsilon}=c+X_{g,\varepsilon}$, and the regularization is necessary since $\phi$ does not make sense pointwise. We define the path integral with electric insertions as
\begin{equation*}
    \langle FV_{\alphabf,\etabf}(\zbf,\xbf)\rangle_{\mathbb D,g}:=C\lim_{\varepsilon\to0}\int_{\mathbb R/2\pi R\mathbb Z}\mathbb E\left[F(\phi)\,V_{g,\varepsilon}(\phi)\,e^{-\frac{\i Q}{4\pi}\int_{\mathbb D} K_g\phi\,\d v_g-\frac{\i Q}{2\pi}\int_{\pa\mathbb D}k_g\phi\,\d\ell_g-\mu M_g(\phi,\mathbb D)-\mu_\pa L_g(\phi, \mu_\pa1_{\pa\mathbb D})}\right]\d c,
\end{equation*}
where $\zbf=(z_1,\cdots,z_\s)$, $\alphabf=(\alpha_1,\cdots,\alpha_\s)$ are the data of bulk insertions, and similarly for the boundary data $\xbf$, $\etabf$, and $V_{g,\varepsilon}(\phi)=\prod_{j=1}^\s V_{\alpha_j,g,\varepsilon}(z_j)\prod_{j=1}^\t V_{\eta_j,g,\varepsilon}(x_j)$. Here the constant $C$ depends on $(g,\zbf,\xbf,\alphabf,\etabf)$. Since $V_{g,\varepsilon}(\phi)$ should be invariant under the translation $\phi\mapsto\phi+ 2\pi R$, we require that $\sum_{j=1}^\s\alpha_jR+\sum_{j=1}^\t\frac{\eta_j}{2}R\in\Zb$. The existence of this limit is proved in Section \ref{bcilt-corr}, which relies on the imaginary Cameron--Martin theorem for the GMC (Proposition \ref{prop_im_CM}). It turns out that we should assume $\alpha_j,\eta_j>Q$ for all $j$ to guarantee convergence.

Taking $F=1$, we obtain the \textbf{correlation functions} $\langle V_{\alphabf,\etabf}(\zbf,\xbf)\rangle_{\mathbb D,g}$, which are fundamental objects exhibiting rich symmetries. \textcolor{black}{These symmetries, especially Ward identities to be proven (the proof will be in a future work), are essential tools in the bootstrap program demonstrated in \cite{BPZ}. The three-point function on the sphere (which is the structure constant of CILT) is known to be given by the imaginary DOZZ formula \cite{CILT}. In Section \ref{sec_open}, we show that the correlation functions on the disk (which contain the structure constants of BCILT) involve combinations of Dotsenko--Fateev 
and Selberg 
integrals. More precisely, they are expectations of mixed moments of the boundary and bulk GMCs. In particular, when the bulk $\mu=0$, the boundary one-point function has a formula which produces the Fyodorov--Bouchaud formula. Note that there have been several papers on moments of the real boundary GMC \cite{Remy:2017phd,Remy:2020akc,Remy:2020suk} (built with tools like the BPZ equations), and we expect similar results for the imaginary one. It would also be interesting to connect our integrals with the Coulomb gas integrals in \cite{math_logcft_Liu:2024ybm}.}

Apart from electric operators, we can also introduce another type called \textbf{magnetic operators}. This is when we require the field $\phi$ to have a winding number $m\in\Zb$ (called the magnetic charge) around an interior point $z\in \Sigma^\circ=\Sigma\setminus\pa\Sigma$. As before, we put magnetic charges $\mbf=(m_1,\ldots,m_\s)\in\Zb^\s$ at $\zbf$. In order to make sense of the electric operators $e^{\alpha_j\phi(z_j)}$ (even for an otherwise smooth field), we specify a nonzero tangent vector $v_j$ at $z_j$ and define $\phi(z_j)$ to be the limit of $\phi(w)$ as $w$ approaches $z_j$ in the direction of $v_j$. Then the above construction carries over, where $\omega$ is now a closed $1$-form on $\Sigma\setminus\zbf$ with the prescribed cycles around each $z_j$. The combined operators are called the \textbf{electromagnetic operators}. \textcolor{black}{Note that we do not put magnetic operators on the boundary, so electromagnetic operators exist only in the interior $\Sigma^\circ$. They also produce certain types of Coulomb gas integrals, for which we refer to \cite[Section 7]{CILT} and \cite{Neretin:2022pbv} for more explanations.}

As a concrete example, let us consider the case of an annulus $\mathbb A$ with two magnetic insertions $\mbf=(m_1,m_2)\in\Zb^2$ at $\zbf=(z_1,z_2)\in(\Ab^\circ)^2$. See \cref{fig:annulus separating} for an illustration. In the notation of Section \ref{cohomologies}, the set of possible choices of $\omega$ is in bijection with the affine subspace $H^1_{\mbf}(\mathbb A\setminus\zbf)$ of $H^1(\mathbb A\setminus\zbf)$ consisting of cohomology classes with the prescribed cycles $\mbf$ around $\zbf$. We have $\int_{\pa_+\Ab}\omega=\int_{\pa_-\Ab}\omega+m_1+m_2$, where $\pa_\pm\Ab$ are the outer and inner boundaries of $\Ab$, both with counterclockwise orientation. Then the cohomology class of $\omega$ is determined by choosing $\int_{\pa_-\Ab}\omega$ (or $\int_{\pa_+\Ab}\omega$), so that $H^1_{\mbf}(\mathbb A\setminus\zbf)\cong\Zb$, and the measure $e^{-\frac1{4\pi}\int_{\Sigma}|\omega|_g^2\,\d v_g}\,\d [\omega]$ becomes a sum over $\Zb$. To make sense of the curvature terms, the separating family should include curves connecting the insertions with the boundary. Changing this family as in \cref{fig:annulus separating} introduces an anomaly of $2\pi$ by a Gauss--Bonnet calculation.
We define similarly the electromagnetic operators as the limit of $V_{g,\varepsilon}(\phi)$, where $\phi$ now has the multi-valued term $\int_{x_0}^z\omega$.
Then the path integral $\la FV_{\alphabf,\etabf,\mbf}(\vbf,\xbf)\ra_{\Sigma,g}$ on $\Sigma$ is
\[C\lim_{\varepsilon\to0}\sum_{[\omega]\in H^1_\mbf(\Sigma\setminus\zbf)}
e^{-\pi R^2\int_\Sigma|\omega|_g^2\,\d v_g}\, \int_{\mathbb R/2\pi R\mathbb Z}\mathbb E\left[F(\phi)\,V_{g,\varepsilon}(\phi)\,e^{-\frac{\i Q}{4\pi}\int_\Sigma K_g\phi\,\d v_g-\frac{\i Q}{2\pi}\int_{\pa\Sigma}k_g\phi\,\d\ell_g-\mu M_g(\phi,\Sigma)-L_g(\phi,\mu_\pa1_{\pa\Sigma})}\right]\d c,\]
where $\phi=c+X_g+2\pi R\int_{x_0}^z\omega$ ($\omega$ has coefficients in $\Zb$) and $C$ depends on $(g,\zbf,\xbf,\alphabf,\etabf,\mbf,\vbf)$.

\begin{figure}
    \centering
    \begin{subfigure}{0.3\textwidth}
    \centering
        \includegraphics[width=\textwidth]{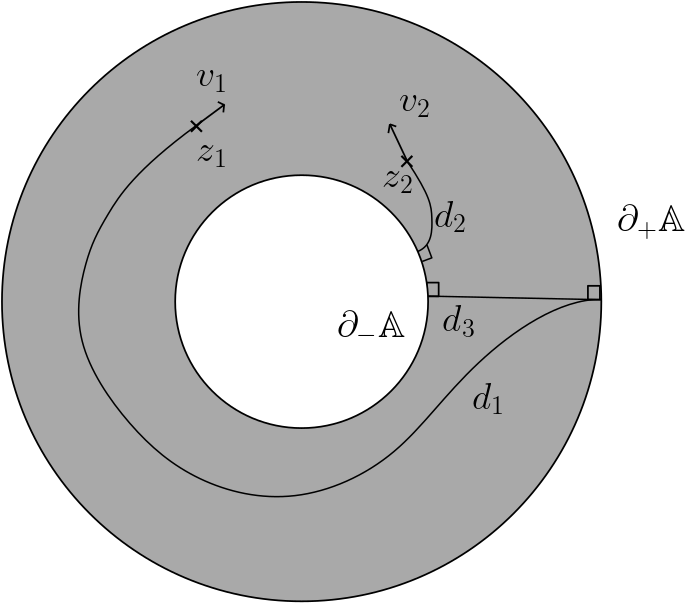}
        \subcaption{left}
    \end{subfigure}
        \begin{subfigure}{0.3\textwidth}
        \centering
        \includegraphics[width=\textwidth]{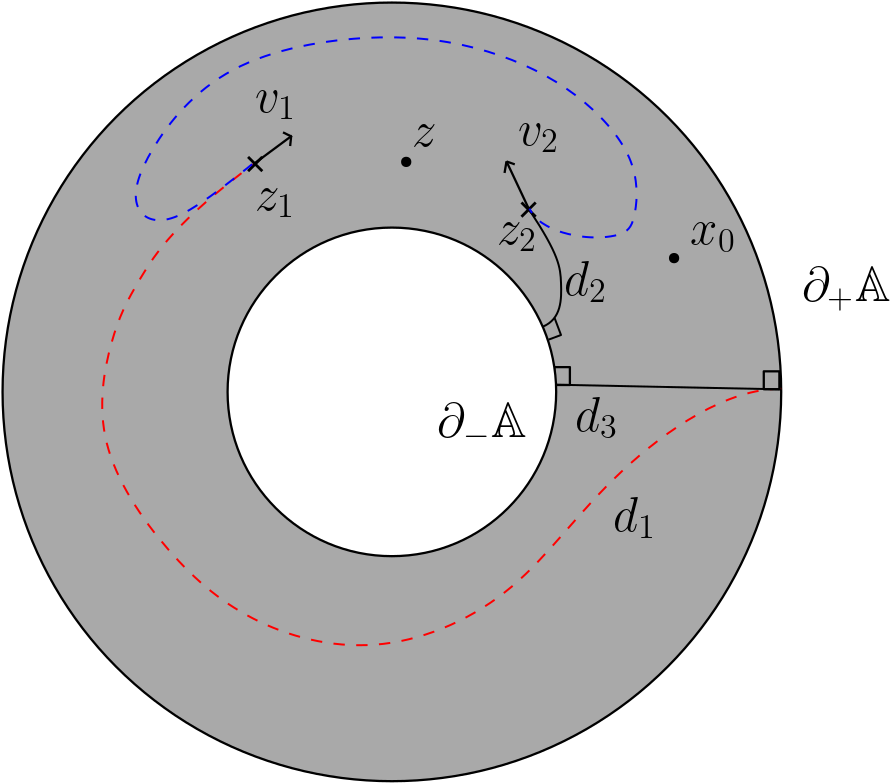}
        \subcaption{right}
    \end{subfigure}
    \caption{On the left is a possible choice of a separating family $(d_1,d_2,d_3)$ on the annulus $\mathbb A$ with two Neumann boundary circles $\pa_\pm\Ab$ and two insertions $(z_1,z_2)$. Note that the orientation of the $d_j$ does not matter. On the right, we obtain a new separating family by replacing $d_1$ with the blue curve between $z_1$ and $z_2$. Then the value of $\int_{x_0}^z\omega$ changes by $m_1$, the winding number around $z_1$.}
    \label{fig:annulus separating}
\end{figure}


As proposed in \cite{Segal04}, an important step in solving the conformal bootstrap of a CFT is constructing the \textbf{Segal functor}. This construction associates to a surface with Dirichlet boundary an operator, called the \textbf{amplitude}, whose integral kernel is given by the path integral over fields with prescribed boundary values at the Dirichlet boundary. Then one can recover the correlation functions by gluing such surfaces along the Dirichlet boundary, which corresponds to composing the amplitude operators. In boundary CFTs, the additional complexity is that one can have both the Neumann boundary and the Dirichlet boundary, and they create corners at the intersection. We refer to \cite{Segal,BLCFT,CILT} for previous works in this direction for the real Liouville CFT and CILT without boundary. For BCILT, the topological structure further complicates the problem, as already seen in defining the correlation functions. One of the main contributions of this paper is a clear framework that facilitates such generalizations.

\subsection{Statement of the main results}

In order to describe the general situation in a precise manner, we propose the following definition:

\begin{defn}\label{extended-surface}
An \textit{extended surface} consists of the following data:
\begin{enumerate}[label=(\arabic*)]
\item $\Sigma$ is a compact Riemann surface with corners (see, e.g., \cite[Definition 2.2]{BLCFT}) of genus $\g$.
\item We put labels on the boundary circles of $\Sigma$ so that we have a partition $\pa\Sigma=\pa_\N\Sigma\sqcup\pa_\D\Sigma\sqcup\pa_\M\Sigma$ where:
    \begin{itemize}
    \item $\pa_\N\Sigma$ consists of $\bN$ real-analytic circles $c_1^\N,\ldots,c_{\bN}^\N$ with Neumann boundary condition.
    \item $\pa_\D\Sigma$ consists of $\bD$ real-analytic circles $c_1^\D,\ldots,c_{\bD}^\D$ with Dirichlet boundary condition.
    \item $\pa_\M\Sigma$ consists of $\bM$ piecewise real-analytic circles $c_1^\M,\ldots,c_{\bM}^\M$ with mixed boundary condition. More precisely, each $c_i^\M$ consists of $k_i$ real-analytic semicircles $c_{i,1}^\MN,\ldots,c_{i,k_i}^\MN$ with Neumann boundary condition and $k_i$ real-analytic semicircles $c_{i,1}^\MD,\ldots,c_{i,k_i}^\MD$ with Dirichlet boundary condition. We assume that $c_{i,j}^\MN\cap c_{i,j'}^\MN=c_{i,j}^\MD\cap c_{i,j'}^\MD=\varnothing$ for $j\neq j'$ and that the intersections $c_{i,j}^\MN\cap c_{i,j'}^\MD$ are orthogonal (if nonempty), i.e., the corners of $\Sigma$ are right angles.
    \end{itemize}
    The precise meaning of the labels will become evident shortly. We write $\pa_\MN\Sigma=\bigsqcup_{i=1}^{\bM}\bigsqcup_{j=1}^{k_i}c_{i,j}^\MN$, $\paN\Sigma=\pa_\N\Sigma\sqcup\pa_\MN\Sigma$, and likewise $\pa_\MD\Sigma$, $\paD\Sigma$. In other words, $\paN\Sigma$ is the full Neumann boundary, whereas $\pa_\N\Sigma$ contains only the Neuamnn circles, and similarly for $\D$.
\item $\zetabf=(\zeta_1^\D,\ldots,\zeta_{\bD}^\D,\zeta_1^\MD,\ldots,\zeta_{\bMD}^\MD)$ where:
    \begin{itemize}
    \item $\zeta_1^\D,\ldots,\zeta_{\bD}^\D$ are real-analytic parametrizations $\Tb=\{e^{\i\theta}:\theta\in\Rb\}\to\pa_\D\Sigma$ of the Dirichlet circles $c_1^\D,\ldots,c_{\bD}^\D$ in $\pa_\D\Sigma$.
    \item $\zeta_1^\MD,\ldots,\zeta_{\bMD}^\MD$ are Neumann extendible (see \cref{doubling}) real-analytic parametrizations $\Tb_+=\{e^{\i\theta}:\theta\in[0,\pi]\}\to\pa_\MD\Sigma$ of the Dirichlet semicircles $c_{1,1}^\MD,\ldots,c_{\bM,k_{\bM}}^\MD$ in $\pa_\MD\Sigma$, where $\bMD=\sum_{i=1}^{\bM}k_i$. We relabel the indices of these semicircles accordingly as $c_1^\MD,\ldots,c_{\bMD}^\MD$ (note that a boundary component may contain several Dirichlet semicircles).
    \end{itemize}
    A Dirichlet circle or semicircle is \textit{outgoing} (resp.\ \textit{incoming}) if the orientation induced by its parametrization coincides with (resp.\ is reverse to) the boundary orientation induced by $\Sigma$ \blue{(we use the convention that the boundary of the unit disk has counterclockwise orientation)}. We write $\varsigma_i^\D=1$ (resp.\ $-1$) if $c_i^\D$ is outgoing (resp.\ incoming), and likewise $\varsigma_i^\MD$ for $c_i^\MD$, so that $\varsigmabf=(\varsigma_1^\D,\ldots,\varsigma_{\bD}^\D,\varsigma_1^\MD,\ldots,\varsigma_{\bMD}^\MD)$ encodes the orientations of the parametrizations in $\zetabf$.
    
    To simplify the notation later on, we write $\paS=\bigsqcup_{i=1}^{\bD}\Tb\sqcup\bigsqcup_{i=1}^{\bMD}\Tb_+$, so that $\zetabf$ can be viewed as a diffeomorphism $\paS\xrightarrow{\sim}\paD\Sigma$. We maintain this distinction to emphasize the fact that the Dirichlet boundary is parametrized.
\item $\zbf=(z_1,\ldots,z_\s)$ are distinct points on $\Sigma^\circ=\Sigma\setminus\pa\Sigma$.
\item $\xbf=(x_1,\ldots,x_\t)$ are distinct points on $\paN\Sigma\setminus\paD\Sigma$.
\end{enumerate}
Such an extended surface is said to be of type $(\g,\bN,\bD,\bM,\bMD,\s,\t)$. See \cref{fig:double} for an illustration.
\end{defn}

\begin{figure}
\centering
\includegraphics{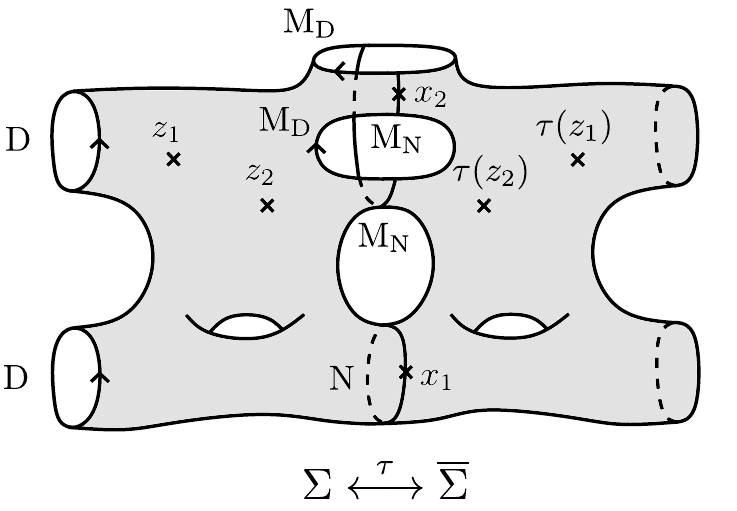}
\caption{An extended surface of type $(1,1,2,1,2,2,2)$ and its double}
\label{fig:double}
\end{figure}

\medskip

Let $\Sigma$ be an extended surface and $g$ a Neumann extendible (see \cref{doubling}) conformal metric on $\Sigma$. Fix a compactification radius $R>0$. We shall give a probabilistic \blue{definition of the formal path integral (see \cref{def:correlation_function} for the case $\paD\Sigma=\varnothing$ and \cref{def:amplitude} for the general case)}
\[\Ac_{\Sigma,g,\alphabf,\etabf,\vbf,\xbf,\mbf,\zetabf}(F,\bt^\kbf)=
\int F(\phi)\prod_{j=1}^\s e^{\i\alpha_j\phi(v_j)}\prod_{j=1}^\t e^{\i\frac{\eta_j}{2}\phi(x_j)}\,e^{-S_\L(\phi)}\,\D\phi,\]
where the integral is over maps $\phi:\Sigma\setminus\zbf\to\Rb/2\pi R\Zb$ satisfying the mixed boundary condition
\[\pa_\nu\phi|_{\paN\Sigma}=0,\qquad\zetabf^*\phi|_{\paD\Sigma}=\bt^\kbf,\]
and having a given winding number $m_j\in\Zb$ around each $z_j\in\Sigma^\circ$. Here $\bt^\kbf:\paS\to\Rb/2\pi R\Zb$ is the Dirichlet boundary value
(\blue{$\kbf\in\Zb^{\b_\D}$ encodes the winding numbers on the Dirichlet circles, see \cref{section:measure-space1}}), $F$ is in a suitable function space containing the constants, $S_\L$ is the imaginary Liouville functional
\[S_\L(\phi)=
\frac{1}{4\pi}\int_{\Sigma}|d\phi|_g^2\,\d v_g
+\frac{\i Q}{4\pi}\int_\Sigma K_g\phi\,\d v_g+\frac{\i Q}{2\pi}\int_{\pa\Sigma}k_g\phi\,\d\ell_g
+\mu\int_\Sigma e^{\i \beta\phi}\,\d v_g+\int_{\paN\Sigma}\mu_\pa\,e^{\i\frac{\beta}{2}\phi}\,\d\ell_g,\]
and the parameters are:
\begin{itemize}
\item $\beta\in(0,\sqrt{2})$ such that $\beta R\in2\Zb$ (or $\Zb$ if $\mu_\pa=0$).
\item $Q=\beta/2-2/\beta$ such that $QR\in4\Zb$ (or $2\Zb$ if $\Sigma$ has no corners). It is called the \textit{background charge}.
\item $\mu\in\Cb$, $\mu_\pa:\paN\Sigma\to\Cb$ is piecewise constant. They are called the \textit{cosmological constants}.
\item $\alphabf=(\alpha_1,\ldots,\alpha_\s)\in\Rb^\s$, $\etabf=(\eta_1,\ldots,\eta_\t)\in\Rb^\t$ such that
\[\begin{cases}\fr j,\,\alpha_j,\eta_j>Q,\\\displaystyle\sum_{j=1}^\s\alpha_jR+\sum_{j=1}^\t\frac{\eta_j}{2}R\in\Zb.\end{cases}\]
They are called the \textit{electric charges}.
\item $\mbf=(m_1,\ldots,m_\s)\in\Zb^\s$. They are called the \textit{magnetic charges}.
\item $\vbf=(v_1,\ldots,v_\s)$ where $v_j$ is a nonzero tangent vector at $z_j$.
\end{itemize}
We remark that both the Neumann boundary $\paN\Sigma$ and the Dirichlet boundary $\paD\Sigma$ can be empty. The case $\paN\Sigma=\varnothing$ is exactly CILT as studied in \cite{CILT}.

\begin{thm}\label{main-theorem}
This defines a Conformal Field Theory in the following sense:
\begin{enumerate}[label=\emph{(\arabic*)}]
\item \emph{(\textbf{Convergence})} The path integral $\Ac_{\Sigma,g,\vbf,\xbf,\alphabf,\mbf,\etabf,\zetabf}(F,\bt^\kbf)$ defines an $L^2$ function in the boundary value $\bt^\kbf$ \blue{with respect to an appropriate measure (defined using the GFF on the (semi)circle, see \cref{section:measure-space1}).} We view it as the Schwartz kernel of an operator $\Ac_{\Sigma,g,\alphabf,\etabf,\vbf,\xbf,\mbf,\zetabf}(F)$ acting on the \blue{$L^2$ space on the measure space of boundary values}, called the \emph{amplitude}. In particular, if $\paD\Sigma=\varnothing$, this defines a complex number, called the \emph{correlation function}.
\item \emph{(\textbf{Weyl anomaly})} This theory is conformally covariant in the following sense. For $\rho\in C^\infty(\Sigma,\Rb)$ Neumann extendible (see \cref{doubling}) with $\rho|_{\paD\Sigma}=0$ if $\paD\Sigma\neq\varnothing$, we have
\begin{align*}
\Ac_{\Sigma,e^\rho g,\alphabf,\etabf,\vbf,\xbf,\mbf,\zetabf}(F)&=\Ac_{\Sigma,g,\alphabf,\etabf,\vbf,\xbf,\mbf,\zetabf}\big(F\big(\cdot-\tfrac{\i Q}{2}\rho\big)\big)\\
&\qquad\exp\left(\frac{c_\L}{96\pi}\int_\Sigma\big(|d\rho|_g^2+2K_g\rho\big)\,\d v_g-\sum_{j=1}^\s\Delta_{\alpha_j,m_j}\rho(z_j)-\frac{1}{2}\sum_{j=1}^\t\Delta_{\eta_j}\rho(x_j)\right),
\end{align*}
where
\[c_\L=1-6Q^2\]
is the \emph{central charge} of the theory and
\[\Delta_{\alpha,m}=\frac{\alpha}{2}\left(\frac{\alpha}{2}-Q\right)+\frac{m^2R^2}{4},\qquad \Delta_{\eta}=\frac{\eta}{2}\left(\frac{\eta}{2}-Q\right)\]
are the \emph{conformal dimensions} of the vertex operators $V_{\alpha,m}(v)=e^{\i\alpha\phi(v)}$, $V_\eta(x)=e^{\i\frac{\eta}{2}\phi(x)}$, respectively.
\item \emph{(\textbf{Spin})} For $\thetabf\in\Rb^\s$, we have
\[\Ac_{\Sigma,g,\alphabf,\etabf,e^{\i\thetabf}\vbf,\xbf,\mbf,\zetabf}(F)=\Ac_{\Sigma,g,\alphabf,\etabf,\vbf,\xbf,\mbf,\zetabf}(F)\,e^{\i R\sum_{j=1}^\s(\alpha_j-Q)m_j\theta_j},\]
where $e^{\i\thetabf}\vbf=(e^{\i\theta_1}v_1,\ldots,e^{\i\theta_\s}v_\s)$.
\item \emph{(\textbf{Gluing})} This defines a projective functor from the category where morphisms are extended surfaces and composition is given by gluing, to the category of Hilbert spaces. For the precise statement, see \cref{gluing-bcilt}.
\end{enumerate}
\end{thm}

\medskip

{The novelty of this work consists mainly of the treatment of the curvature terms in the presence of boundary with possible corners (see \cref{curv-term-section}) and the corresponding estimates for the GMC required for proving convergence (see \cref{GMC}). In particular, we generalize several estimates in \cite{CILT}, like the exponential moments of the GMC and its Cameron--Martin theorem, to surfaces with Neumann boundary.}

\subsection{Open questions and future works}

{For more insights on this CFT, we refer to \cite[Section 1]{CILT}. As another concrete example of a CFT, it would be interesting to understand the relation between BCILT and other \textcolor{black}{important} continuum models like the Schramm--Loewner Evolution\cite{Schramm:1999rp} and Conformal Loop Ensembles\cite{CLE}. They are \textcolor{black}{also powerful ingredients in} the mating-of-trees formalism \cite{Duplantier:2014daa} and are expected to have strong interplay with CFTs (first advocated in physics like \cite{connection_Ikhlef:2015eua} or the review \cite{cilt_exp_review_rushkin2007critical}). In the case of Liouville CFT, the connection between these models has led to mutual better understandings like \cite{SLEipp,allbdystruc,CLEint}, and also unexpected results for statistical physics models like \cite{connection_new_Ang:2022iqv,connection_new_Nolin:2023vkn}.}

There are several questions to be answered for BCILT itself. First, it is unclear how (or whether it is meaningful) to put magnetic charges on the (Neumann) boundary, since having a nontrivial winding number violates the Neumann boundary condition. Second, the structure constants of BCILT are not fully understood yet. In Section \ref{sec_open}, we briefly describe how these constants can be written as certain Coulomb gas integrals, and we hope their explicit forms will be studied somewhere else. The next step would be to investigate the conformal bootstrap for BCILT, which might be easier than the bulk case, see \cite{kogan2000boundary,read2007associative} for similar works in this direction.

\subsection{Organization of the paper}

The paper is organized as follows. \cref{sec_geo_pre} is a compilation of the geometric and topological preliminaries used in this paper. \cref{sec_gff} defines the spaces of random fields in the path integrals. \cref{gluing-free} verifies gluing for the free theory, i.e., the compactified free boson with boundary. \cref{curv-term-section} is devoted to the curvature terms in the Liouville action. \cref{GMC} provides all GMC estimates used in the construction of BCILT. Finally, in \cref{sec_path_int}, we prove \cref{main-theorem}. For the previous sections, the reader who is only interested in the correlation functions (as opposed to Segal's amplitudes) can assume $\paD\Sigma=\varnothing$ and skip \cref{Poisson,gluing-top,section:measure-space1,section:measure-space3,gluing-free}.

\subsection{Acknowledgments}

We would like to thank C. Guillarmou and R. Rhodes for introducing us to this theory and for their comments on the draft of this paper. We also thank Yulai Huang, N. Huguenin, and B. Wu (as well as any others whom we may have unwittingly omitted) for many valuable discussions. Y. Xiao acknowledges that this project has received funding from the European Union’s Horizon 2020 research and innovation programme under the Marie Skłodowska-Curie grant agreement No 101126554, \textcolor{black}{also acknowledges the support of ANR-21-CE40-0003. Y. Xie acknowledges support from the Jean-Pierre Aguilar scholarship from the Fondation CFM pour la Recherche. Both authors acknowledge the support from the Simons foundation grant ``Probabilistic Paths to QFT".}

\section{Geometric preliminaries}\label{sec_geo_pre}

In this section, we collect the geometric and topological ingredients needed for BCILT.

\subsection{Neumann doubling}\label{doubling}

A useful method for studying the Neumann boundary condition is \textit{Cardy's doubling trick}.

Let $\Sigma$ be an extended surface. The \textit{(Neumann) double} of $\Sigma$ is the surface obtained by gluing $\Sigma$ and its oppositely oriented copy $\ol{\Sigma}$ along their Neumann boundaries via the identity map $\paN\Sigma\xleftrightarrow{\sim}\paN\ol{\Sigma}$. We denote it by $\Sigma^{\#2}$. It is naturally an extended surface with $\pa_\N\Sigma^{\#2}=\pa_\M\Sigma^{\#2}=\varnothing$, $\pa_\D\Sigma^{\#2}=\paD\Sigma\cup\paD\ol{\Sigma}$. If $\Sigma$ is of type $(\g,\bN,\bD,\bM,\bMD,\s,\t)$, then $\Sigma^{\#2}$ is of type $(2\g,0,2\bD+\bMD,0,0,2\s,\t)$. See \cref{fig:double} for an illustration.

Consider the \textit{reflection} $\tau$ on $\Sigma^{\#2}$ with respect to $\paN\Sigma=\paN\ol{\Sigma}$, which is an antiholomorphic involution exchanging $\Sigma$ and $\ol{\Sigma}$. A functorial object $u$ on $\Sigma^{\#2}$ is \textit{even} (resp.\ \textit{odd}) if $\tau_*u=u$ (resp.\ $\tau_*u=-u$). This definition applies to functions, tensors, distributions, etc.

Any tensor $u$ on $\Sigma$ has a unique even extension to $\Sigma^{\#2}$, denoted by $u^{\#2}$. We say $u$ is \textit{Neumann extendible} if $u^{\#2}$ is smooth at $\paN\Sigma$. We are interested in the following cases:
\begin{itemize}
\item If a metric $g$ is Neumann extendible, then $\tau$ is an isometry with respect to $g^{\#2}$, and $\paN\Sigma$ is necessarily geodesic, being its fixed point set.
\item If a function $f$ is Neumann extendible, then $f$ satisfies the Neumann boundary condition $\pa_\nu f|_{\paN\Sigma}=0$.
\item If a $1$-form $\omega$ is Neumann extendible, then $i_\nu\omega|_{\paN\Sigma}=0$, where $i_\nu$ denotes contraction with $\nu$.
\end{itemize}

One can view $\Tb$ as the double of $\Tb_+$ with respect to $\pa\Tb_+$. A real-analytic embedding $\zeta:\Tb_+\to\Sigma$ with $\zeta(\pa\Tb_+)\subset\paN\Sigma$ is \textit{Neumann extendible} if its unique extension $\zeta^{\#2}:\Tb\to\Sigma^{\#2}$ with $\zeta^{\#2}(e^{-\i\theta})=\tau(\zeta^{\#2}(e^{\i\theta}))$ is real-analytic. In \cref{extended-surface}, we require the parametrizations $\zeta_i^\MD$ to be Neumann extendible, so that the corresponding parametrizations $(\zeta_i^\MD)^{\#2}$ for $\Sigma^{\#2}$ are real-analytic.

In view of the relation between the Neumann boundary condition and Neumann extendibility, one can think of a boundary CFT on an extended surface $\Sigma$ with $\paN\Sigma\neq\varnothing$ as an ordinary CFT on its double $\Sigma^{\#2}$ where the path integrals are restricted to even fields. This doubling idea will be a guiding principle throughout the paper.

\subsection{Distribution spaces}

For important technical reasons (see \cref{GFF}), it is necessary to integrate over fields that are not even in $L^1_\loc$, i.e., they will generally be distributions. Therefore, we need to specify the appropriate distribution space on an extended surface. As a rule of thumb, we require the space $\Dc$ of test functions to be Neumann extendible and supported away from the Dirichlet boundary, and we define the space $\Dc'$ of distributions to be the subspace of even distributions on the double.

Let $\Sigma$ be an extended surface. If $\paN\Sigma=\varnothing$, we define $\Dc'(\Sigma,\Rb)$ to be the usual distribution space, i.e., the space of continuous linear functionals on the space of compactly supported smooth $2$-forms on $\Sigma$. In general, we define $\Dc'(\Sigma,\Rb)$ to be the space of even distributions on $\Sigma^{\#2}$. We view $C(\Sigma,\Rb)$ as a subspace of $\Dc'(\Sigma,\Rb)$ in the following way. For $u\in C(\Sigma,\Rb)$, we view it as a distribution in $\Dc'(\Sigma,\Rb)$ with
\[\la u,\eta\ra=\frac{1}{2}\int_{\Sigma^{\#2}} u^{\#2}\eta\]
for a test $2$-form $\eta$ on $\Sigma^{\#2}$. Let
\[\Dc(\Sigma,\Rb)=\{f\in C^\infty(\Sigma,\Rb):f\text{ Neumann extendible},\supp f\cap\paD\Sigma=\varnothing\}.\]
Note that $C_c^\infty(\Sigma,\Rb)=\{f\in C^\infty(\Sigma,\Rb):\supp f\subset\Sigma^\circ\}\subset\Dc(\Sigma,\Rb)$, with equality if and only if $\paN\Sigma=\varnothing$. The pairing between $\Dc'(\Sigma,\Rb)$ and $\Dc(\Sigma,\Rb)$ depends on a volume form. Let $g$ be a Neumann extendible metric on $\Sigma$. Then its volume form $\d v_g$ is also Neumann extendible. For $u\in\Dc'(\Sigma,\Rb)$, $f\in\Dc(\Sigma,\Rb)$, we define
\[\int_\Sigma uf\,\d v_g=\la u,(f\,\d v_g)^{\#2}\ra.\]
This is consistent with the usual meaning of the integral for $u\in C(\Sigma,\Rb)$. Note that the constant functions are in $\Dc(\Sigma,\Rb)$ if and only if $\paD\Sigma=\varnothing$. In this case, we define
\[\Dc_0'(\Sigma,\Rb)=\{u\in\Dc'(\Sigma,\Rb):\textstyle\int_\Sigma u\,\d v_g=0\},\]
which depends implicitly on $g$. We have a decomposition $\Dc'(\Sigma,\Rb)=\Rb\oplus\Dc_0'(\Sigma,\Rb)$. The constant part is called the \textit{zero mode}.

For the boundary value, we also need distribution spaces on $\Tb$ and $\Tb_+$. The definitions are similar, except that we use the canonical angle form $d\theta$ to identify $1$-forms with functions. We define $\Dc(\Tb,\Rb)=C^\infty(\Tb,\Rb)$ and $\Dc'(\Tb,\Rb)$ to be the space of continuous linear functionals on the space of smooth functions on $\Tb$. Viewing $\Tb$ as the double of $\Tb_+$, we define $\Dc(\Tb_+,\Rb)=\{f\in C^\infty(\Tb_+,\Rb):f\text{ Neumann extendible}\}\supsetneq C_c^\infty(\Tb_+,\Rb)$ and $\Dc'(\Tb_+,\Rb)$ to be the space of even distributions on $\Tb$, i.e., distributions $u\in\Dc'(\Tb,\Rb)$ such that $\la u,f(e^{\i\theta})\ra=\la u,f(e^{-\i\theta})\ra$ for $f\in C^\infty(\Tb,\Rb)$. We define $\Dc'_0(\Tb,\Rb)$, $\Dc_0'(\Tb_+,\Rb)$ as above, so that we have decompositions $\Dc'(\Tb,\Rb)=\Rb\oplus\Dc_0'(\Tb,\Rb)$, $\Dc'(\Tb_+,\Rb)=\Rb\oplus\Dc_0'(\Tb_+,\Rb)$. For $\wt{\varphi}$ in $\Dc'(\Tb,\Rb)$ or $\Dc'(\Tb_+,\Rb)$, we denote its decomposition by 
\[\wt{\varphi}=c+\varphi,\qquad\text{where }\textstyle c=\frac{1}{2\pi}\int_\Tb\wt{\varphi}\,\d\theta,\;\int_\Tb\varphi\,\d\theta=0,\]
i.e., the presence or absence of the tilde indicates the presence or absence of the zero mode. Finally, for an extended surface $\Sigma$, we define
\[\Dc'(\zetabf^*\paD\Sigma,\Rb)=\Dc'(\Tb,\Rb)^{\bD}\times\Dc'(\Tb_+,\Rb)^{\bMD}.\]
This is consistent with the natural identification $\Fc(\paS)=\Fc(\Tb)^{\bD}\times\Fc(\Tb_+)^{\bMD}$ for any function class $\Fc$ ($\Dc$, $C^k$, $L^p$, etc.), and it is merely for notational convenience. We have a decomposition $\Dc'(\paS,\Rb)=\Rb^{\bD+\bMD}\oplus\Dc_0'(\paS,\Rb)$, and for $\bt=(\wt{\varphi}_1^\D,\ldots,\wt{\varphi}_{\bD}^\D,\wt{\varphi}_1^\MD,\ldots,\wt{\varphi}_{\bMD}^\MD)\in\Dc'(\zetabf^*\paD\Sigma,\Rb)$, we denote its decomposition by $\bt=\cbf+\varphibf$, with the obvious meanings.

\subsection{Laplacians and Green's functions}

Let $\Sigma$ be an extended surface and $g$ a Neumann extendible conformal metric on $\Sigma$. We denote by $\Delta_g=d^*d=-{*}d{*}d$ the geometer's Laplacian with respect to $g$, acting on functions on $\Sigma$. It extends to a continuous map $\Dc'(\Sigma,\Rb)\to\Dc'(\Sigma,\Rb)$ and its restriction to $\Dc(\Sigma,\Rb)$ extends uniquely to an unbounded self-adjoint operator on $L^2(\Sigma,\Rb)$ with compact resolvent. By integration by parts, for $\phi\in C^2(\Sigma,\Rb)$ satisfying the mixed boundary condition $\pa_\nu\phi|_{\paN\Sigma}=\zetabf^*\phi|_{\paD\Sigma}=0$, we have
\[\int_\Sigma|d\phi|_g^2\,\d v_g=\int_\Sigma\phi\,\Delta_g\phi\,\d v_g.\]
For $\rho\in C^\infty(\Sigma,\Rb)$ Neumann extendible,\footnote{This is to ensure that $e^\rho g$ is also Neumann extendible.}
\[\Delta_{e^\rho g}=e^{-\rho}\Delta_g,\qquad\d v_{e^\rho g}=e^\rho\,\d v_g,\qquad|d\phi|_{e^\rho g}^2=e^{-\rho}|d\phi|_g^2,\]
so both sides are \textit{conformally invariant}, i.e., independent of the choice of the conformal metric $g$.

If $\paD\Sigma=\varnothing$, then $\ker\Delta_g=\Rb$ and $\Delta_g:\Dc_0'(\Sigma,\Rb)\to\Dc_0'(\Sigma,\Rb)$ is invertible. The \textit{Green's function} $G_g$ with respect to $g$ is the Schwartz kernel (with respect to $\d v_g$) of this inverse. In other words, for $f\in\Dc(\Sigma,\Rb)$,
\[(R_gf)(x)=\int_\Sigma G_g(x,x')f(x')\,\d v_g(x')\]
solves
\[\Delta_gR_gf=f-\frac{1}{\vol_g\Sigma}\int_\Sigma f\,\d v_g\]
with the Neumann boundary condition $\pa_\nu R_gf|_{\pa\Sigma}=0$ (note that $\pa\Sigma=\pa_\N\Sigma=\paN\Sigma$). If $\pa\Sigma\neq\varnothing$, we have \[G_g(x,x')=G_{g^{\#2}}(x,x')+G_{g^{\#2}}(x,\tau(x'))=G_{g^{\#2}}(x,x')+G_{g^{\#2}}(\tau(x),x').\]
For $\rho\in C^\infty(\Sigma,\Rb)$ Neumann extendible,
\[G_{e^\rho g}(x,x')=G_g(x,x')-f(x)-f(x')+\frac{1}{\vol_{e^\rho g}\Sigma}\int_\Sigma f\,\d v_{e^\rho g},\]
where
\[f(x)=\frac{1}{\vol_{e^\rho g}\Sigma}\int_\Sigma G_g(x,y)\,\d v_{e^\rho g}(y).\]

If $\paD\Sigma\neq\varnothing$, then $\ker\Delta_g=0$ and $\Delta_g:\Dc'(\Sigma,\Rb)\to\Dc'(\Sigma,\Rb)$ is invertible. The Green's function $G_g$ in this case is the Schwartz kernel of $\Delta_g^{-1}$. This is similar to but simpler than the previous case. There is no zero mode due to the Dirichlet boundary condition, and $G_g$ is conformally invariant, so we sometimes also write $G_\Sigma$.

\begin{rem*}
These facts are standard for manifolds without corners (see \cite[Chapter 5]{Taylor} for example). For extended surfaces, one can reduce to this case by doubling, noting that $\Delta_g$ is essentially the restriction of $\Delta_{g^{\#2}}$ on $\Sigma^{\#2}$ to even functions. The same remark applies to the next subsection.
\end{rem*}

\subsection{Poisson and Dirichlet-to-Neumann operators}\label{Poisson}

Let $\Sigma$ be an extended surface with $\paD\Sigma\neq\varnothing$.

Any $\bt\in C(\paS,\Rb)$ admits a unique harmonic extension $P_\Sigma\bt$ to $\Sigma$ with mixed boundary condition. More precisely, $P_\Sigma\bt\in C(\Sigma,\Rb)\cap C^\infty(\Sigma\setminus\paD\Sigma,\Rb)$, $d{*}dP_\Sigma\bt=0$ on $\Sigma\setminus\paD\Sigma$, $\pa_\nu P_\Sigma\bt|_{\paN\Sigma\setminus\paD\Sigma}=0$, $\zetabf^*P_\Sigma\bt|_{\paD\Sigma}=\bt$. If $\paN\Sigma\neq\varnothing$, we have $(P_\Sigma\bt)^{\#2}=P_{\Sigma^{\#2}}\bt^{\#2}$. This defines the \textit{Poisson operator} $P_\Sigma$ on $\Sigma$. It extends to a continuous linear map $\Dc'(\paS,\Rb)\to\Dc'(\Sigma,\Rb)\cap C^\infty(\Sigma\setminus\paD\Sigma,\Rb)$ with Schwartz kernel (with respect to $\d\theta$) in $C^\infty((\Sigma\setminus\paD\Sigma)\times\paS,\Rb)$ given by $K_{P_\Sigma}(x,\cdot)=-\zetabf^*\pa_{\nu_\zetabf}G_\Sigma(x,\cdot)|_{\paD\Sigma}$ for $x\in\Sigma\setminus\paD\Sigma$, where $\nu_\zetabf$ is the outward-pointing normal vector along $\paD\Sigma$ determined by $i_{\nu_\zetabf}(\zetabf_*d\theta)=0$, $i_{\nu_\zetabf}(*\zetabf_*d\theta)=-\varsigmabf$. Note that the Hodge star $*$ on $1$-forms is conformally invariant, so neither $P_\Sigma$ nor $\nu_\zetabf$ depends on a conformal metric.

For a general $\bt\in\Dc'(\paS,\Rb)$, $P_\Sigma\bt$ is singular near $\paD\Sigma$, but the Dirichlet boundary condition $\zetabf^*P_\Sigma\bt|_{\paD\Sigma}=\bt$ can be interpreted in the following weak sense. We extend the parametrizations in $\zetabf$ uniquely to holomorphic or antiholomorphic parametrizations of a tubular neighborhood. Let $c$ be a Dirichlet (semi)circle in $\paD\Sigma$ with parametrization $\zeta$. For $\phi\in C(\Sigma\setminus\paD\Sigma,\Rb)$, we define the \textit{trace} $\zeta^*\phi|_c$ of $\phi$ at $c$ to be the limit of $e^{\i\theta}\mapsto\phi(\zeta(e^{-t+\i\theta}))$ in $\Dc'(\Tb,\Rb)$ or $\Dc'(\Tb_+,\Rb)$ as $t\to0$, if it exists. Since $P_\Sigma$ is continuous as a map $\Dc(\paS,\Rb)\to\Dc(\paS,\Rb)$, it is easy to check that for any $\bt\in\Dc'(\paS,\Rb)$, the trace of $P_\Sigma\bt$ at $\paD\Sigma$ exists and $\zetabf^*P_\Sigma\bt|_{\paD\Sigma}=\bt$ in this sense.

It is well-known that $P_\Sigma$ extends to a continuous linear map $H^s(\paS,\Rb)\to H^{s+1/2}(\Sigma,\Rb)$ for $s\geq-1/2$, where $H^s$ denotes the $L^2$ Sobolev space of order $s\in\Rb$ (see \cite[Chapter 5, Proposition 1.8]{Taylor}).

The \textit{Dirichlet-to-Neumann operator} $\Dbf_\Sigma$ on $\Sigma$ is defined by $\Dbf_\Sigma\bt=\zetabf^*\pa_{\nu_\zetabf}P_\Sigma\bt|_{\paD\Sigma}$ for $\bt\in\Dc(\paS,\Rb)$, where $\nu_\zetabf$ is as above. Its Schwartz kernel $K_{\Dbf_\Sigma}(x,x')=\zetabf_x^*\pa_{\nu_\zetabf}^xK_{P_\Sigma}(x,x')|_{x\in\paD\Sigma}=-\zetabf^*_x\zetabf^*_{x'}\pa_{\nu_{\zetabf}}^x\pa_{\nu_\zetabf}^{x'}G_\Sigma(x,x')|_{\paD\Sigma\times\paD\Sigma}$ is singular at the diagonal. If $\paN\Sigma\neq\varnothing$, we have $(\Dbf_{\Sigma}\bt)^{\#2}=\Dbf_{\Sigma^{\#2}}\bt^{\#2}$. By integration by parts, for $\bt\in C^1(\paS,\Rb)$, we have
\[\int_\Sigma|dP_\Sigma\bt|_g^2\,\d v_g=\int_{\paS}\bt\,\Dbf_\Sigma\bt\,\d\theta,\]
where $g$ is any conformal metric on $\Sigma$. We denote by $\la\cdot,\cdot\ra$ the $L^2$ inner product on $\paS$ with respect to $\d\theta$, so that the right-hand side is $\la\bt,\Dbf_\Sigma\bt\ra$.

\begin{rem*}
Our definition differs from that of \cite{Segal,CILT} in that we use $\zetabf$ to fix the length of the normal vector $\nu_\zetabf$ along $\paD\Sigma$, whereas in \cite{Segal,CILT}, the length is fixed by requiring it to be unit with respect to a conformal metric $g$. These two definitions coincide if $g$ is \textit{admissible} in the terminology there, i.e., $\zetabf^*g=|dz|^2/|z|^2$ in a neighborhood of $\paD\Sigma$. In general, they are conformally equivalent and induce the same $L^2$ inner product on $\paS$ (in their version, the measure on $\paS$ is induced by $g$).
\end{rem*}

Let $\Dbf$ be the Fourier multiplier by $(|n|)_{n\in\Zb}$ on $\Dc'(\Tb,\Rb)$ and $\Dc'(\Tb_+,\Rb)$, i.e., $\Dbf1=0$, $\Dbf\cos n\theta=|n|\cos n\theta$, $\Dbf\sin n\theta=|n|\sin n\theta$ for $n\neq0$. It coincides with the Dirichlet-to-Neumann operators on the unit disk $\Db=\{z:|z|\leq1\}$ with Dirichlet boundary $\Tb$ and the unit half disk $\Db_+=\{z\in\Db:\im z\geq0\}$ with Dirichlet boundary $\Tb_+$ and Neumann boundary $\Db_+\cap\Rb=[-1,1]$. We use the same symbol $\Dbf$ to denote the corresponding product operator on $\Dc'(\paS,\Rb)$. Since the Green's function on any Riemann surface has the same $\log$ singularity at the diagonal, it is not hard to see that for any extended surface $\Sigma$, the difference $\Dbf_\Sigma-\Dbf$ is a smoothing operator, i.e., its Schwartz kernel is in $C^\infty(\paS\times\paS,\Rb)$ (see \cite[Lemma 5.7]{BLCFT} for a detailed proof).

\subsection{Separating families}\label{sec_sep_fam}

In preparation for \cref{curv-term-section}, we introduce the following definition:

\begin{defn}\label{sep-family-extended}
Let $\Sigma$ be an extended surface. We relabel the boundary circles $c_i^\N,c_i^\D,c_i^\M$ of $\Sigma$ as $c_1,\ldots,c_\b$ where $\b=\bN+\bD+\bM$. Fix $\vbf=(v_1,\ldots,v_\s)$ where $v_i$ is a nonzero tangent vector at $z_i$. A \textit{separating family} of $\Sigma$ with respect to $\vbf$ is a tuple $\deltabf=(a_1,b_1,\ldots,a_\g,b_\g,d_1,\ldots,d_{\b+\s-1})$ where:
\begin{enumerate}[label=(\arabic*)]
\item $a_1,b_1,\ldots,a_\g,b_\g$ are smoothly embedded circles (i.e., simple closed curves) on $\Sigma^\circ$.
\item $d_1,\ldots,d_{\b+\s-1}$ are smoothly embedded semicircles (i.e., simple paths) on $\Sigma$ such that for each $d_i$:
\begin{itemize}
    \item $d_i^\circ\subset\Sigma^\circ\setminus\zbf$, $\pa d_i\subset(\pa\Sigma\setminus(\pa_\N\Sigma\cap\pa_\D\Sigma))\cup\zbf$ (so that $d_i$ is disjoint from the corners of $\Sigma$).\footnote{By abuse of notation, we identify an (ordered) tuple with the union of its elements. For example, here $\zbf$ means $\bigcup_i\{z_i\}$.}
    \item If $d_i\cap\pa\Sigma\neq\varnothing$, then $d_i$ intersects $\pa\Sigma$ orthogonally.
    \item If $d_i\cap z_j\neq\varnothing$ for some $j$, i.e., $z_j\in\pa d_i$, then $v_j$ is tangent to $d_i$ at $z_j$.
\end{itemize}
Note that only the directions of the $v_i$ are relevant.
\item Each $a_i$ intersects $b_i$ exactly once and the intersection is transversal and positively oriented. There are no other intersections between the curves in $\deltabf$, except for perhaps intersections between the $d_i$ at their ends.
\item $\pa\Sigma\cup\zbf\cup\bigcup_id_i$ is connected. Equivalently, if we view $c_1,\ldots,c_\b,z_1,\ldots,z_\s$ as vertices and $d_1,\ldots,d_{\b+\s-1}$ as edges between them, then they form a tree. In particular, $\zbf\subset\bigcup_id_i$.
\end{enumerate}
We write $\deltabf=\sigmabf\times\dbf$\footnote{We use the Cartesian product notation $\times$ instead of the union notation $\cup$ to emphasize that these are ordered tuples.} where $\sigmabf=(a_i,b_i)_{i=1}^\g$ is an \textit{interior topological basis} of $\Sigma$ and $\dbf=(d_i)_{i=1}^{\b+\s-1}$. Clearly any interior topological basis extends to a separating family.
\end{defn}

\begin{rem*}
This generalizes the following definitions in \cite{CILT}:
\begin{itemize}
    \item A \textit{canonical geometric basis for $H_1(\Sigma)$} is an interior topological basis.
    \item A \textit{canonical geometric basis for $H_1(\Sigma,\pa\Sigma)$} is a separating family in the case $(\g, 0, \bD, 0, 0, 0,0)$.
    \item A \textit{defect graph} is a separating family in the case $(\g, 0, \bD, 0, 0, \s, 0)$.
\end{itemize}
There is one subtle difference: in \cite{CILT}, the tangent vector of a $d_i$ at an end depends on whether $d_i$ goes \textit{from} or \textit{to} that end, whereas in our definition, the orientation of $d_i$ is irrelevant. The latter appears to be more natural, for reasons explained later.
\end{rem*}

\begin{figure}
\centering
\includegraphics{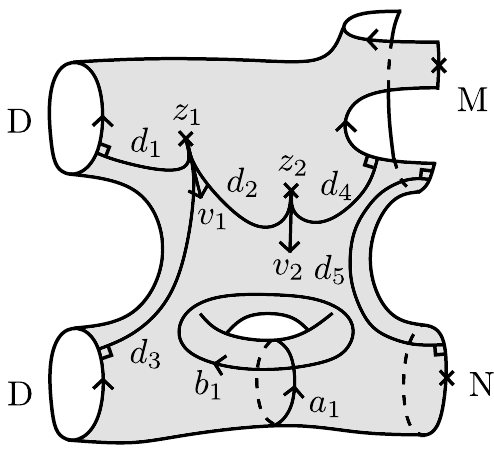}
\caption{A separating family of an extended surface}
\label{sep-family-pic}
\end{figure}

\begin{lem}\label{null-homologous3}
Let $\deltabf$ be a separating family of $\Sigma$.
\begin{enumerate}[label=\emph{(\arabic*)}]
\item $\Sigma\setminus\deltabf$ is connected.
\item Any simple circle on $\Sigma^\circ\setminus\deltabf$ bounds a compact region contained in $\Sigma^\circ\setminus\zbf$.
\item For any simple semicircle $d$ on $\Sigma$ with $\pa d\subset\pa\Sigma\cup\zbf$, $d^\circ\subset\Sigma^\circ\setminus\deltabf$, the set $\Sigma\setminus(\deltabf\cup d)$ consists of two connected components.
\end{enumerate}
\end{lem}
\begin{proof}
Consider the compact surface $\Sigma_\deltabf$ obtained by cutting $\Sigma$ along $\deltabf$. More precisely, we compactify the open surface $\Sigma\setminus\deltabf$ into a compact topological surface with boundary $\Sigma_\deltabf$ using charts near $\deltabf$ on $\Sigma$. Clearly $\Sigma_\deltabf$ is homeomorphic to the sphere minus $\g+1$ disjoint disks, where $\g$ boundary circles correspond to the $a_i,b_i$ and the last boundary circle corresponds to $\pa\Sigma\cup\bigcup_id_i$. Then (1-2) are trivial. As for (3), $d$ corresponds to a semicircle on $\Sigma_\deltabf$ with ends in the last boundary circle, so it is clear that $\Sigma^\circ\setminus(\deltabf\cup d)\cong\Sigma_\deltabf^\circ\setminus d$ consists of two connected components.
\end{proof}

\subsection{Cohomology}\label{cohomologies}

Let $\Sigma$ be an extended surface. In BCILT, the Liouville field is a map $\phi:\Sigma\setminus\zbf\to\Rb/2\pi R\Zb$. Since the target space $\Rb/2\pi R\Zb$ has nontrivial topology, the topology of $\Sigma$ comes into play. In particular, we are interested in the first cohomology groups $H^1(\Sigma\setminus\zbf)$, $H^1(\Sigma\setminus\zbf,\paD\Sigma)$, and the kernel of the natural homomorphism $H^1(\Sigma\setminus\zbf,\paD\Sigma)\to H^1(\Sigma\setminus\zbf)$, which we denote by $H_\D^1(\Sigma\setminus\zbf)$.\footnote{Unless otherwise noted, homology and cohomology groups take coefficients in $\Zb$.}

Cohomology classes in these cohomology groups can be represented by differential forms:
\begin{align*}
&H^1(\Sigma\setminus\zbf;\Rb)=\{\omega:d\omega=0,i_\nu\omega|_{\paN\Sigma}=0\}/
\{df:\pa_\nu f|_{\paN\Sigma}=0\},\\
&H^1(\Sigma\setminus\zbf,\paD\Sigma;\Rb)=\{\omega:d\omega=0,i_\nu\omega|_{\paN\Sigma}=0,\zetabf^*\omega|_{\paD\Sigma}=0\}/
\{df:\pa_\nu f|_{\paN\Sigma}=f|_{\paD\Sigma}=0\},\\
&H_\D^1(\Sigma\setminus\zbf;\Rb)=\{df:\pa_\nu f|_{\paN\Sigma}=0,f|_{\paD\Sigma}\text{ locally constant}\}/\{df:\pa_\nu f|_{\paN\Sigma}=f|_{\paD\Sigma}=0\},
\end{align*}
where $\omega$ is a smooth $1$-form on $\Sigma\setminus\zbf$, $f$ is a smooth function on $\Sigma\setminus\zbf$. Restricting the coefficients to $\Zb$ amounts to requiring that $\int_c\omega\in\Zb$ for any circle $c$ on $\Sigma\setminus\zbf$ and that $f|_{\paD\Sigma}$ takes values in $\Zb$. To see that these quotients are isomorphic to the standard cohomology groups (e.g., singular cohomology), it suffices to check that the pairing $([\omega],[c])\mapsto\int_c\omega$ induces isomorphisms with the duals of the corresponding homology groups. Since we shall make no use of this fact, we leave the proof to the interested reader. In particular, $H^1(\Sigma\setminus\zbf)\cong\Zb^{2\g+\bN+\bD+\bM+\s-1}$, $H_\D^1(\Sigma\setminus\zbf)\cong\Zb^{\bD+\bMD-1}$.

For each $z_i$, let $\pc_i$ be the boundary of a small disk $D_i$ on $\Sigma$ centered at $z_i$, with the boundary orientation induced by $D_i$. The isotopy class of $\pc_i$ does not depend on the choice of $D_i$. For $\mbf=(m_1,\ldots,m_\s)\in\Zb^\s$, we add the subscript $_\mbf$ to a cohomology group $H^1(-)$ to denote the preimage of $\mbf$ under the homomorphism $H^1(-)\to\Zb^\s,[\omega]\mapsto(\int_{\pc_i}\omega)_{i=1}^\s$. Similarly, for $\kbf=(k_1,\ldots,k_{\bD})\in\Zb^{\bD}$, we add the subscript $_\kbf$ to $H^1(-)$ to denote the preimage of $\kbf$ under the homomorphism $H^1(-)\to\Zb^{\bD},[\omega]\to(\varsigma_i^\D\int_{c_i^\D}\omega)_{i=1}^{\bD}$. Note that the orientations of the parametrizations $\zeta_i^\D$ are used here. For example, $H_{\mbf,\kbf}^1(\Sigma\setminus\zbf)$ is the set of cohomology classes $[\omega]\in H^1(\Sigma\setminus\zbf)$ with $\int_{\pc_i}\omega=m_i$, $\int_{c_i^\D}\omega=\varsigma_i^\D k_i$. It is easy to see that $H_{\mbf,\kbf}^1(\Sigma\setminus\zbf)=\varnothing$ if and only if $\bN=\bM=0$ (i.e., $\pa\Sigma=\pa_\D\Sigma$) and $\sum_{i=1}^\s m_i\neq\sum_{i=1}^{\bD}\varsigma_i^\D k_i$. If it is nonempty, it is affine over $H_{\mathbf{0},\mathbf{0}}^1(\Sigma\setminus\zbf)\cong H_{\mathbf{0}}^1(\Sigma)\cong\Zb^{2\g+\bN+\bM-1}$ in the sense that $H_{\mbf,\kbf}^1(\Sigma\setminus\zbf)=[\omega]+H_{\mathbf{0},\mathbf{0}}^1(\Sigma\setminus\zbf)$ for any $[\omega]\in H_{\mbf,\kbf}^1(\Sigma\setminus\zbf)$.

Concerning the group $H_\D^1(\Sigma\setminus\zbf)$, we have an explicit isomorphism $H_\D^1(\Sigma\setminus\zbf)\xrightarrow{\sim}\Zb^{\bD+\bMD}/\Zb,[df]\mapsto(f|_{c_i^\D})_{i=1}^{\bD}\times(f|_{c_i^\MD})_{i=1}^{\bMD}+\Zb$. This is essentially the inverse of the connecting homomorphism $H^0(\paD\Sigma)\to H^1(\Sigma\setminus\zbf,\paD\Sigma)$ in the long exact sequence of cohomology for the pair $(\Sigma\setminus\zbf,\paD\Sigma)$. In particular, we have a natural isomorphism $H_\D^1(\Sigma)\xrightarrow{\sim}H_\D^1(\Sigma\setminus\zbf)$ induced by restriction. Thus we shall work with $H_\D^1(\Sigma)$ instead.

\subsubsection*{Admissible representatives}

In our application to BCILT, we will take representatives of cohomology classes that satisfy the conditions summarized in the following definition:

\begin{defn}\label{admissible}
A $1$-form $\omega$ on $\Sigma$ is \textit{admissible} if:
\begin{enumerate}[label=(\arabic*)]
\item $\omega$ is smooth on $\Sigma\setminus\zbf$ with $d\omega=0$ pointwise.
\item  Near each $z_i$, $\omega$ is $L^1$ and $d\omega=m_i\delta_{z_i}$, $d{*}\omega=0$ in the sense of distributions, where $m_i=\int_{\pc_i}\omega$, $\delta_z$ denotes the Dirac measure at $z$.
\item $\omega$ is Neumann extendible. In particular, $i_\nu\omega|_{\paN\Sigma}=0$.
\item In a neighborhood of each $c_i^\D$, $(\zeta_i^\D)^*\omega=\varsigma_i^\D k_i\frac{1}{2\pi}d\theta$, where $k_i=\int_{c_i^\D}\omega$.
\item In a neighborhood of each $c_i^\MD$, $(\zeta_i^\MD)^*\omega=0$.
\item $\int_c\omega\in\Zb$ for any circle $c$ on $\Sigma\setminus\zbf$.
\end{enumerate}
A function $f$ on $\Sigma$ is \textit{admissible} if:
\begin{enumerate}[label=(\arabic*)]
\item $f$ is smooth on $\Sigma$ and harmonic near $\zbf$.
\item $f$ is Neumann extendible. In particular, $\pa_\nu f|_{\paN\Sigma}=0$.
\item In a neighborhood of $\paD\Sigma$, $f$ is locally constant and takes values in $\Zb$.
\end{enumerate}
If $f$ is an admissible function, then $df$ is an admissible $1$-form. The converse is false.
\end{defn}

The following lemma describes the behavior of admissible $1$-forms near $\zbf$:

\begin{lem}
Let $\omega$ be an $L_\loc^1$ $1$-form on a simply connected neighborhood $D$ of $0$ in $\Cb$ such that $d\omega=m\delta_0$, $d{*}\omega=0$ in the sense of distributions, where $m\in\Rb$. Then there exists a harmonic function $h$ on $D$ such that $\omega=m\frac{1}{2\pi}d\theta+dh$, where $\theta$ is the argument on $\Cb\setminus0$.
\end{lem}
\begin{proof}
First we check that $dd\theta=2\pi\delta_0$, $d{*}d\theta=0$ in the sense of distributions. Clearly $dd\theta=d{*}d\theta=0$ on $\Cb\setminus0$. For a test function $f\in\Dc(D)$ with $\supp f\subset D_0=\{z:|z|\leq\varepsilon\}$ for $\varepsilon>0$ small, by integration by parts,
\begin{align*}
&\int_Dd\theta\wedge df=-\int_{D_0}\pa_r f\,dr\wedge d\theta=-\int_0^\varepsilon\pa_r\left(\int_0^{2\pi}f(re^{\i\theta})\,\d\theta\right)\d r=2\pi f(0),\\
&\int_Dd\theta\wedge *df=-\int_{D_0}\pa_\theta f\,\frac{1}{r}dr\wedge d\theta=-\int_0^\varepsilon\frac{1}{r}\int_0^{2\pi}\pa_\theta f(re^{\i\theta})\,\d\theta\,\d r=0.
\end{align*}
Let $\omega$ be as in the lemma. Then $\omega_0=\omega-m\frac{1}{2\pi}d\theta$ satisfies $d\omega_0=d{*}\omega_0=0$. By elliptic regularity, $\omega_0$ is smooth. Since $D$ is simply connected, $\omega_0=dh$ for some function $h$ on $D$. Then $\Delta h=d^*\omega=0$.
\end{proof}

\begin{lem}
\mbox{}
\begin{enumerate}[label=\emph{(\arabic*)}]
\item Any cohomology class $[\omega]\in H^1(\Sigma\setminus\zbf)$ has an admissible representative $\omega$.
\item Any cohomology class $[df]\in H_\D^1(\Sigma\setminus\zbf)$ has a representative with an admissible primitive $f$.
\end{enumerate}
\end{lem}
\begin{proof}
By doubling, we may assume $\paN\Sigma=\varnothing$, so that $\pa\Sigma=\pa_\D\Sigma=\paD\Sigma$. Fix a conformal metric $g$ on $\Sigma$. For each $z_i$, let $\xi_i$ be a holomorphic parametrization of a neighborhood of $z_i$.

(1): Let $\omega$ be any smooth representative of a cohomology class $[\omega]\in H^1(\Sigma\setminus\zbf)$. Clearly we may take $\omega$ Neumann extendible. In a neighborhood of each $c_i^\D$, $\omega-(\zeta_i^\D)_*(\varsigma_i^\D k_i\frac{1}{2\pi}d\theta)$ (where $k_i=\int_{c_i^\D}\omega$) is exact, so $\omega=(\zeta_i^\D)_*(\varsigma_i^\D k_i\frac{1}{2\pi}d\theta)+df_i^\D$ for some smooth function $f_i^\D$ defined near $c_i^\D$. Similarly, $\omega=df_i^\MD$ for some smooth function $f_i^\MD$ defined near each $c_i^\MD$, and $\omega=(\zeta_i^{\rm P})_*(m_i\frac{1}{2\pi}d\theta)+df_i^{\rm P}$ (where $m_i=\int_{\pc_i}\omega$, $\zeta_i^{\rm P}$ is a holomorphic parametrization of a neighborhood of $z_i$) for some smooth function $f_i^{\rm P}$ defined near each $z_i$. Let $f$ be a Neumann extendible smooth function on $\Sigma$ such that $f=f_i^\D$ near each $c_i^\D$, $f=f_i^\MD$ near each $c_i^\MD$, $f=f_i^{\rm P}$ near each $z_i$. Then $\omega-df$ is an admissible representative for $[\omega]$.

(2): Let $f$ be any smooth primitive for a cohomology class $[df]\in H_\D^1(\Sigma\setminus\zbf)$. With $\chi_i'$, $\chi$ as above, let $f_0=\sum_{i=1}^{\bD}\chi_i'f|_{c_i^\D}$, then $\chi f_0+(1-\chi)f$ is an admissible primitive for $[df]$.
\end{proof}

Let $g$ be a conformal metric on $\Sigma$. We will need to consider the $L^2$-norm of an admissible $1$-form $\omega$ on $\Sigma$. Since $d\theta$ is not $L^2$ near $0$, we introduce the regularization
\[\int_\Sigma^\reg|\omega|_g^2\,\d v_g=\lim_{\varepsilon\to0}\left(\int_{\Sigma\setminus\bigsqcup_{i=1}^\s D_g(z_i,\varepsilon)}\!|\omega|_g^2\,\d v_g+\frac{1}{2\pi}\sum_{i=1}^\s m_i^2\log\varepsilon\right),\]
where $m_i=\int_{\pc_i}\omega$ (i.e., $[\omega]\in H_\mbf^1(\Sigma\setminus\zbf)$), $D_g(z,r)$ denotes the geodesic disk of radius $r$ centered at $z$ with respect to $g$. It is straightforward to check that this converges. For $\rho\in C^\infty(\Sigma,\Rb)$, we have (see \cite[Lemma 3.11]{CILT})
\[\int_\Sigma^\reg|\omega|_{e^\rho g}^2\,\d v_{e^\rho g}=\int_\Sigma^\reg|\omega|_g^2\,\d v_g+\frac{1}{4\pi}\sum_{i=1}^\s m_i^2\rho(z_i).\]

\subsubsection*{Harmonic representatives}

It will often be convenient to take harmonic representatives of cohomology classes.

\begin{defn}
A $1$-form $\omega$ on $\Sigma$ is \textit{almost harmonic} if it is Neumann extendible, $L^1$, and $d\omega=\sum_{i=1}^\s m_i\delta_{z_i}$, $d{*}\omega=0$ on all of $\Sigma$ in the sense of distributions.
\end{defn}

\begin{lem}\label{harmonic}
\mbox{}
\begin{enumerate}[label=\emph{(\arabic*)}]
\item If $\paD\Sigma=\varnothing$, any cohomology class in $H^1(\Sigma\setminus\zbf)$ has a unique almost harmonic representative.
\item If $\paD\Sigma\neq\varnothing$, for any admissible $1$-form $\omega$ on $\Sigma$, there exists a unique smooth function $f$ on $\Sigma$ with $f|_{\paD\Sigma}=0$ such that $\omega+df$ is almost harmonic.
\end{enumerate}
\end{lem}
\begin{proof}
Let $\omega$ be an admissible representative of a cohomology class $[\omega]\in H^1(\Sigma\setminus\zbf)$.

Existence: Fix a Neumann extendible conformal metric $g$ on $\Sigma$. For a smooth function $f$ on $\Sigma$, $\omega+df$ is almost harmonic if and only if $d{*}df=d{*}\omega$, i.e., $\Delta_gf=d^*\omega$. Now $d^*\omega$ is smooth on $\Sigma$ with $d^*\omega=0$ near $\paD\Sigma\cup\zbf$. If $\paD\Sigma=\varnothing$, then $\int_\Sigma d^*\omega\,\d v_g=-\int_\Sigma d{*}\omega=0$, and we take $f=\Delta_g^{-1}d^*\omega$. If $\paD\Sigma\neq\varnothing$, we take $f=h-P_\Sigma\zetabf^*h|_{\paD\Sigma}$ where $h=\Delta_g^{-1}d^*\omega$.

Uniqueness: The difference of two almost harmonic forms in the same cohomology class in $H^1(\Sigma\setminus\zbf)$ has the form $dh$ with $d{*}dh=0$, i.e., $h$ is harmonic. In particular, it is smooth on $\Sigma$. If $\paD\Sigma=\varnothing$, then $h$ is constant, so $dh=0$. If $\paD\Sigma\neq\varnothing$, then by assumption, $dh=d(f_1-f_2)$ where $f_1|_{\paD\Sigma}=f_2|_{\paD\Sigma}=0$, so $h|_{\paD\Sigma}$ is constant, thus $h$ is constant, so again $dh=0$.
\end{proof}

\begin{rem*}
If $\paD\Sigma\neq\varnothing$, harmonic $1$-forms or functions on $\Sigma$ are generally not admissible, since admissibility requires them to be constant in a neighborhood of $\paD\Sigma$. This latter condition is important for gluing (see the next subsection).
\end{rem*}

For the summability of BCILT correlation functions, we need:

\begin{lem}\label{cohomology-summable}
Suppose $\paD\Sigma=\varnothing$. For $\mbf\in\Zb^\s$, $a>0$, we have $\sum_{[\omega]\in H_\mbf^1(\Sigma\setminus\zbf)}e^{-a\int_\Sigma^\reg|\omega^\h|_g^2\,\d v_g}<\infty$, where $\omega^\h$ denotes the unique almost harmonic representative in $[\omega]$.
\end{lem}
\begin{proof}
Let $\Hc^1(\Sigma)=\{\omega^\h:[\omega]\in H^1(\Sigma)=H_{\mathbf{0}}^1(\Sigma\setminus\zbf)\}$, which is a finite-dimensional lattice. Fix $[\omega_0]\in H_\mbf^1(\Sigma\setminus\zbf)$. The sum is $\sum_{\omega^\h\in\Hc^1(\Sigma)}e^{-a\int_\Sigma^\reg|\omega_0^\h+\omega^\h|_g^2\,\d v_g}$. Now $\int_\Sigma^\reg|\omega_0^\h+\omega^\h|_g^2\,\d v_g=\int_\Sigma^\reg|\omega_0^\h|_g^2\,\d v_g+\int_\Sigma|\omega^\h|_g^2\,\d v_g+2\int_\Sigma\la\omega_0^\h,\omega^\h\ra_g\,\d v_g$, where $\int_\Sigma\la\omega_0^\h,\cdot\ra_g\,\d v_g$ is a linear functional on $\Hc^1(\Sigma)$, so this clearly converges.
\end{proof}

\subsection{Gluing}\label{gluing-top}

Gluing and cutting are inverse operations that produce new extended surfaces. In this subsection, we study the change of topology under these operations, in preparation for \cref{gluing-free,curv-term-section}. We are interested in separating families and the groups $H^1(\Sigma\setminus\zbf)$, $H_\D^1(\Sigma)$.

One can glue two Dirichlet circles or semicircles. One can glue two surfaces or self-glue, i.e., glue two Dirichlet (semi)circles of the same surface. We discuss each of these four cases. For topological considerations, it is necessary to glue one (semi)circle at a time.

\begin{notation}
For a tuple $\Xbf=(X_i)_i$, we write $\Xbf^\c=(X_i)_{i\geq2}$, so that $\Xbf=X_1\times\Xbf^\c=X_1\times X_2\times \Xbf^{\c\c}$. \blue{This notation will come in handy in the discussion that follows.}
\end{notation}

\subsubsection{Gluing two surfaces along a circle}\label{glue1}

Let $\Sigma$, $\Sigma'$ be two extended surfaces. We add a prime $'$ to denote the corresponding object for $\Sigma'$. Suppose $c_i^\D\subset\pa_\D\Sigma$ is outgoing and ${c_j^\D}'\subset\pa_\D\Sigma'$ is incoming, or vice versa. Let $\Sigma\#\Sigma'$ be the surface obtained by gluing $c_i^\D$ and ${c_j^\D}'$ via ${\zeta_j^\D}'\circ(\zeta_i^\D)^{-1}$ (\cref{glue1-pic}). It is naturally an extended surface of type $(\g+\g',\bN+\bN',\bD+\bD'-2,\bM+\bM',\bMD+\bMD',\s+\s',\t+\t')$. We always denote by $\Cc$ the glued (semi)circle, which in this case is the common image of $\zeta_i^\D$, ${\zeta_j^\D}'$ on $\Sigma\#\Sigma'$.

For notational convenience, suppose $i=j=1$.

\begin{itemize}
\item Let $\deltabf=\sigmabf\times\dbf$ be a separating family of $\Sigma$ such that $c_1^\D\cap d_i\neq\varnothing$ if and only if $i=1$, and likewise for $\deltabf'=\sigmabf'\times\dbf'$ of $\Sigma'$. We assume that ${\zeta_1^\D}'\circ(\zeta_1^\D)^{-1}$ sends $c_1^\D\cap d_1$ to ${c_1^\D}'\cap d_1'$, so that $d_1\cup d_1'$ glues to a simple semicircle $d_1^\#$ on $\Sigma\#\Sigma'$. Then $\deltabf\#\deltabf'=(\sigmabf\times\sigmabf')\times(d_1^\#\times\dbf^\c\times\dbf'^\c)$ is a separating family of $\Sigma\#\Sigma'$.

This assumes that $d_1,d_1'$ exist. Otherwise $\Sigma$ or $\Sigma'$ has only one boundary component and no punctures. Then $\deltabf\#\deltabf'=(\sigmabf\times\sigmabf')\times(\dbf^\c\times\dbf'^\c)$ is a separating family of $\Sigma\#\Sigma'$.
\item For $\kbf^\c\in\Zb^{\bD-1}$, $\kbf'^\c\in\Zb^{\bD'-1}$, we have a bijection
\[\setlength\arraycolsep{1pt}
\begin{array}[t]{cccccc}
H^1_{\kbf^\c\times\kbf'^\c}(\Sigma\#\Sigma'\setminus\zbf\cup\zbf')&\xleftrightarrow{\sim\,}&\displaystyle\bigsqcup_{k\in\Zb}&H_{k\times\kbf^\c}^1(\Sigma\setminus\zbf)&\times&H_{k\times\kbf'^\c}^1(\Sigma'\setminus\zbf')\\
{[\omega]}&\mapsto&\int_\Cc\omega&{[\omega|_\Sigma]}&&{[\omega|_{\Sigma'}]}\\
{[\omega^\#]}&\mapsfrom&&{[\omega]}&&{[\omega']}
\end{array}\]
where in the inverse map, we choose admissible representatives $\omega$, $\omega'$ of $[\omega]\in H_{k\times\kbf^\c}^1(\Sigma\setminus\zbf)$, $[\omega']\in H_{k\times\kbf'^\c}^1(\Sigma'\setminus\zbf')$, so that the $1$-form $\omega^\#$ defined to be $\omega$ on $\Sigma$ and $\omega'$ on $\Sigma'$ is admissible on $\Sigma\#\Sigma'$.
\item We have a surjective homomorphism
\[\setlength\arraycolsep{1pt}
\begin{array}[t]{ccccc}
H_\D^1(\Sigma)&\times&H_\D^1(\Sigma')&\to&H_\D^1(\Sigma\#\Sigma')\\
{[df]}&&{[df']}&\mapsto&{[df^\#]}
\end{array}\]
where we choose admissible primitives $f$, $f'$ for $[df]\in H_\D^1(\Sigma)$, $[df']\in H_\D^1(\Sigma')$ such that $f|_{c_1^\D}=f'|_{{c_1^\D}'}=0$, so that the function $f^\#$ defined to be $f$ on $\Sigma$ and $f'$ on $\Sigma'$ is admissible on $\Sigma\#\Sigma'$. Its kernel is $\{([df],[df']):f|_{\paD\Sigma\setminus c_1^\D}=f'|_{\paD\Sigma'\setminus{c_1^\D}'}=0,f|_{c_1^\D}=f'|_{{c_1^\D}'}\in\Zb\}$.
\end{itemize}

\begin{rem*}
It is important that admissible $1$-forms and functions are constant \textit{in a neighborhood of $\paD\Sigma$}, so that $\omega^\#$, $f^\#$ here are smooth at $\Cc$. The same remark applies to the other cases below.
\end{rem*}

\begin{figure}
\centering
\includegraphics{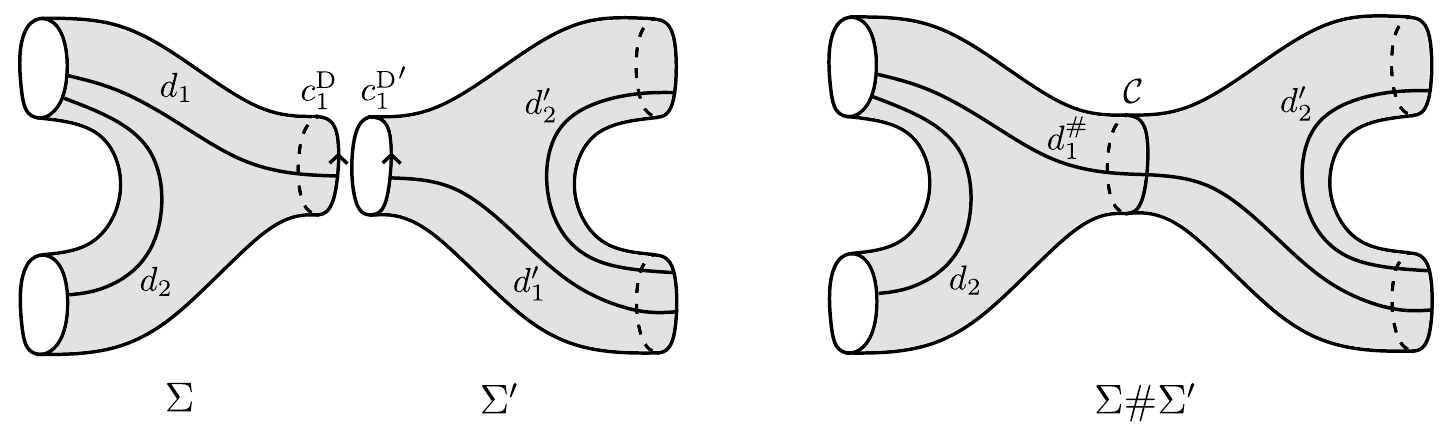}
\caption{Gluing two surfaces along a Dirichlet circle}
\label{glue1-pic}
\end{figure}

\subsubsection{Gluing two surfaces along a semicircle}\label{glue2}

Let $\Sigma$, $\Sigma'$ be two extended surfaces, as in the previous case. Suppose $c_i^\MD\subset\pa_\MD\Sigma$ is outgoing and ${c_j^\MD}'\subset\pa_\MD\Sigma'$ is incoming, or vice versa. Let $\Sigma\#\Sigma'$ be the surface obtained by gluing $c_i^\MD$ and ${c_j^\MD}'$ via ${\zeta_j^\MD}'\circ(\zeta_i^\MD)^{-1}$ (\cref{glue2-pic}). It is naturally an extended surface of type $(\g+\g',\bN+\bN'+n,\bD+\bD',\bM+\bM'-1-n,\bMD+\bMD-2,\s+\s',\t+\t')$ for $n\in\{0,1\}$, depending on whether the new boundary circle is Neumann or mixed.

Again, suppose $i=j=1$. This case is similar to but simpler than the previous case.

\begin{itemize}
\item Let $\deltabf=\sigmabf\times\dbf$ be a separating family of $\Sigma$ and likewise for $\deltabf'=\sigmabf'\times\dbf'$ of $\Sigma'$. Then $\deltabf\#\deltabf'=(\sigmabf\times\sigmabf')\times(\dbf\times\dbf')$ is a separating family of $\Sigma\#\Sigma'$.
\item For $\kbf\in\Zb^{\bD}$, $\kbf'\in\Zb^{\bD'}$, we have a bijection
\[\setlength\arraycolsep{1pt}
\begin{array}[t]{ccccc}
H^1_{\kbf\times\kbf'}(\Sigma\#\Sigma'\setminus\zbf\cup\zbf')&\xleftrightarrow{\sim\,}&H_\kbf^1(\Sigma\setminus\zbf)&\times&H_{\kbf'}^1(\Sigma'\setminus\zbf')\\
{[\omega]}&\mapsto&{[\omega|_\Sigma]}&&{[\omega|_{\Sigma'}]}\\
{[\omega^\#]}&\mapsfrom&{[\omega]}&&{[\omega']}
\end{array}\]
where the inverse map is as in the previous case.
\item We have a surjective homomorphism
\[\setlength\arraycolsep{1pt}
\begin{array}[t]{ccccc}
H_\D^1(\Sigma)&\times&H_\D^1(\Sigma')&\to&H_\D^1(\Sigma\#\Sigma')\\
{[df]}&&{[df']}&\mapsto&{[df^\#]}
\end{array}\]
which is as in the previous case with $f|_{c_1^\D}=f'|_{{c_1^\D}'}=0$ replaced by $f|_{c_1^\MD}=f'|_{{c_1^\MD}'}=0$. Its kernel is $\{([df],[df']):f|_{\paD\Sigma\setminus c_1^\MD}=f'|_{\paD\Sigma'\setminus{c_1^\MD}'}=0,f|_{c_1^\MD}=f'|_{{c_1^\MD}'}\in\Zb\}$.
\end{itemize}

\begin{figure}
\centering
\includegraphics{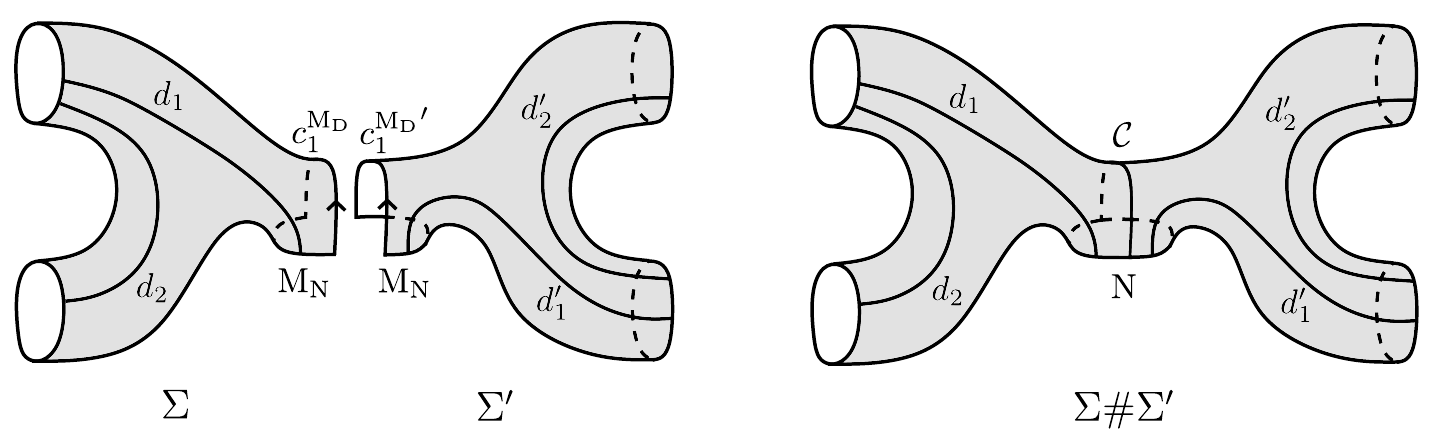}
\caption{Gluing two surfaces along a Dirichlet semicircle}
\label{glue2-pic}
\end{figure}

\subsubsection{Self-gluing a surface along a circle}\label{glue3}

Let $\Sigma$ be an extended surface. Suppose $c_i^\D\subset\pa_\D\Sigma$ is outgoing and $c_j^\D\subset\pa_\D\Sigma$ is incoming. Let $\#\Sigma$ be the surface obtained by gluing $c_i^\D$ and $c_j^\D$ via $\zeta_j^\D\circ(\zeta_i^\D)^{-1}$ (\cref{glue3-pic}). It is naturally an extended surface of type $(\g+1,\bN,\bD-2,\bM,\bMD,\s,\t)$.

For notational convenience, suppose $i=1$, $j=2$.

\begin{itemize}
\item Let $\deltabf=\sigmabf\times\dbf$ be a separating family of $\Sigma$ such that $d_1$ goes from $c_1^\D$ to $c_2^\D$, $c_2^\D\cap d_i=\varnothing$ for $i\neq1$, $c_1^\D\cap d_i=\varnothing$ for $i\neq1,2$. We assume that $\zeta_2^\D\circ(\zeta_1^\D)^{-1}$ sends $c_1^\D\cap d_1$ to $c_2^\D\cap d_1$, so that $d_1$ glues to a simple circle $d^\#_1$ on $\#\Sigma^\circ$. Then $\#\deltabf=(\sigmabf\times\Cc\times d^\#_1)\times\dbf^{\c\c}$ is a separating family of $\#\Sigma$ (recall that $\Cc$ denotes the glued circle). (If $d_2$ does not exist, then $\dbf^{\c\c}$ is replaced by the empty tuple.)
\item Fix an admissible function $h$ on $\Sigma$ with $h|_{\paD\Sigma\setminus c_2^\D}=0$, $h|_{c_2^\D}=1$. By assumption, $dh=0$ near $\paD\Sigma$, so it glues to a $1$-form $(dh)^\#$ on $\#\Sigma$. For $\kbf^{\c\c}\in\Zb^{\bD-2}$, we have a bijection
\[\setlength\arraycolsep{1pt}
\begin{array}[t]{cccccc}
H^1_{\kbf^{\c\c}}(\#\Sigma\setminus\zbf)&\xleftrightarrow{\sim\,}&\displaystyle\bigsqcup_{k\in\Zb}&H_{k\times k\times\kbf^{\c\c}}^1(\Sigma\setminus\zbf)&\times&\Zb\\
{[\omega]}&\mapsto&\int_\Cc\omega&{[\omega|_\Sigma]}&&\int_{d_1^\#}\omega\\
{[\omega^\#+n(dh)^\#]}&\mapsfrom&&{[\omega]}&&n
\end{array}\]
where in the inverse map, we choose an admissible representative $\omega$ of $[\omega]\in H_{k\times k\times\kbf^{\c\c}}^1(\Sigma\setminus\zbf)$ with $\int_{d_1}\omega=0$, so that it glues to an admissible $1$-form $\omega^\#$ on $\#\Sigma$.
\item We have a surjective homomorphism
\[\setlength\arraycolsep{1pt}
\begin{array}[t]{cccccc}
H_\D^1(\Sigma)&\to&H_\D^1(\#\Sigma)&\times&\Zb\\
{[df]}&\mapsto&{[d(f-(f|_{c_2^\D}-f|_{c_1^\D})h)]}&&f|_{c_2^\D}-f|_{c_1^\D}
\end{array}\]
where we choose an admissible primitive $f$ for $[df]\in H_\D^1(\Sigma\setminus\zbf)$, so that $f-(f|_{c_2^\D}-f|_{c_1^\D})h$ glues to an admissible function (denoted in the same way) on $\#\Sigma$. Its kernel is $\{[df]:f|_{\paD\Sigma\setminus c_1^\D\cup c_2^\D}=0,f|_{c_1^\D}=f|_{c_2^\D}\in\Zb\}$.
\end{itemize}

\begin{figure}
\centering
\includegraphics{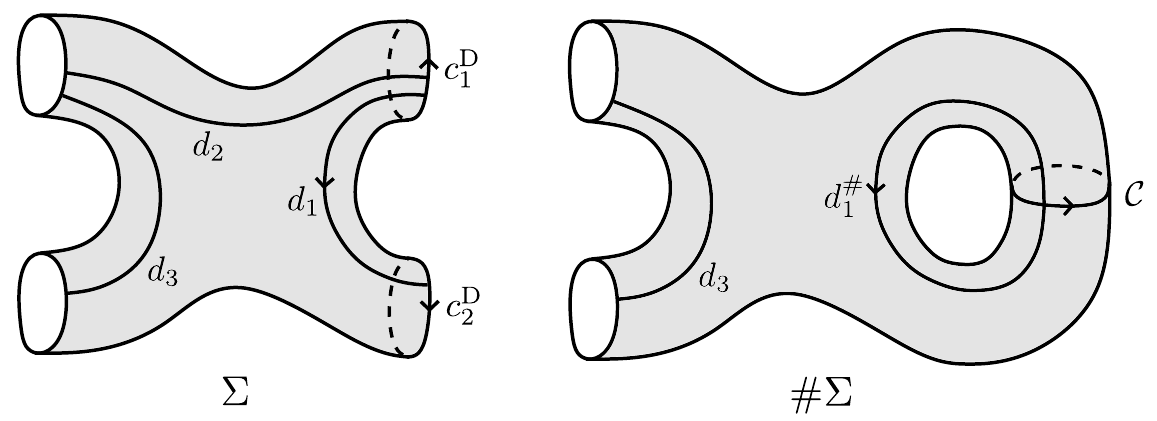}
\caption{Self-gluing a surface along a Dirichlet circle}
\label{glue3-pic}
\end{figure}

\subsubsection{Self-gluing a surface along a semicircle}\label{glue4}

Let $\Sigma$ be an extended surface. Suppose $c_i^\MD\subset\pa_\MD\Sigma$ is outgoing and $c_j^\MD\subset\pa_\MD\Sigma$ is incoming. Let $\#\Sigma$ be the surface obtained by gluing $c_i^\MD$ and $c_j^\MD$ via $\zeta_j^\MD\circ(\zeta_i^\MD)^{-1}$, which is naturally an extended surface. It is necessary to further distinguish two cases, depending on whether $c_i^\MD$, $c_j^\MD$ are in the same mixed boundary circle of $\Sigma$.

\paragraph{Case 1.}
Suppose $c_i^\MD\subset c_{i'}^\M$, $c_j^\MD\subset c_{j'}^\M$, $i'\neq j'$ (\cref{glue4-pic}). Then $\#\Sigma$ is of type $(\g+1,\bN+n,\bD,\bM-1-n,\bMD-2,\s,\t)$ for $n\in\{0,1\}$. For notational convenience, suppose $i=i'=1$, $j=j'=2$.

\begin{itemize}
\item Let $\deltabf=\sigmabf\times\dbf$ be a separating family of $\Sigma$ such that $d_1$ goes from $c_1^\MD$ to $c_2^\MD$, $c_2^\M\cap d_i=\varnothing$ for $i\neq1$, $c_1^\M\cap d_i=\varnothing$ for $i\neq1,2$, $c_1^\MD\cap d_2=\varnothing$. We assume that $\zeta_2^\MD\circ(\zeta_1^\MD)^{-1}$ sends $c_1^\MD\cap d_1$ to $c_2^\MD\cap d_1$, so that $d_1$ glues to a simple circle $d^\#_1$ on $\#\Sigma^\circ$. We move the image of $c_2^\M$ on $\#\Sigma$ slightly to get a smoothly embedded circle $c$ on $(\#\Sigma^\circ\setminus\deltabf)\cup d^\#_1$. Then $\#\deltabf=(\sigmabf\times c\times d^\#_1)\times(d_2\times\dbf^{\c\c})$ is a separating family of $\#\Sigma$. (If $d_2$ does not exist, then $d_2\times\dbf^{\c\c}$ is replaced by the empty tuple.)
\item For $\kbf\in\Zb^{\bD}$, we have a bijection
\[\setlength\arraycolsep{1pt}
\begin{array}[t]{cccccc}
H^1_\kbf(\#\Sigma\setminus\zbf)&\xleftrightarrow{\sim\,}&H_\kbf^1(\Sigma\setminus\zbf)&\times&\Zb\\
{[\omega]}&\mapsto&{[\omega|_\Sigma]}&&\int_{d_1^\#}\omega\\
{[\omega^\#+n(dh)^\#]}&\mapsfrom&{[\omega]}&&n
\end{array}\]
where the inverse map is as in the previous case with $h|_{\paD\Sigma\setminus c_2^\D}=0$, $h|_{c_2^\D}=1$ replaced by $h|_{\paD\Sigma\setminus c_2^\MD}=0$, $h|_{c_2^\MD}=1$.
\item We have a surjective homomorphism
\[\setlength\arraycolsep{1pt}
\begin{array}[t]{cccccc}
H_\D^1(\Sigma)&\to&H_\D^1(\#\Sigma)&\times&\Zb\\
{[df]}&\mapsto&{[d(f-(f|_{c_2^\MD}-f|_{c_1^\MD})h)]}&&f|_{c_2^\MD}-f|_{c_1^\MD}
\end{array}\]
as in the previous case. Its kernel is $\{[df]:f|_{\paD\Sigma\setminus c_1^\MD\cup c_2^\MD}=0,f|_{c_1^\MD}=f|_{c_2^\MD}\in\Zb\}$.
\end{itemize}

\begin{figure}
\centering
\includegraphics{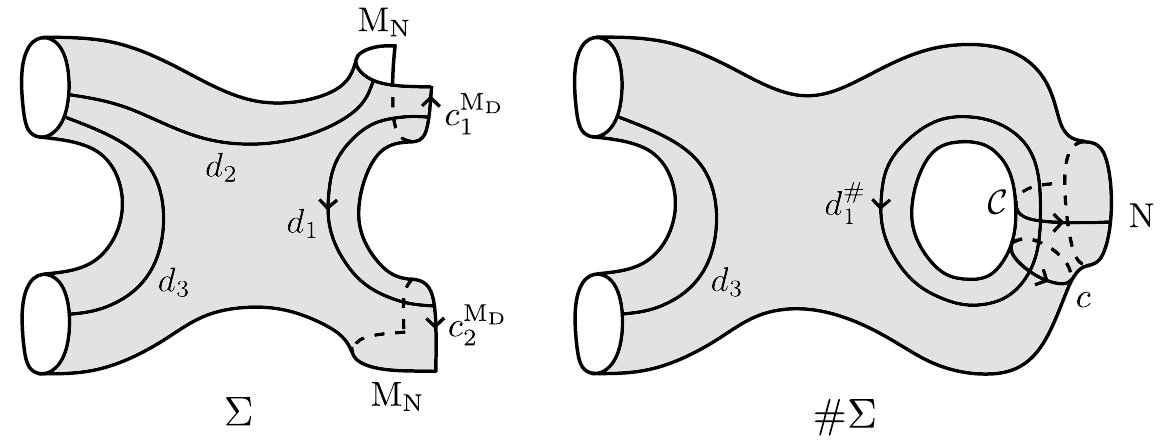}
\caption{Self-gluing a surface along a Dirichlet semicircle (case 1)}
\label{glue4-pic}
\end{figure}

\paragraph{Case 2.}
Suppose $c_i^\MD\cup c_j^\MD\subset c_k^\M$ for some $k$ (\cref{glue4'-pic}). Then $\#\Sigma$ is of type $(\g,\bN+n,\bD,\bM+1-n,\bMD-2,\s,\t)$ for $n\in\{0,1,2\}$. Again, suppose $i=1$, $j=2$.

\begin{itemize}
    \item Let $\deltabf=\sigmabf\times\dbf$ be a separating family of $\Sigma$ disjoint from $c_1^\MD$, $c_2^\MD$. Then $\#\deltabf=\sigmabf\times(\Cc\times\dbf)$ is a separating family of $\#\Sigma$.
    \item  Let $d$ be a simple semicircle on $\Sigma\setminus\zbf$ from $c_1^\MD$ to $c_2^\MD$ such that $\zeta_2^\MD\circ(\zeta_1^\MD)^{-1}$ sends $d\cap c_1^\MD$ to $d\cap c_2^\MD$, so that $d$ glues to a simple circle $d^\#$ on $\#\Sigma$. The homology class of $d^\#$ in $H_1(\#\Sigma\setminus\zbf)$ does not depend on the choice of $d$. For $\kbf\in\Zb^{\bD}$, we have a bijection
\[\setlength\arraycolsep{1pt}
\begin{array}[t]{cccccc}
H^1_\kbf(\#\Sigma\setminus\zbf)&\xleftrightarrow{\sim\,}&H_\kbf^1(\Sigma\setminus\zbf)&\times&\Zb\\
{[\omega]}&\mapsto&{[\omega|_\Sigma]}&&\int_{d^\#}\omega\\
{[\omega^\#+n(dh)^\#]}&\mapsfrom&{[\omega]}&&n
\end{array}\]
as in case 1.
\item We have a surjective homomorphism
\[\setlength\arraycolsep{1pt}
\begin{array}[t]{cccccc}
H_\D^1(\Sigma)&\to&H_\D^1(\#\Sigma)&\times&\Zb\\
{[df]}&\mapsto&{[d(f-(f|_{c_2^\MD}-f|_{c_1^\MD})h)]}&&f|_{c_2^\MD}-f|_{c_1^\MD}
\end{array}\]
as in case 1.
\end{itemize}

\begin{figure}
\centering
\includegraphics{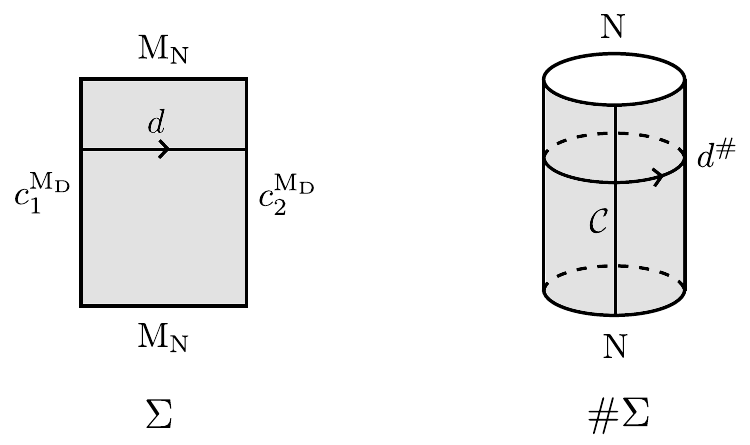}
\caption{Self-gluing a surface along a Dirichlet semicircle (case 2)}
\label{glue4'-pic}
\end{figure}

\section{Construction of the free field}\label{sec_gff}

From now on, we fix $R>0$ and write $\TR=\Rb/2\pi R\Zb$. In this section, we define the measure spaces over which our path integrals take place. More precisely, we give a probabilistic construction of the measure 
\[e^{-\frac{1}{4\pi}\int_\Sigma|d\phi|_g^2\,\d v_g}\,\D\phi\]
on the space of maps $\phi:\Sigma\setminus\zbf\to\TR$. This is known as the \textit{compactified massless free boson}. It is also a \textit{nonlinear $\sigma$-model} with values in a one-dimensional torus.

\begin{notation}
We add a bar to denote the projection $\Rb\to\TR$. For example, if $\phi:\Sigma\to\Rb$, then $\ol{\phi}:\Sigma\to\TR$ is its projection to $\TR$. Maintaining this distinction helps to reduce ambiguity.
\end{notation}

\subsection{Gaussian Free Fields (GFFs)}\label{GFF}

The probabilistic tool here is the \textit{Gaussian Free Field (GFF)}. We briefly review the main ideas while adopting a more measure-theoretic viewpoint.

The general construction is as follows. Let $(M,g)$ be a compact Riemannian manifold of dimension $D$ and $\Delta$ a nonnegative elliptic operator on $M$ of order $k$. We would like to make sense of the measure $e^{-\frac{1}{2}\int_M\phi\,\Delta\phi\,\d v_g}\,\D\phi$ on the space of functions $\phi:M\to\Rb$. The idea is to think of $\Delta$ as an infinite-dimensional positive semidefinite matrix. Let $(e_j)_{j=1}^\infty\subset C^\infty(M,\Rb)$ be an $L^2$-orthonormal basis of eigenfunctions of $\Delta$ with $\Delta e_j=\lambda_je_j$, where $\lambda_j\geq0$ is nondecreasing in $j$. Here if $\pa M\neq\varnothing$, one imposes a boundary condition, which essentially amounts to choosing a self-adjoint extension of $\Delta$ on $L^2(M,\Rb)$. Suppose $\dim\ker\Delta=d$, so that $\lambda_1=\cdots=\lambda_d=0<\lambda_{d+1}\leq\lambda_{d+2}\leq\cdots\nearrow\infty$. The \textit{GFF} of $\Delta$ is defined by the series
\[X_\Delta=\sum_{j=d+1}^\infty\frac{a_j}{\sqrt{\lambda_j}}\,e_j,\]
where the $a_j$ are i.i.d.\ standard Gaussians. By Weyl's law, $\lambda_j\sim Cj^{k/D}$ for some constant $C$ as $j\to\infty$. By elliptic regularity, the $L^2$ Sobolev space of order $s\in\Rb$ on $M$ is $H^s(M,\Rb)=\bigl\{\sum_{j=1}^\infty c_je_j:\sum_{j=1}^\infty\lambda_j^{2s/k}c_j^2<\infty\bigr\}\subset\Dc'(M,\Rb)$. Thus $X_\Delta$ converges a.s.\ in $H^s(M,\Rb)$ for $s<(k-D)/2$. In particular, it defines a random distribution on $M$. Let $\Dc_\Delta'(M,\Rb)=\{u\in\Dc'(M,\Rb):\fr f\in\ker\Delta,\int_\Sigma uf\,\d v_g=0\}$, so that $X_\Delta\in\Dc_\Delta'(M,\Rb)$ a.s. We have a decomposition $\Dc'(M,\Rb)=\ker\Delta\oplus\Dc_\Delta'(M,\Rb)$. The covariance of $X_\Delta$ is the Green's function $G_\Delta$ of $\Delta$ orthogonal to $\ker\Delta$, i.e., the Schwartz kernel of the inverse of $\Delta:\Dc_\Delta'(M,\Rb)\xrightarrow{\sim}\Dc_\Delta'(M,\Rb)$:
\[\Eb[X_\Delta(x)X_\Delta(x')]=G_\Delta(x,x')=\sum_{j=d+1}^\infty\frac{1}{\lambda_j}e_j(x)e_j(x').\]
Identifying $\ker\Delta$ with $\Rb^d$ using the basis $(e_j)_{j=1}^d$, we equip $\Dc_\Delta'(M,\Rb)$ with the probability measure $\d X_\Delta$ induced by $X_\Delta$ and $\Dc'(M,\Rb)$ with the product measure $\d\mu_\Delta=\d c^d\,\d X_\Delta$ where $\d c^d$ is the Lebesgue measure on $\Rb^d$. Intuitively, the total mass of $e^{-\frac{1}{2}\int_\Sigma\phi\,\Delta\phi\,\d v_g}\,\D\phi$ on $\Dc_\Delta'(M,\Rb)$ is $\prod_{j=d+1}^\infty\sqrt{\frac{2\pi}{\lambda_j}}$, which can be formalized as $(\det'\frac{1}{2\pi}\Delta)^{-1/2}$, where $\det'$ denotes the \textit{regularized determinant} of Ray--Singer. Thus we define $e^{-\frac{1}{2}\int_M\phi\,\Delta\phi\,\d v_g}\,\D\phi=(\det'\frac{1}{2\pi}\Delta)^{-1/2}\,\d\mu_\Delta$.

We are interested in the following cases:
\begin{itemize}
\item $M=\Tb$, $g=d\theta^2$, $\Delta=\frac{1}{\pi}\Dbf$. We denote the GFF by $\varphi$. Explicitly,
\[\varphi(\theta)=\sum_{n>0}\frac{1}{\sqrt{n}}\big(x_n\cos(n\theta)-y_n\sin(n\theta)\big)=\sum_{n\neq0}\varphi_ne^{\i n\theta},\]
where the $x_n,y_n$ are i.i.d.\ standard Gaussians, $\varphi_n=\frac{1}{2\sqrt{n}}(x_n+\i y_n)$, $\varphi_{-n}=\ol{\varphi_n}$ for $n>0$. Its covariance is\footnote{By abuse of notation, we write functions on $\Tb$ or $\Tb_+$ as functions of $\theta$.}
\[\Eb[\varphi(\theta)\varphi(\theta')]=-\log|e^{\i\theta}-e^{\i\theta'}|.\]
We equip $\Dc_0'(\Tb,\Rb)$ with the probability measure $\d\varphi$ and $\Dc'(\Tb,\Rb)=\Rb\oplus\Dc_0'(\Tb,\Rb)$ with the product measure $\d\wt{\varphi}=\d c\,\d\varphi$ where $\wt{\varphi}=c+\varphi$.
\item $M=\Tb_+$, $g=d\theta^2$, $\Delta=\frac{1}{\pi}\Dbf$ with Neumann boundary condition. We denote the GFF by $\varphi^\h$. Explicitly,
\[\varphi^\h(\theta)=\sum_{n>0}\sqrt{\frac{2}{n}}\,x_n^\h\cos(n\theta)=\sum_{n\neq0}\varphi_n^\h e^{\i n\theta},\]
where the $x_n^\h$ are i.i.d.\ standard Gaussians, $\varphi_n^\h=\frac{1}{\sqrt{2n}}x_n^\h$, $\varphi_{-n}^\h=\varphi_n^\h$ for $n>0$. Note that the coefficients differ from the previous case by $\sqrt{2}$, since $\vol\Tb=2\vol\Tb_+$. Its covariance is
\[\Eb[\varphi^\h(\theta)\varphi^\h(\theta')]=-2\log|e^{\i\theta}-e^{\i\theta'}|.\]
We equip $\Dc_0'(\Tb_+,\Rb)$ with the probability measure $\d\varphi^\h$ and $\Dc'(\Tb_+,\Rb)=\Rb\oplus\Dc_0'(\Tb_+,\Rb)$ with the product measure $\d\wt{\varphi}^\h=\d c^\h\,\d\varphi^\h$ where $\wt{\varphi}^\h=c^\h+\varphi^\h$.
\item $M=\Sigma$ is an extended surface, $g$ is a Neumann extendible conformal metric on $\Sigma$, $\Delta=\frac{1}{2\pi}\Delta_g$ with mixed boundary condition. We denote the GFF by $X_g$. Its covariance is
\[\Eb[X_g(x)X_g(x')]=2\pi\,G_g(x,x').\]
If $\paN\Sigma\neq\varnothing$, we have
\[X_g\law\frac{X_{g^{\#2}}+X_{g^{\#2}}\circ\tau}{\sqrt{2}}.\]
If $\paD\Sigma=\varnothing$, we equip $\Dc_0'(\Sigma,\Rb)$ with the probability measure $\d X_g$ and $\Dc'(\Sigma,\Rb)$ with the product measure $\d c\,\d X_g$. For $\rho\in C^\infty(\Sigma,\Rb)$ Neumann extendible, we have
\[X_{e^\rho g}\law X_g-\frac{1}{\vol_{e^\rho g}\Sigma}\int_\Sigma X_g\,\d v_{e^\rho g}.\]
If $\paD\Sigma\neq\varnothing$, we equip $\Dc'(\Sigma,\Rb)$ with the probability measure $\d X_g$. There is no zero mode, and the law of $X_g$ is conformally invariant.
\end{itemize}

In all of these cases, the GFF converges in $H^s$ for $s<0$ and is a.s.\ not in $H^0=L^2$. Note that here the zero mode (corresponding to $\ker\Delta$) is not $L^2$-normalized. This normalization will manifest itself as multiplicative constants in the formulas later.

Returning to the general theory, we state some standard results on GFFs that will be used in calculations.

\begin{thm}[\textbf{Cameron--Martin}]\label{CM}
For $h\in H_\Delta^{k/2}(M,\Rb)=H^{k/2}(M,\Rb)\cap\Dc_\Delta'(M,\Rb)$, the translated measure $\d(X_\Delta-h)$ is equivalent to $\d X_\Delta$ with Radon--Nikodym derivative given by
\[\d(X_\Delta-h)=\exp\left(-\frac{1}{2}\int_Mh\,\Delta h\,\d v_g+\int_MX_\Delta\,\Delta h\,\d v_g\right)\d X_\Delta.\]
\end{thm}

\blue{Here the term $\int_MX_\Delta\,\Delta h\,\d v_g$ is not to be interpreted pointwise but as follows. Using the notation above, for $h=\sum c_je_j$, $\int_MX_\Delta\,\Delta h\,\d v_g=\sum\sqrt{\lambda_j}\,c_ja_j\,e_j$. By assumption, $\sum\lambda_jc_j^2<\infty$, so this converges a.s.\ and in $L^2$. For a more general statement, we refer to \cite[Proposition 2.4.2]{bogachev1998gaussian}.}

\begin{cor}
For $h\in H^{k/2}(M,\Rb)$, the translated measure $\d\mu_\Delta(\cdot-h)$ is equivalent to $\d\mu_\Delta$ with Radon--Nikodym derivative given by
\[\d\mu_\Delta(X-h)=\exp\left(-\frac{1}{2}\int_Mh\,\Delta h\,\d v_g+\int_MX\,\Delta h\,\d v_g\right)\d\mu_\Delta(X).\]
\end{cor}

In probabilistic language, it can be stated in the following way:

\begin{thm}[\textbf{Girsanov transform}]\label{girsanov}
For $F\in L^1(\Dc_\Delta'(M,\Rb))$, $Y$ a centered real Gaussian variable, we have
\[\Eb[F(X_\Delta)\,e^Y]=e^{\frac{1}{2}\Eb[Y^2]}\,\Eb[F(X_\Delta+\Eb[YX_\Delta])].\]
\end{thm}
\begin{proof}
Take $h=\Eb[YX_\Delta]$ in the Cameron--Martin theorem.
\end{proof}

We also need an imaginary version where $Y$ is replaced by $\i Y$. For this to make sense, we introduce the following function class. Let $\Ec(\Dc_\Delta'(M,\Rb))$ be the space of functions $F:\Dc_\Delta'(M,\Rb)\to\Cb$ of the form $F(X)=P(\la X,h_1\ra,\ldots,\la X,h_m\ra)$ where $m\in\Nb$, $P$ is a complex polynomial, $h_1,\ldots,h_m\in\Dc(M,\Rb)$. For $h\in H_\Delta^{k/2}(M,\Rb)$, $F(\cdot+\i h)$ makes sense for $F\in\Ec(\Dc_\Delta'(M,\Rb))$.

\begin{thm}[\textbf{Imaginary Girsanov transform}]\label{imaginary-girsanov}
For $F\in\Ec(\Dc_\Delta'(M,\Rb))$, $Y$ a centered real Gaussian variable, we have
\[\Eb[F(X_\Delta)\,e^{\i Y}]=e^{-\frac{1}{2}\Eb[Y^2]}\,\Eb[F(X_\Delta+\i\,\Eb[YX_\Delta])].\]
\end{thm}

The next theorem allows us to compare the GFFs associated to different operators, which is essentially a special case of the Feldman--Hájek theorem:

\begin{thm}
Let $\Delta$, $\wt{\Delta}$ be two strictly positive elliptic operators on $M$ such that $\wt{\Delta}\Delta^{-1}-I$ is smoothing (i.e., it extends to a continuous map $\Dc'(M,\Rb)\to\Dc(M,\Rb)$). Then $\d X_{\wt{\Delta}}$ is equivalent to $\d X_\Delta$ with Radon--Nikodym derivative given by
\[\d X_{\wt{\Delta}}=\sqrt{\det\nolimits_\Fr(\wt{\Delta}\Delta^{-1})}\,e^{-\frac{1}{2}\int_M X_\Delta\,(\wt{\Delta}-\Delta)X_\Delta\,\d v_g}\,\d X_\Delta,\]
where $\det_\Fr$ denotes the Fredholm determinant.
\end{thm}
\begin{proof}
Let $\Pi_n$ be the projection from $L^2(M,\Rb)$ to the finite subspace $\mathcal H_n$ spanned by $e_1,\ldots,e_n$, the first $n$ orthonormal basis of $\Delta$ . Define $\Delta_n,\wt \Delta_n$ as $\Pi_n\circ\Delta\circ \Pi_n$ and $\Pi_n\circ\wt \Delta\circ \Pi_n$. Then, by the knowledge of the Gaussian distribution,
    \begin{equation*}
        \frac{\d X_{\wt{\Delta}_n}}{\d X_{{\Delta}_n}}=\sqrt{\det\nolimits(\wt{\Delta}_n\Delta^{-1}_n)}\,e^{-\frac{1}{2}\int_M X_{\Delta_n}\,(\wt{\Delta}_n-\Delta_n)X_{\Delta_n}\,\d v_g}.
    \end{equation*}
For a suitable test function $h$ constructed from $e_0,\ldots,e_N$ where $N<n$, \textcolor{black}{$\int_M X_{\Delta}h\  \d v_g=\int_M X_{\Delta_n}h\  \d v_g$ according to the construction of GFF. Therefore, we can study the distribution of $\int_M X_{\wt{\Delta}_n}h\  \d v_g$ using $\d X_{{\Delta}}$, and}
\begin{equation*}
\begin{split}
&\int e^{\i\int_M X_{\wt{\Delta}}h\  \d v_g}\d X_{\wt{\Delta}}\\
    =&\int e^{\i\int_M X_{\wt{\Delta}}h\  \d v_g}\sqrt{\det\nolimits(\wt{\Delta}_n\Delta^{-1}_n)}\,e^{-\frac{1}{2}\int_M X_{\Delta}\,(\wt{\Delta}_n-\Delta_n)X_{\Delta}\,\d v_g}\d X_{{\Delta}}\\
    =&\int e^{\i\int_M X_{\wt{\Delta}_n}h\  \d v_g}\frac{\d X_{\wt{\Delta}_n}}{\d X_{{\Delta}_n}}\d X_{{\Delta}_n}= e^{-\frac12\int_M h\  (\wt{\Delta}_n|_{\mathcal H_n})^{-1}\Pi_n h\ \d v_g},
\end{split}
\end{equation*}
where on the last line,  $\int_M h\  (\wt{\Delta}_n|_{\mathcal H_n})^{-1}\Pi_n h\ \d v_g$ converges to $\int_M h(x)  G_{\wt{\Delta}}(x,x')h(x') \d v_g$. Next, we claim that the first line converges to 
\begin{equation*}
    \int e^{\i\int_M X_{\wt{\Delta}}h\  \d v_g}\sqrt{\det\nolimits_\Fr(\wt{\Delta}\Delta^{-1})}\,e^{-\frac{1}{2}\int_M X_\Delta\,(\wt{\Delta}-\Delta)X_\Delta\,\d v_g}\ \d X_{{\Delta}}.
\end{equation*}
Since  $\wt{\Delta}-\Delta$ is smoothing, this follows from the discussion in \cite[Lemma 5.3]{Segal}.
\end{proof}
We will need the following generalization to operators with nonzero kernels. For an operator $\Delta$, we write $\Delta_0=\Delta+\Pi_\Delta$ where $\Pi_\Delta$ is the orthogonal projection on $L^2(M,\Rb)$ onto $\ker\Delta$. Note that $\Delta_0=\Delta$ if and only if $\ker\Delta=0$.

\begin{cor}\label{FH-var}
Let $\Delta$, $\wt{\Delta}$ be two nonnegative elliptic operators on $M$ such that $\wt{\Delta}-\Delta$ is smoothing. Then $\d\mu_{\wt{\Delta}}$ is equivalent to $\d\mu_\Delta$ with Radon--Nikodym derivative given by
\[\d\mu_{\wt{\Delta}}(X)=(2\pi)^{\frac{\dim\ker\wt{\Delta}-\dim\ker\Delta}{2}}\sqrt{\det\nolimits_\Fr(\wt{\Delta}_0\Delta_0^{-1})}\,e^{-\frac{1}{2}\int_M X\,(\wt{\Delta}-\Delta)X\,\d v_g}\,\d\mu_\Delta(X).\]
\end{cor}
\begin{proof}
By definition,
\[\d X_{\Delta_0}=(2\pi)^{-\frac{d}{2}}e^{-\frac{1}{2}\sum_{j=1}^dc_j^2}\,\d c^d\,\d X_\Delta=(2\pi)^{-\frac{d}{2}}e^{-\frac{1}{2}\int_M X\,\blue{\Pi_\Delta} X\,\d v_g}\,\d\mu_\Delta(X),\]
where $d=\dim\ker\Delta$, $X=\sum_{j=1}^dc_je_j+X_\Delta$. Likewise for $\wt{\Delta}$. Clearly $\wt{\Delta}_0\Delta_0^{-1}-I=(\wt{\Delta}-\Delta+\Pi_{\wt{\Delta}}-\Pi_\Delta)\Delta_0^{-1}$ is smoothing, hence Hilbert--Schmidt. Thus the formula follows from the Feldman--Hájek theorem for $\Delta_0$, $\wt{\Delta}_0$.
\end{proof}

\subsection{The spaces $\Dc'(\Tb,\TR)$ and $\Dc'(\Tb_+,\TR)$}\label{section:measure-space1}

We would like to define the spaces of $\TR$-valued distributions on $\Tb$ and $\Tb_+$. This is a toy example for the next subsection, and they also play a role in the gluing of amplitudes. Taking inspiration from the decomposition (recall that the bar denotes the projection $\Rb\to\TR$)
\[\setlength\arraycolsep{1pt}
\begin{array}[t]{ccccccc}
\TR & \times & C_0(\Tb,\Rb) & \times & \Zb & \xrightarrow{\sim} & C(\Tb,\TR)\\
c&&\varphi&&k&\mapsto&c+\ol{\varphi}+\ol{kR\theta}
\end{array}\]
where $C_0(\Tb,\Rb)=\{\varphi\in C(\Tb,\Rb):\int_\Tb\varphi\,\d\theta=0\}=C(\Tb,\Rb)\cap\Dc_0'(\Tb,\Rb)$, we define
\[\Dc'(\Tb,\TR)=\TR\times\Dc_0'(\Tb,\Rb)\times\Zb.\]
We equip it with the product measure $\d\wt{\varphi}^k=\d c\,\d\varphi\,\d k$ where $\wt{\varphi}^k=(c,\varphi,k)$, $\d c$ is the Lebesgue measure on $\TR$, $\d k$ is the counting measure on $\Zb$. By the decomposition above, we view $C(\Tb,\TR)$ naturally as a subspace of $\Dc'(\Tb,\TR)$.

Similarly, we define
\[\Dc'(\Tb_+,\TR)=\TR\times\Dc_0'(\Tb_+,\Rb),\]
and equip it with the product measure $\d\wt{\varphi}^\h=\d c^\h\,\d\varphi^\h$ where $\wt{\varphi}^\h=(c^\h,\varphi^\h)$ ($\h$ for ``half''). By doubling, we view $C(\Tb_+,\TR)$ naturally as a subspace of $\Dc'(\Tb_+,\TR)$. Note that there is no discrete part $\Zb$ on $\Tb_+$ because $\ol{kR\theta}$ is even if and only if $k=0$. This reflects the fact that $\Tb_+$ is topologically trivial.

Finally, for an extended surface $\Sigma$, we define $\Dc'(\paS,\TR)=\Dc'(\Tb,\TR)^{\bD}\times\Dc'(\Tb_+,\TR)^{\bMD}$ \blue{and equip it with the product measure}. For 
\[\bt^\kbf=(\wt{\varphi}_1^{k_1},\ldots,\wt{\varphi}_{\bD}^{k_{\bD}},\wt{\varphi}_1^\h,\ldots,\wt{\varphi}_{\bMD}^\h)\in\Dc'(\paS,\TR),\]
we have a decomposition $\bt^\kbf=\cbf^\kbf+\varphibf$, where 
\begin{align*}
&\cbf^\kbf=(c_1^{k_1},\ldots,c_{\bD}^{k_{\bD}},c_1^\h,\ldots,c_{\bMD}^\h)\in(\TR\times\Zb)^{\bD}\times\TR^{\bMD},\\
&\varphibf=(\varphi_1,\ldots,\varphi_{\bD},\varphi_1^\h,\ldots,\varphi_{\bMD}^\h)\in\Dc_0'(\Tb,\Rb)^{\bD}\times\Dc_0'(\Tb_+,\Rb)^{\bMD}.
\end{align*}

The Hilbert spaces $\Hc=L^2(\Dc'(\Tb,\TR))$ and $\Hc_+=L^2(\Dc'(\Tb_+,\TR))$ are the \textit{state spaces} of BCILT.

\subsection{The space $\Dc'^\N(\Sigma\setminus\zbf,\TR)$ ($\paD\Sigma=\varnothing$)}\label{section:measure-space2}

Let $\Sigma$ be an extended surface with $\paD\Sigma=\varnothing$ and $g$ a Neumann extendible conformal metric on $\Sigma$. We would like to define the space $\Dc'^\N(\Sigma\setminus\zbf,\TR)$ of $\TR$-valued distributions $\phi$ on $\Sigma\setminus\zbf$ satisfying the Neumann boundary condition $\pa_\nu\phi|_{\pa\Sigma}=0$ (with $\pa\Sigma=\pa_\N\Sigma=\paN\Sigma$). Let (recall \cref{admissible})
\[C^\N(\Sigma\setminus\zbf,\TR)=\{\phi\in C^\infty(\Sigma\setminus\zbf,\TR):\tfrac{1}{2\pi R}d\phi\text{  admissible}\}.\]
Here we trivialize the tangent bundle of $\TR$ using the Lie group structure, so that $d\phi$ is a closed $1$-form on $\Sigma\setminus\zbf$. Note that the admissibility of $d\phi$ implies $\pa_\nu\phi|_{\pa\Sigma}=i_\nu d\phi|_{\pa\Sigma}=0$. We equip this space with the $C_\loc^\infty$ topology on $\Sigma\setminus\zbf$. We define
\[\Dc'^\N(\Sigma\setminus\zbf,\TR)=\Dc_0'(\Sigma,\Rb)\times C^\N(\Sigma\setminus\zbf,\TR)/\!\sim\]
where $(Y_1,\psi_1)\sim(Y_2,\psi_2)$ if $Y_1-Y_2\in C(\Sigma,\Rb)$ and $\ol{Y_1-Y_2}=\psi_2-\psi_1$. This is the quotient by the closed subspace $\{(Y,-\ol{Y}):Y\in C^\N(\Sigma,\Rb),\int_\Sigma Y\,\d v_g=0\}$, so it is Hausdorff.

For $x_0\in\Sigma\setminus\zbf$ and $\omega$ a closed $1$-form on $\Sigma\setminus\zbf$, we denote by $I_{x_0}(\omega)$ the multi-valued function on $\Sigma\setminus\zbf$ defined by $x\mapsto\int_\gamma\omega$ where $\gamma$ is any path on $\Sigma\setminus\zbf$ from $x_0$ to $x$. If $\int_c\omega\in\Zb$ for any circle $c$ on $\Sigma\setminus\zbf$, then $I_{x_0}(2\pi R\omega)$ projects to a well-defined map $\ol{I_{x_0}(2\pi R\omega)}:\Sigma\setminus\zbf\to\TR$ with $\frac{1}{2\pi R}d\ol{I_{x_0}(2\pi R\omega)}=\omega$. This defines a one-to-one correspondence between the homotopy classes of maps $\Sigma\setminus\zbf\to\TR$ and the cohomology classes in $H^1(\Sigma\setminus\zbf)$, i.e., a bijection $\pi_0(C(\Sigma\setminus\zbf,\TR))\xleftrightarrow{\sim}H^1(\Sigma\setminus\zbf)$, which motivates the following decomposition of $\Dc'^\N(\Sigma\setminus\zbf,\TR)$:

\begin{prop}\label{measure1}
Fix $x_0\in\Sigma\setminus\zbf$ and choose an admissible representative $\omega$ in each cohomology class $[\omega]\in H^1(\Sigma\setminus\zbf)$. We have a homeomorphism
\[\setlength\arraycolsep{1pt}
\begin{array}[t]{ccccccc}
\TR & \times & \Dc_0'(\Sigma,\Rb) & \times & H^1(\Sigma\setminus\zbf) & \xrightarrow{\sim} & \Dc'^\N(\Sigma\setminus\zbf,\TR)\\
c&&X_g&&[\omega]&\mapsto&[(X_g,c+\ol{I_{x_0}(2\pi R\omega)})]
\end{array}\]
On the left-hand side, consider the measure
\[\sqrt{\frac{\vol_g\Sigma}{\det'\Delta_{g}}}\exp\left(-\pi R^2\int_\Sigma^\reg|\omega|_g^2\,\d v_g-R\int_{\Sigma}X_g\,d^*\omega\,\d v_g\right)\d c\,\d X_g\,\d[\omega],\]
where \blue{$\d X_g$ is the measure induced by the GFF $X_g$ (recall \cref{GFF})}, $\d c$ is the Lebesgue measure on $\TR$, $\d[\omega]$ is the counting measure on $H^1(\Sigma\setminus\zbf)$. Pushing forward this measure defines a measure on $\Dc'^\N(\Sigma\setminus\zbf,\TR)$ that does not depend on the choices of $x_0$ or the representatives $\omega$. We denote it by $\d\phi_g$ where $\phi_g=[(X_g,c+\ol{I_{x_0}(2\pi R\omega)})]$.
\end{prop}

Here the notation $X_g$ alludes to the GFF and emphasizes the dependence on the metric $g$. From the measure-theoretic viewpoint, it denotes a deterministic element in the space $\Dc'_0(\Sigma,\Rb)$, the measure on which depends on $g$. From the probabilistic viewpoint, it can be thought of as a random distribution sampled from the GFF with respect to $g$.

\begin{proof}
The inverse map is $[(Y,\psi)]\mapsto\textstyle(\psi(x_0)+\ol{\frac{1}{\vol_g\Sigma}\int_\Sigma f\,\d v_g},Y+f-\frac{1}{\vol_g\Sigma}\int_\Sigma f\,\d v_g,[\omega])$, where $\omega$ is the chosen representative of $[\frac{1}{2\pi R}d\psi]\in H^1(\Sigma\setminus\zbf)$, $f=I_{x_0}\big(\frac{1}{2\pi R}d\psi-\omega\big)\in\Dc(\Sigma,\Rb)$. Continuity is trivial.

Changing $x_0$ amounts to a translation in $c$, i.e., $[(X_g,c+I_{x_0'}(2\pi R\omega)]=[(X_g,(c-\ol{\int_{x_0}^{x_0'}2\pi R\omega})+I_{x_0}(2\pi R\omega))]$. Suppose we change a representative $\omega$ to $\omega+\frac{1}{2\pi R}df$. By admissibility, $f\in\Dc(\Sigma,\Rb)$. We have $[(X_g,c+I_{x_0}(2\pi R\omega+df))]=[(X_g+(f-\frac{1}{\vol_g\Sigma}\int_\Sigma f\,\d v_g),(c+\frac{1}{\vol_g\Sigma}\int_\Sigma f\,\d v_g-f(x_0))+I_{x_0}(2\pi R\omega))]$, so this amounts to a translation in $c$ and the translation by $h=f-\frac{1}{\vol_g\Sigma}\int_\Sigma f\,\d v_g$ on the GFF $X_g$. By the Cameron--Martin theorem,
\[\d(X_g+h)=\exp\left(-\frac{1}{4\pi}\int_\Sigma|dh|_g^2\,\d v_g-\frac{1}{2\pi}\int_\Sigma X_g\,\Delta_gh\,\d v_g\right)\d X_g.\]
Now
\begin{align*}
\left(-\pi R^2\int_\Sigma^\reg|\omega+\tfrac{1}{2\pi R}df|_g^2\,\d v_g-R\int_{\Sigma}X_g\,d^*(\omega+\tfrac{1}{2\pi R}df)\,\d v_g\right)-\left(-\pi R^2\int_\Sigma^\reg|\omega|_g^2\,\d v_g-R\int_{\Sigma}(X_g+h)\,d^*\omega\,\d v_g\right)\\
=-\frac{1}{4\pi}\int_\Sigma|df|_g^2\,\d v_g-\frac{1}{2\pi}\int_\Sigma X_g\,\Delta_gf\,\d v_g,
\end{align*}
so the factors cancel out, as expected.
\end{proof}

\begin{rem*}
While it is necessary to put the Lebesgue measure on $\TR$ for translation invariance, one can put any measure on $H^1(\Sigma\setminus\zbf)$. In particular, putting a Dirac measure on $H^1(\Sigma\setminus\zbf)$ amounts to restricting the path integral to fields in a given homotopy class. The resulting theories are well-defined but generally not diffeomorphism invariant in the sense of Quantum Field Theory.
\end{rem*}

For $\mbf\in\Zb^\s$, let $C_\mbf^\N(\Sigma\setminus\zbf,\TR)$ be the subspace of maps $\phi\in C^\N(\Sigma\setminus\zbf,\TR)$ with winding number $m_i$ around $z_i$, i.e., $\phi_*[\pc_i]=m_i$, where $\phi_*:H_1(\Sigma\setminus\zbf)\to H_1(\TR)\cong\Zb$ is the homomorphism on the first homology induced by $\phi$, and $\pc_i$ is the boundary of a small disk centered at $z_i$ (as in \cref{cohomologies}). We define
\[\Dc'_\mbf(\Sigma\setminus\zbf,\TR)=\Dc_0'(\Sigma,\Rb)\times C_\mbf^\N(\Sigma\setminus\zbf,\TR)/\!\sim\,\subset\Dc'(\Sigma\setminus\zbf,\TR).\]
The homeomorphism above is compatible with the subscript $_\mbf$ in the sense that it restricts to a homeomorphism $\TR\times\Dc_0'(\Sigma,\Rb)\times H_\mbf^1(\Sigma\setminus\zbf)\xrightarrow{\sim}\Dc_\mbf'^\N(\Sigma\setminus\zbf,\TR)$, and the measure on $\Dc'^\N(\Sigma\setminus\zbf,\TR)$ restricts to a measure on $\Dc_\mbf'^\N(\Sigma\setminus\zbf,\TR)$, denoted in the same way.

\begin{prop}[\textbf{Conformal covariance}]\label{measure-conf1}
For $\rho\in C^\infty(\Sigma,\Rb)$ Neumann extendible, we have on $\Dc_\mbf'^\N(\Sigma\setminus\zbf,\TR)$
\[\d\phi_{e^\rho g}=\exp\left(\frac{1}{96\pi}\int_\Sigma(|d\rho|_g^2+2K_g\rho)\,\d v_g-\frac{R^2}{4}\sum_{i=1}^\s m_i^2\rho(z_i)\right)\d\phi_g.\]
\end{prop}
\begin{proof}
Fix $x_0,\omega$ as in \cref{measure1}. For a Neumann extendible metric $g$ on $\Sigma$, let $\Phi_g:\TR\times\Dc'_{0,g}(\Sigma,\Rb)\times H^1(\Sigma\setminus\zbf)\xrightarrow{\sim}\Dc'^\N(\Sigma\setminus\zbf,\TR)$ be the associated homeomorphism with respect to $g$, where the subscript $_g$ in $\Dc'_{0,g}(\Sigma,\Rb)$ denotes dependence on $g$. Then $\Phi_{e^\rho g}^{-1}\circ\Phi_g:(c,X_g,[\omega])\mapsto(c+\ol{\frac{1}{\vol_{e^\rho g}\Sigma}\int_\Sigma X_g\,\d v_{e^\rho g}},X_g-\frac{1}{\vol_{e^\rho g}\Sigma}\int_\Sigma X_g\,\d v_{e^\rho g},[\omega])$. Recall that $X_{e^\rho g}\law X_g-\frac{1}{\vol_{e^\rho g}\Sigma}\int_\Sigma X_g\,\d v_{e^\rho g}$. Thus $\d c\,\d X_{e^\rho g}\,\d[\omega]=\d c\,\d X_g\,\d[\omega]$. The conformal covariance of $\int_\Sigma^\reg|\omega|_g^2\,\d v_g$ gives the second factor. By the Polyakov formula (\cite[Proposition 5.10]{BLCFT}),
\[\frac{\det'\Delta_{e^\rho g}}{\vol_{e^\rho g}\Sigma}=\frac{\det'\Delta_{g}}{\vol_g\Sigma}\,\exp\left(-\frac{1}{48\pi}\int_\Sigma\big(|d\rho|_g^2+2K_g\rho\big)\,\d v_g\right),\]
which gives the first factor.
\end{proof}

\begin{rem*}
Without the assumption that $g$ is Neumann extendible (so that $\pa\Sigma=\paN\Sigma$ is geodesic), one should multiply the measure on $\Dc'^\N(\Sigma\setminus\zbf,\TR)$ by $e^{\frac{1}{8\pi}\int_{\pa\Sigma}k_g\,\d\ell_g}$ to have the correct Polyakov formula.
\end{rem*}

\subsection{The space $\Dc'^{\M,\bt^\kbf}(\Sigma\setminus\zbf,\TR)$ ($\paD\Sigma\neq\varnothing$)}\label{section:measure-space3}

Let $\Sigma$ be an extended surface with $\paD\Sigma\neq\varnothing$ and $g$ a Neumann extendible conformal metric on $\Sigma$. For $\bt^\kbf\in\Dc'(\paS,\TR)$ \blue{(recall \cref{section:measure-space1})}, we would like to define the space $\Dc'^{\M,\bt^\kbf}(\Sigma\setminus\zbf,\TR)$ of $\TR$-valued distributions $\phi$ on $\Sigma\setminus\zbf$ satisfying the mixed boundary condition $\pa_\nu\phi|_{\paN\Sigma}=0$, $\zetabf^*\phi|_{\paD\Sigma}=\bt^\kbf$. The definition is more intricate due to two things:

\begin{itemize}
    \item It is not obvious how to define the boundary value of a distribution. Moreover, the boundary value $\bt^\kbf$ here is highly singular (not even in $L_\loc^1$).
    \item The solution to the previous problem in the real theory was to define the field as the Dirichlet GFF plus the harmonic extension of the boundary value. However, in our case, the field takes values in $\TR$, and harmonic extensions (in the sense of harmonic maps) are not unique. This nonuniqueness manifests itself as the cohomology $H_\D^1(\Sigma)$ (recall \cref{cohomologies}) below.
\end{itemize}

Our strategy is to separate the singular part and the topological part in the boundary value $\bt^\kbf$. First let
\[C^{\M,\const}(\Sigma\setminus\zbf,\TR)=\{\phi\in C^\infty(\Sigma\setminus\zbf,\TR):\tfrac{1}{2\pi R}d\phi\text{ admissible}\}.\]
This is formally the same as $C^\N(\Sigma\setminus\zbf,\TR)$ for the case $\paD\Sigma=\varnothing$. Here the admissibility of $d\phi$ implies:
\begin{itemize}
\item $\pa_\nu\phi|_{\paN\Sigma}=i_\nu d\phi|_{\paN\Sigma}=0$, as before.
\item In a neighborhood of each $c_i^\D$, $(\zeta_i^\D)^*\phi-\ol{k_i R\theta}$ is constant, where $k_i=\phi_*[c_i^\D]$ is the winding number of $\phi$ on $c_i^\D$.
\item In a neighborhood of each $c_i^\MD$, $(\zeta_i^\MD)^*\phi$ is constant.
\end{itemize}
We equip this space with the $C_\loc^\infty$ topology on $\Sigma\setminus\zbf$. Let
\[C^\M(\Sigma\setminus\zbf,\TR)=\{\phi\in C^\infty(\Sigma\setminus(\zbf\cup\paD\Sigma),\TR):\zetabf^*\phi|_{\paD\Sigma}\text{ exists},\phi-\ol{P_{\Sigma}\Pi\zetabf^*\phi|_{\paD\Sigma}}\in C^{\M,\const}(\Sigma\setminus\zbf,\TR)\},\]
where the trace $\zetabf^*\phi|_{\paD\Sigma}$ is in the weak sense (see \cref{Poisson}), $\Pi$ is the projection $\bt^\kbf=\cbf^\kbf+\varphibf\mapsto\varphibf$ on $\Dc'(\paS,\TR)$. Note that the maps in $C^\M(\Sigma\setminus\zbf,\TR)$ are generally singular at $\paD\Sigma$. By definition, we have a bijection $\Dc_0'(\paS,\Rb)\times C^{\M,\const}(\Sigma\setminus\zbf,\TR)\xrightarrow{\sim} C^\M(\Sigma\setminus\zbf,\TR),(\varphibf,\psi)\mapsto\ol{P_\Sigma\varphibf}+\psi$. We equip $C^\M(\Sigma\setminus\zbf,\TR)$ with the pushforward of the product topology, so that the trace map $C^\M(\Sigma\setminus\zbf,\TR)\to\Dc'(\paS,\TR),\phi\mapsto\zetabf^*\phi|_{\paD\Sigma}$ is a continuous surjection. For $\bt^\kbf\in\Dc'(\paS,\TR)$, let
\begin{align*}
C^{\M,\bt^\kbf}(\Sigma\setminus\zbf,\TR)&=\{\phi\in C^\M(\Sigma\setminus\zbf,\TR):\zetabf^*\phi|_{\paD\Sigma}=\bt^\kbf\}\\
&=\{\ol{P_\Sigma\varphibf}+\psi:\psi\in C^{\M,\const}(\Sigma\setminus\zbf,\TR),\zetabf^*\psi|_{\paD\Sigma}=\cbf^\kbf\}\\
&=\ol{P_\Sigma\varphibf}+C^{\M,\cbf^\kbf}(\Sigma\setminus\zbf,\TR),
\end{align*}
which is a closed subspace of $C^\M(\Sigma\setminus\zbf,\TR)$. Note that $C^{\M,\const}(\Sigma\setminus\zbf,\TR)=\bigsqcup_{\cbf^\kbf}C^{\M,\cbf^\kbf}(\Sigma\setminus\zbf,\TR)$, hence the notation. Finally, we define
\[\Dc'^\M(\Sigma\setminus\zbf,\TR)=\Dc'(\Sigma,\Rb)\times C^\M(\Sigma\setminus\zbf,\TR)/\!\sim\]
where $(Y_1,\psi_1)\sim(Y_2,\psi_2)$ if $Y_1-Y_2\in C(\Sigma,\Rb)$ with $(Y_1-Y_2)|_{\paD\Sigma}=0$ and $\ol{Y_1-Y_2}=\psi_2-\psi_1$. This is the quotient by the closed subspace $\{(Y,-\ol{Y}):Y\in C^{\M,\mathbf{0}}(\Sigma,\Rb)\}$, so it is Hausdorff. For $\bt^\kbf\in\Dc'(\paS,\TR)$, let
\begin{align*}
\Dc'^{\M,\bt^\kbf}(\Sigma\setminus\zbf,\TR)&=\Dc'(\Sigma,\Rb)\times C^{\M,\bt^\kbf}(\Sigma\setminus\zbf,\TR)/\!\sim\\
&=\{[(Y,\psi)]\in\Dc'^\M(\Sigma\setminus\zbf,\TR):\zetabf^*\psi|_{\paD\Sigma}=\bt^\kbf\},
\end{align*}
which is a closed subspace of $\Dc'^\M(\Sigma\setminus\zbf,\TR)$.

Unlike the previous case, $H^1(\Sigma\setminus\zbf)$ is insufficient to determine the homotopy type of a map in $C^{\M,\bt^\kbf}(\Sigma\setminus\zbf,\TR)$, i.e., two maps in $C^{\M,\bt^\kbf}(\Sigma\setminus\zbf,\TR)$ that are homotopic in $C(\Sigma\setminus\zbf,\TR)$ are not necessarily homotopic in $C^{\M,\bt^\kbf}(\Sigma\setminus\zbf,\TR)$. Indeed, the kernel of the natural homomorphism $\pi_0(C^{\M,\mathbf{0}}(\Sigma\setminus\zbf,\TR))\to\pi_0(C(\Sigma\setminus\zbf,\TR))$ is isomorphic to $H_\D^1(\Sigma\setminus\zbf)$, so there is a (noncanonical) bijection $\pi_0(C^{\M,\bt^\kbf}(\Sigma\setminus\zbf,\TR))\cong H_\kbf^1(\Sigma\setminus\zbf)\times H_\D^1(\Sigma)$, which motivates the following decomposition of $\Dc'^{\M,\bt^\kbf}(\Sigma\setminus\zbf,\TR)$:

\begin{prop}\label{measure2}
Let $\bt^\kbf\in\Dc'(\paS,\TR)$. We make the following choices:
\begin{itemize}
\item Fix $x_0\in\Sigma\setminus\zbf$.
\item Take $c_0=\cbf^\kbf(\zetabf^{-1}(x_0))\in\TR$ if $x_0\in\paD\Sigma$, otherwise take any $c_0\in\TR$.
\item Choose an admissible representative $\omega$ in each cohomology class $[\omega]\in H_\kbf^1(\Sigma\setminus\zbf)$ such that $\zetabf^*I_{x_0}(2\pi R\omega)|_{\paD\Sigma}+c_0=\cbf^\kbf$.
\item Choose an admissible primitive $f$ for each cohomology class $[df]\in H_\D^1(\Sigma)$.
\end{itemize}
We have a homeomorphism
\[\setlength\arraycolsep{1pt}
\begin{array}[t]{ccccccc}
\Dc'(\Sigma,\Rb) & \times & H_\D^1(\Sigma) & \times & H^1_\kbf(\Sigma\setminus\zbf) & \to & \Dc'^{\M,\bt^\kbf}(\Sigma\setminus\zbf,\TR)\\
X_g&&[df]&&[\omega]&\mapsto&[(X_g,\ol{P_\Sigma\varphibf}+\ol{2\pi Rf}+I_{x_0}(2\pi R\omega)+c_0)]
\end{array}\]
On the left-hand side, consider the measure
\[\frac{e^{\frac{1}{8\pi}\int_{\paD\Sigma}k_g\,\d\ell_g}}{\sqrt{\det'\Delta_{g}}}\,e^{-\frac{1}{4\pi}\la\varphibf,(\Dbf_\Sigma-\Dbf)\varphibf\ra}\exp\left(-\pi R^2\int_\Sigma^\reg|\omega+df|_g^2\,\d v_g-R\int_\Sigma(X_g+P_\Sigma\varphibf)\,d^*(\omega+df)\,\d v_g\right)\d X_g\,\d[df]\,\d[\omega],\]
where $\d[df]$, $\d[\omega]$ are the counting measures on $H_\D^1(\Sigma)$, $H^1(\Sigma\setminus\zbf)$, respectively. Pushing forward this measure defines a measure on $\Dc'^{\M,\bt^\kbf}(\Sigma\setminus\zbf,\TR)$ that does not depend on the choices of $x_0,c_0,\omega,f$. We denote it by $\d\phi_g$ where $\phi_g=[(X_g,\ol{P_\Sigma\varphibf}+\ol{2\pi Rf}+I_{x_0}(2\pi R\omega)+c_0)]$.
\end{prop}
\begin{proof}
The inverse map is $[(Y,\ol{P_\Sigma\varphibf}+\psi)]\mapsto([Y+h],[df],[\omega])$, where $\omega$ is the chosen representative of $[\frac{1}{2\pi R}d\psi]\in H_\kbf^1(\Sigma\setminus\zbf)$, $f$ is the chosen primitive for $[\frac{1}{2\pi R}d\psi-\omega]\in H_\D^1(\Sigma)$, $h=I_{x_0}\big(\frac{1}{2\pi R}d\psi-\omega\big)-c_0-f\in\Dc(\Sigma,\Rb)$. Continuity is trivial.

The rest of the proof is similar to \cref{measure1}. We record the formulas for changing the choices:
\begin{itemize}
\item Changing $x_0$ to $x_0'$ amounts to changing $c_0$ to $c_0+\int_{x_0}^{x_0'}\omega$.
\item If $x_0\notin\paD\Sigma$, changing $c_0$ to $c_0'$ forces changing $\omega$ to $\omega+dh$ with $h\in\Dc(\Sigma,\Rb)$ admissible, $h(x_0)=2\pi R(c_0-c_0')$, which amounts to the translation $X_g\mapsto X_g+2\pi Rh$.
\item Changing $\omega$ to $\omega+dh$ with $h$ admissible, $h(x_0)=0$ amounts to the translation $[df]\mapsto[df]+[dh]$ on $H_\D^1(\Sigma)$.
\item Changing $f$ to $f+h$ with $h\in\Dc(\Sigma,\Rb)$ admissible amounts to the translation $X_g\mapsto X_g+2\pi Rh$.\qedhere
\end{itemize}
\end{proof}

\begin{rem*}
Similarly to the remark after \cref{measure1}, while it is necessary to to put the Lebesgue measure on $\TR$ and the counting measure on $H_\D^1(\Sigma)$ for translation invariance, one can put any measure on $H^1(\Sigma\setminus\zbf)$.
\end{rem*}

For $\mbf\in\Zb^\s$, we define the corresponding subspaces with the subscript $_\mbf$ as before, so that the homeomorphism above restricts to a homeomorphism $\Dc'(\Sigma,\Rb)\times H_\D^1(\Sigma)\times H_{\kbf,\mbf}^1(\Sigma\setminus\zbf)\xrightarrow{\sim}\Dc_{\mbf}'^{\M,\bt^\kbf}(\Sigma\setminus\zbf,\TR)$, and the measure on $\Dc'^{\M,\bt^\kbf}(\Sigma\setminus\zbf,\TR)$ restricts to a measure on $\Dc_{\mbf}'^{\M,\bt^\kbf}(\Sigma\setminus\zbf,\TR)$, denoted in the same way.

\begin{prop}[\textbf{Conformal covariance}]\label{measure-conf2}
For $\rho\in C^\infty(\Sigma,\Rb)$ Neumann extendible with $\rho|_{\paD\Sigma}=0$, we have on $\Dc_{\mbf}'^{\M,\bt^\kbf}(\Sigma\setminus\zbf,\TR)$
\[\d\phi_{e^\rho g}=\exp\left(\frac{1}{96\pi}\int_\Sigma(|d\rho|_g^2+2K_g\rho)\,\d v_g-\frac{R^2}{4}\sum_{i=1}^\s m_i^2\rho(z_i)\right)\d\phi_g.\]
\end{prop}
\begin{proof}
This is similar to but simpler than \cref{measure-conf1}, since $X_{e^\rho g}\law X_g$ in this case.
\end{proof}

\begin{rem*}
Our definition is equivalent to that in \cite{CILT} in the following sense. Let $\Sigma$ be an extended surface with Dirichlet boundary only.
\begin{itemize}
    \item The topological part in \cite{CILT} is $H_{\mbf,\kbf}^1(\Sigma\setminus\zbf,\pa\Sigma)$, which is in bijection with our $H_{\mbf,\kbf}^1(\Sigma\setminus\zbf)\times H_\D^1(\Sigma)$ (recall that $H_\D^1(\Sigma)\cong H_\D^1(\Sigma\setminus\zbf)$ is the kernel of the natural homomorphism $H^1(\Sigma\setminus\zbf,\pa\Sigma)\to H^1(\Sigma\setminus\zbf)$). This bijection is not canonical. 
    \item In \cite{CILT}, the zero mode of the boundary value $\bt^\kbf$ is in $\Rb$, and one has to check that the resulting path integral is invariant under translation by $2\pi R\Zb$. However, since we define the boundary value $\bt^\kbf\in\Dc'(\paS,\TR)$ to have zero mode in $\TR=\Rb/2\pi R\Zb$, this invariance is automatic for us.
\end{itemize}


\end{rem*}

\section{Gluing for the free field}\label{gluing-free}

In this section, we prove Segal's gluing axiom for the free field, which is really the crux of the matter. Before we state the result, let us explain the idea. We return to the discussion in \cref{introduction} on general path integrals (with values in a target space $M$). Focusing on the Dirichlet boundary and ignoring the parametrizations $\zetabf$, the path integral is
\[\int_{\substack{\phi:\Sigma\to M\\\phi|_{\paD\Sigma}=\bt}}F(\phi)\,e^{-S(\phi)}\,\D\phi.\]
Suppose we cut $\Sigma$ into two surfaces $\Sigma_1$, $\Sigma_2$ along a simple curve $\Cc$. We consider $\Cc$ to be a Dirichlet boundary of both $\Sigma_1$ and $\Sigma_2$. Write $\bt=\bt_1\times\bt_2$ where $\bt_i=\bt|_{\paD\Sigma_i\setminus\Cc}$. Suppose the Lagrangian $S$ is \textit{local}, i.e.,
\[S(\phi)=S(\phi|_{\Sigma_1})+S(\phi|_{\Sigma_2}).\]
Then for a \textit{local} observable $F$ on $\Sigma$ of the form $F(\phi)=F_1(\phi|_{\Sigma_1})F_2(\phi|_{\Sigma_2})$ for some observables $F_i$ on $\Sigma_i$, it is natural to expect that
\[\int_{\substack{\phi:\Sigma\to M\\\phi|_{\paD\Sigma}=\bt}}F(\phi)\,e^{-S(\phi)}\,\D\phi=\int_{\wt{\varphi}:\Cc\to M}
\bigg(\int_{\substack{\phi_1:\Sigma_1\to M\\\phi_1|_{\paD\Sigma_1}=\wt{\varphi}\times\bt_1}}F_1(\phi_1)\,e^{-S(\phi_1)}\,\D\phi_1\bigg)
\bigg(\int_{\substack{\phi_2:\Sigma_2\to M\\\phi_2|_{\paD\Sigma_2}=\wt{\varphi}\times\bt_2}}F_2(\phi_2)\,e^{-S(\phi_2)}\,\D\phi_2\bigg)\D\wt{\varphi},\]
since both sides integrate over the same space of fields. We shall state this as an isomorphism of measure spaces
\[\bigsqcup_{\wt{\varphi}\in\Dc'(\Cc,M)}\Dc'^{\M,\wt{\varphi}\times\bt_1}(\Sigma_1,M)\times\Dc'^{\M,\wt{\varphi}\times\bt_2}(\Sigma_2,M)\xrightarrow{\sim}\Dc'^{\M,\bt}(\Sigma,M)\]
where we equip the left-hand side with an appropriate product measure.

Instead of cutting, we state the result for the inverse operation of gluing. In each of the four cases discussed in \cref{gluing-top}, one needs to distinguish between the cases $\paD\Sigma=\varnothing$ and $\paD\Sigma\neq\varnothing$ for the glued surface $\Sigma$, since the definitions of $\Dc'^\N(\Sigma\setminus\zbf,\TR)$ and $\Dc'^{\M,\bt}(\Sigma\setminus\zbf,\TR)$ are fundamentally different. This results in a total of eight cases. To avoid the tedium of stating a separate result for each case, we introduce the following conventions to unify the notation:

\begin{itemize}
    \item For an extended surface $\Sigma$ with $\paD\Sigma=\varnothing$, we define $\Dc'(\paS,\TR)=\{*\}$, $C^{\M,\const}(\Sigma\setminus\zbf,\TR)=C^\N(\Sigma\setminus\zbf,\TR)$, $\Dc'^{\M,*}(\Sigma\setminus\zbf,\TR)=\Dc'^\N(\Sigma\setminus\zbf,\TR)$, where $*$ denotes the empty tuple. This unifies the cases $\paD\Sigma=\varnothing$ and $\paD\Sigma\neq\varnothing$.
    \item We extend all previous definitions naturally to disjoint unions of extended surfaces. For example, if $\Sigma=\Sigma_1\sqcup\Sigma_2$, then $\Dc'^\M(\Sigma\setminus\zbf,\TR)=\Dc'^\M(\Sigma_1\setminus\zbf_1,\TR)\times\Dc'^\M(\Sigma_2\setminus\zbf_2,\TR)$. This unifies the cases of gluing two surfaces and self-gluing a surface.
    \item We allow arbitrary reordering of the parametrizations in $\zetabf$ and reindex them as $\zetabf=(\zeta_i)_{i=1}^{\bD+\bMD}$. This unifies the cases of gluing along a circle or semicircle.
\end{itemize}

Let $\Sigma$ be either an extended surface or a disjoint union of two extended surfaces with $\zeta_1$, $\zeta_2$ on different connected components. Suppose $\zeta_1$ is outgoing and $\zeta_2$ is incoming, and they both parametrize a circle or semicircle. Let $\Sigma^\#$ be the surface obtained by gluing the images of $\zeta_1$, $\zeta_2$ on $\Sigma$ via $\zeta_2\circ\zeta_1^{-1}$, which is naturally an extended surface with parametrizations $\zetabf_0=\zetabf^{\c\c}=(\zeta_i)_{i\geq3}$. We denote by $\Cc$ the common image of $\zeta_1$, $\zeta_2$ on $\Sigma^\#$ and $\zeta$ its induced parametrization. It is either a simple circle with $\Cc\subset(\Sigma^\#)^\circ$ or a Neumann extendible simple semicircle with $\Cc^\circ\subset(\Sigma^\#)^\circ$, $\pa\Cc\subset\paN\Sigma^\#$. We write $\zeta^*\Cc=\Tb$ or $\Tb_+$ accordingly.

We shall need the following interior Poisson and Dirichlet-to-Neumann operators. For $\wt{\varphi}\in\Dc'(\zeta^*\Cc,\Rb)$, we define $P_{\Sigma^\#,\Cc}\wt{\varphi}=P_\Sigma(\wt{\varphi}\times\wt{\varphi}\times\mathbf{0})$, $\Dbf_{\Sigma^\#,\Cc}\wt{\varphi}=\zeta_1^*\pa_\nu P_{\Sigma^\#,\Cc}\wt{\varphi}|_{\im\zeta_1}+\zeta_2^*\pa_\nu P_{\Sigma^\#,\Cc}\wt{\varphi}|_{\im\zeta_2}$. By integration by parts, for $\wt{\varphi}\in C^1(\zeta^*\Cc,\Rb)$, we have
\[\int_\Sigma|dP_{\Sigma^\#,\Cc}\wt{\varphi}|_g^2\,\d v_g=\int_{\zeta^*\Cc}\wt{\varphi}\,\Dbf_{\Sigma^\#,\Cc}\wt{\varphi}\,\d\theta,\]
where $g$ is any conformal metric on $\Sigma^\#$. Note that $P_{\Sigma^\#\Cc}1=1$ if and only if $\paD\Sigma^\#=\varnothing$, so
\[\ker\Dbf_{\Sigma^\#,\Cc}=\begin{cases}0,&\paD\Sigma^\#\neq\varnothing,\\\Rb,&\paD\Sigma^\#=\varnothing.\end{cases}\]
As in \cref{Poisson}, we have that $\Dbf_{\Sigma^\#,\Cc}-2\Dbf$ is smoothing.

For a functorial object $u$ on $\Sigma^\#$, we denote by $u|_\Sigma$ its pullback to $\Sigma$ via the quotient map $\Sigma\to\Sigma^\#$. Let $g$ be a Neumann extendible conformal metric on $\Sigma^\#$. We equip $\Sigma$ with the metric $g|_\Sigma$.

We are now ready to state the gluing theorem.

\begin{thm}[\textbf{Gluing}]\label{gluing}
For $\bt_0^{\kbf_0}\in\Dc'(\zetabf_0^*\paD\Sigma^\#,\TR)$, we have an isomorphism of measure spaces
\[\setlength\arraycolsep{1pt}
\begin{array}[t]{cccc}
\displaystyle\bigsqcup_{\wt{\varphi}^k\in\Dc'(\zeta^*\Cc,\TR)}  & \Dc'^{\M,\wt{\varphi}^k\times\wt{\varphi}^k\times\bt_0^{\kbf_0}}(\Sigma\setminus\zbf,\TR) & \xrightarrow{\sim} & \Dc'^{\M,\bt_0^{\kbf_0}}(\Sigma^\#\setminus\zbf,\TR)\\
&[(Y,\psi)]&\mapsto&[(Y+P_{\Sigma^\#,\Cc}(\varphi-h_\Cc),\psi-\ol{P_{\Sigma^\#,\Cc}(\varphi-h_\Cc)})]
\end{array}\]
where $h_\Cc=\zeta^*P_{\Sigma^\#}\varphibf_0|_\Cc$, and the measure on the left-hand side is $C\,\d\phi_{g|_\Sigma}\,\d\wt{\varphi}^k$ with
\[C=\begin{cases}(\sqrt{2}\pi)^{-1},&\paD\Sigma^\#\neq\varnothing,\;\Cc\,\textrm{circle},\\\sqrt{2},&\paD\Sigma^\#\neq\varnothing,\;\Cc\,\textrm{semicircle},\\
(2\pi)^{-3/4},&\paD\Sigma^\#=\varnothing,\;\Cc\,\textrm{circle},\\(2\pi)^{1/4},&\paD\Sigma^\#=\varnothing,\;\Cc\,\textrm{semicircle}.\end{cases}\]
\end{thm}

Comments on the statement:
\begin{itemize}
    \item Recall (\cref{section:measure-space1}) the decompositions $\bt_0^{\kbf_0}=\cbf_0^{\kbf_0}+\varphibf_0$, $\wt{\varphi}^k=c^k+\varphi$, where $k\in\Zb$ if $\Cc$ is a circle, $k=\h$ (for ``half'') if $\Cc$ is a semicircle. By definition, $P_{\Sigma^\#}\varphibf_0|_\Sigma=P_\Sigma(h_\Cc\times h_\Cc\times\varphibf_0)$. Here $\psi\in\ol{P_\Sigma(\varphi\times\varphi\times\varphibf_0)}+C^{\M,c^k\times c^k\times\cbf_0^\kbf}(\Sigma\setminus\zbf,\TR)$, so $\psi-\ol{P_{\Sigma^\#,\Cc}(\varphi-h_\Cc)}\in\ol{P_\Sigma(h_\Cc\times h_\Cc\times\varphibf_0)}+C^{\M,c^k\times c^k\times\cbf_0^\kbf}(\Sigma\setminus\zbf,\TR)$ glues to a map in $\ol{P_{\Sigma^\#}\varphibf_0}+C^{\M,\cbf_0^\kbf}(\Sigma^\#\setminus\zbf,\TR)$.
    \item There is a caveat with the interpretation of $Y+P_{\Sigma^\#,\Cc}(\varphi-h_\Cc)\in\Dc'(\Sigma^\#,\Rb)$. Since $\Dc(\Sigma,\Rb)\subsetneq\Dc(\Sigma^\#,\Rb)$, a general distribution on $\Sigma$ does not glue to a distribution on $\Sigma^\#$. Here $Y\in\Dc'(\Sigma^\#,\Rb)$ is defined a.s.\ on $\Dc_0'(\Sigma,\Rb)$: one checks that the series defining the GFF $X_{g|_\Sigma}$ on $\Sigma$ also converges a.s.\ in $\Dc'(\Sigma^\#,\Rb)$. As for the other term, $\varphi\in H^{-1/2}(\zeta^*\Cc,\Rb)$ a.s., so $P_{\Sigma^\#,\Cc}\varphi\in L^2(\Sigma,\Rb)=L^2(\Sigma^\#,\Rb)$ a.s. There is no problem with $P_{\Sigma^\#,\Cc}h_\Cc$ since $h_\Cc$ is smooth.
\end{itemize}

It is straightforward to check that this map is well-defined. Note that it is clearly compatible with the subcript $_\mbf$, i.e., it restricts to an isomorphism $\bigsqcup_{\wt{\varphi}^k\in\Dc'(\zeta^*\Cc,\TR)}\Dc_\mbf'^{\M,\wt{\varphi}^k\times\wt{\varphi}^k\times\bt_0^{\kbf_0}}(\Sigma\setminus\zbf,\TR)\xrightarrow{\sim}\Dc_\mbf'^{\M,\bt_0^{\kbf_0}}(\Sigma^\#\setminus\zbf,\TR)$.

The remainder of this section is devoted to the proof of this theorem.

\begin{lem}\label{Markov-reformulation}
We have an isomorphism of measure spaces
\[\setlength\arraycolsep{1pt}
\begin{array}[t]{ccccc}
\Dc'(\Sigma,\Rb) & \times & \Dc'(\zeta^*\Cc,\Rb) & \xrightarrow{\sim} & \Dc'(\Sigma^\#,\Rb)\\
X_{g|_\Sigma}&&\wt{\varphi}&\mapsto&X_{g|_\Sigma}+P_{\Sigma^\#,\Cc}\wt{\varphi}
\end{array}\]
where the measure on the left-hand side is
\[C\sqrt{\det\nolimits_\Fr(\Dbf_{\Sigma^\#,\Cc,0}(2\Dbf_0)^{-1})}\,e^{-\frac{1}{4\pi}\la\wt{\varphi},(\Dbf_{\Sigma^\#,\Cc}-2\Dbf)\wt{\varphi}\ra}\,\d X_{g|_\Sigma}\,\d\wt{\varphi}\]
with
\[C=\begin{cases}
1/\sqrt{\pi},&\paD\Sigma^\#\neq\varnothing,\;\Cc\,\textrm{circle},\\
1/\sqrt{2\pi},&\paD\Sigma^\#\neq\varnothing,\;\Cc\,\textrm{semicircle},\\
\sqrt{2},&\paD\Sigma^\#=\varnothing.
\end{cases}\]
Here the subscript $_0$ has the same meaning as in \cref{FH-var}.
\end{lem}
\begin{proof}
The inverse map is $Y\mapsto(Y|_\Sigma-P_{\Sigma^\#,\Cc}\zeta^*Y|_\Cc,\zeta^*Y|_\Cc)$. By the Markov property of GFFs, $X_g\law X_{g|_\Sigma}+P_{\Sigma^\#,\Cc}\zeta^*X_g|_\Cc$. Thus the measure on the left-hand side should be $\d X_{g|_\Sigma}\,\d\bt_\Cc$ where $\d\bt_\Cc$ is the pushforward of the measure on $\Dc'(\Sigma^{\#},\Rb)$ to $\Dc'(\zeta^*\Cc,\Rb)$ via the map $Y\mapsto\zeta^*Y|_\Cc$. We claim that $\d\bt_\Cc$ coincides with the measure induced by $\frac{1}{2\pi}\Dbf_{\Sigma^\#,\Cc}$ where the zero mode is not $L^2$-normalized. Then the comparison between $\d\bt_\Cc$ and $\d\wt{\varphi}$ follows from an application of the proof of \cref{FH-var} to $\frac{1}{2\pi}\Dbf_{\Sigma^\#,\Cc}$ and $\frac{1}{\pi}\Dbf$, with a different constant due to the normalization of the zero mode. More precisely, $\d X_{\frac{1}{\pi}\Dbf_0}=\sqrt{\frac{\vol\zeta^*\Cc}{2\pi^2}}\,e^{-\frac{\vol\zeta^*\Cc}{2\pi}c^2}\,\d c\,\d X_{\frac{1}{\pi}\Dbf}$, $\d X_{\frac{1}{2\pi}\Dbf_{\Sigma^\#,\Cc,0}}=\sqrt{\frac{\vol\zeta^*\Cc}{4\pi^2}}\,e^{-\frac{\vol\zeta^*\Cc}{4\pi}c^2}\,\d c\,\d X_{\frac{1}{2\pi}\Dbf_{\Sigma^\#,\Cc}}$ if $\paD\Sigma^\#=\varnothing$, hence $C$.

Clearly $\zeta^*X_g|_\Cc$ is a Gaussian field on $\zeta^*\Cc$ with covariance $\Eb[\wt{\varphi}(\theta)\wt{\varphi}(\theta')]=2\pi\,G_g(\zeta(e^{\i\theta}),\zeta(e^{\i\theta'}))$. If $\paD\Sigma^\#\neq\varnothing$, then $\Dbf_{\Sigma^\#,\Cc}$ is invertible and the Schwartz kernel of $\Dbf_{\Sigma^\#,\Cc}^{-1}$ is $K_{\Dbf_{\Sigma^\#,\Cc}^{-1}}=(\zeta\times\zeta)^*G_g|_{\Cc\times\Cc}$ (\cite[Lemma 5.6]{BLCFT}), so $\zeta^*X_g|_\Cc\law X_{\frac{1}{2\pi}\Dbf_{\Sigma^\#,\Cc}}$. If $\paD\Sigma^\#=\varnothing$, then the inverse of $\Dbf_{\Sigma^\#,\Cc}:\Dc_0'(\zeta^*\Cc,\Rb)\to\Dc_0'(\zeta^*\Cc,\Rb)$ is the projection of the operator $\Dc_0'(\zeta^*\Cc,\Rb)\to\Dc'(\zeta^*\Cc,\Rb)$ with Schwartz kernel $(\zeta\times\zeta)^*G_g|_{\Cc\times\Cc}$ to $\Dc_0'(\zeta^*\Cc,\Rb)$ (ibid.), so $\d\zeta^*(c+X_g)|_\Cc=\d c\,\d\zeta^*X_g|_\Cc=\d c\,\d X_{\frac{1}{2\pi}\Dbf_{\Sigma^\#,\Cc}}$.
\end{proof}

\begin{lem}\label{gluing-det}
We have
\[\textstyle\det'\Delta_g=C\det'\Delta_{g|_\Sigma}\det_\Fr(\Dbf_{\Sigma^\#,\Cc,0}(2\Dbf_0)^{-1})\]
with
\[C=\begin{cases}2\pi,&\paD\Sigma^\#\neq\varnothing,\;\Cc\,\textrm{circle},\\\sqrt{2\pi},&\paD\Sigma^\#\neq\varnothing,\;\Cc\,\textrm{semicircle},\\
\vol_g\Sigma,&\paD\Sigma^\#=\varnothing,\;\Cc\,\textrm{circle},\\\sqrt{\frac{2}{\pi}}\vol_g\Sigma,&\paD\Sigma^\#=\varnothing,\;\Cc\,\textrm{semicircle}.\end{cases}\]
\end{lem}
\begin{proof}
See \cite[Proposition B.3]{BLCFT}.
\end{proof}

\begin{lem}\label{lem-free-field-amplitude}
For $\bt_0\in\Dc'(\zetabf_0^*\paD\Sigma^\#,\Rb)$, $\wt{\varphi}\in\Dc'(\zeta^*\Cc,\Rb)$, we have
\begin{align*}
&\la\wt{\varphi}\times\wt{\varphi}\times\bt_0,(\Dbf_\Sigma-\Dbf)(\wt{\varphi}\times\wt{\varphi}\times\bt_0)\ra-\la\bt_0,(\Dbf_{\Sigma^\#}-\Dbf)\bt_0\ra\\
&\hspace{5em}=\la\wt{\varphi},(\Dbf_{\Sigma^\#,\Cc}-2\Dbf)\wt{\varphi}\ra-2\la\wt{\varphi},\Dbf_{\Sigma^\#,\Cc}h_\Cc\ra+\la h_\Cc,\Dbf_{\Sigma^\#,\Cc}h_\Cc\ra,
\end{align*}
where $h_\Cc=\zeta^*P_{\Sigma^\#}\bt_0|_\Cc$.
\end{lem}
\begin{proof}
By definition,
\[\la\wt{\varphi},(\Dbf_{\Sigma^\#,\Cc}-2\Dbf)\wt{\varphi}\ra=\la\wt{\varphi}\times\wt{\varphi}\times\mathbf{0},(\Dbf_\Sigma-\Dbf)(\wt{\varphi}\times\wt{\varphi}\times\mathbf{0})\ra.\]
Since $P_\Sigma(h_\Cc\times h_\Cc\times\bt_0)=P_{\Sigma^\#}\bt_0|_\Sigma$ where $P_{\Sigma^\#}\bt_0$ is smooth at $\Cc$, we have
\[\zeta_1^*\pa_\nu P_\Sigma(h_\Cc\times h_\Cc\times\bt_0)|_{\im\zeta_1}+\zeta_2^*\pa_\nu P_\Sigma(h_\Cc\times h_\Cc\times\bt_0)|_{\im\zeta_2}=0,\]
so
\[\la\wt{\varphi}\times\wt{\varphi}\times\mathbf{0},\Dbf_\Sigma(h_\Cc\times h_\Cc\times\bt_0)\ra=0,\qquad\la\wt{\varphi}\times\wt{\varphi}\times\mathbf{0},(\Dbf_\Sigma-\Dbf)(0\times0\times\bt_0)\ra=-\la\wt{\varphi},\Dbf_{\Sigma^\#,\Cc}h_\Cc\ra.\]
Similarly, $\la h_\Cc\times h_\Cc\times\mathbf{0},\Dbf_\Sigma(h_\Cc\times h_\Cc\times\bt_0)\ra=0$, so
\[\la h_\Cc,\Dbf_{\Sigma^\#,\Cc}h_\Cc\ra=\la h_\Cc\times h_\Cc\times\mathbf{0},\Dbf_\Sigma(h_\Cc\times h_\Cc\times\mathbf{0})\ra=-\la h_\Cc\times h_\Cc\times\mathbf{0},\Dbf_\Sigma(0\times 0\times\bt_0)\ra.\]
We have $\la 0\times0\times\bt_0,(\Dbf_\Sigma-\Dbf)(h_\Cc\times h_\Cc\times\bt_0)\ra=\la\bt_0,(\Dbf_{\Sigma^\#}-\Dbf)\bt_0\ra$, so
\[\la 0\times0\times\bt_0,(\Dbf_\Sigma-\Dbf)(0\times0\times\bt_0)\ra-\la\bt_0,(\Dbf_{\Sigma^\#}-\Dbf)\bt_0\ra=-\la 0\times0\times\bt_0,\Dbf_\Sigma(h_\Cc\times h_\Cc\times\mathbf{0})\ra=\la h_\Cc,\Dbf_{\Sigma^\#,\Cc}h_\Cc\ra.\]
Summing up, we get the desired formula.
\end{proof}

With these three lemmas, one can immediately prove gluing for the real theory, which in our notation can be stated as an isomorphism of measure spaces
\[\setlength\arraycolsep{1pt}
\begin{array}[t]{cccc}
\displaystyle\bigsqcup_{\wt{\varphi}\in\Dc'(\zeta^*\Cc,\Rb)}  & \Dc'^{\M,\wt{\varphi}\times\wt{\varphi}\times\bt_0}(\Sigma,\Rb) & \xrightarrow{\sim} & \Dc'^{\M,\bt_0}(\Sigma^\#,\Rb)\\
&P_\Sigma(\wt{\varphi}\times\wt{\varphi}\times\bt_0)+X_{g|_\Sigma}&\mapsto&P_{\Sigma^\#}\bt_0+(P_{\Sigma^\#,\Cc}(\wt{\varphi}-h_\Cc)+X_g)
\end{array}\]
where $h_\Cc=\zeta^*P_{\Sigma^\#}\bt_0|_\Cc$, and the measure on the left-hand side is $C\,\d\phi_{g|_\Sigma}\,\d\wt{\varphi}$ with the same $C$ coming from \cref{Markov-reformulation,gluing-det}. For BCILT, one still needs to study the gluing of the cohomology part, which has to be done on a case-by-case basis as in \cref{gluing-top}. We shall explain the idea and give the complete proof only in one case, the other cases being similar.

The gist of the proof is the following. We would like to apply \cref{Markov-reformulation}, but recall that the zero mode of $\wt{\varphi}^k$ is in $\TR$. By the discussion in \cref{gluing-top}, in all cases, we have a surjective map
\[\bigsqcup_{k\in H^1(\zeta^*\Cc)}H^1_\D(\Sigma)\times H_{k\times k\times\kbf_0}^1(\Sigma\setminus\zbf)\to H_\D^1(\Sigma^\#)\times H_{\kbf_0}^1(\Sigma^\#\setminus\zbf)\]
If $\paD\Sigma^\#\neq\varnothing$, this map has kernel $\Zb$, which complements the zero mode of $\wt{\varphi}^k$, so that the zero mode in the end is really integrated over $\Zb\times\TR\cong\Rb$ (as measure spaces), as in \cref{Markov-reformulation}. If $\paD\Sigma^\#=\varnothing$, this map is bijective, and we restrict the zero mode in \cref{Markov-reformulation} to $\TR$, giving exactly the definition of $\Dc'^\N(\Sigma\setminus\zbf,\TR)$ for $\paD\Sigma=\varnothing$ where the zero mode is over $\TR$. This shows that the map in \cref{gluing} is bijective on the set-theoretic level (up to sets of measure zero). The factors in the product measures cancel out by calculations using the Cameron--Martin theorem, similarly to the proofs of \cref{measure1,measure2}.

\begin{proof}[Proof of \cref{gluing} in the case of gluing two surfaces along a Dirichlet circle]
We resume the notations in \cref{glue1}. Let $g$ be a Neumann extendible conformal metric on $(\Sigma\sqcup\Sigma')^\#=\Sigma\#\Sigma'$.

\paragraph{The case $\paD(\Sigma\#\Sigma')\neq\varnothing$.}
Unraveling the notational conventions, the map in the theorem in this case is
\[\small\setlength\arraycolsep{1pt}
\begin{array}[t]{cccccc}
\displaystyle\bigsqcup_{\wt{\varphi}^k\in\Dc'(\Tb,\TR)} &  \Dc'^{\M,\wt{\varphi}^k\times(\bt^\c)^{\kbf^\c}}\!(\Sigma\setminus\zbf,\TR) & \times & \Dc'^{\M,\wt{\varphi}^k\times(\bt'^\c)^{\kbf'^\c}}\!(\Sigma'\setminus\zbf',\TR) & \xrightarrow{\sim} & \Dc'^{\M,(\bt^\c)^{\kbf^\c}\times(\bt'^\c)^{\kbf'^\c}}\!(\Sigma\#\Sigma'\setminus\zbf\cup\zbf',\TR)\\
&[(Y,\psi)]&&[(Y',\psi')]&\mapsto&[(Y+Y'+P_{\Sigma\#\Sigma',\Cc}(\varphi-h_\Cc),\psi^\#-\ol{P_{\Sigma\#\Sigma',\Cc}(\varphi-h_\Cc)})]
\end{array}\]
where $h_\Cc=\zeta^*P_{\Sigma\#\Sigma'}(\varphibf^\c\times\varphibf'^\c)|_\Cc$, $\psi^\#$ is defined to be $\psi$ on $\Sigma$ and $\psi'$ on $\Sigma'$.

First we explain the topology on the left-hand side. Let $\aux$ be an admissible function on $\Sigma\#\Sigma'$ such that $\aux=0$ near $\paD(\Sigma\#\Sigma')$, $\aux=1$ near $\Cc$. For $\wt{\varphi}^k=c^k+\varphi\in\Dc'(\Tb,\TR)$, we have a homeomorphism $\Dc'^{\M,0^k\times(\bt^\c)^{\kbf^\c}}\!(\Sigma\setminus\zbf,\TR)\xrightarrow{\sim}\Dc'^{\M,\wt{\varphi}^k\times(\bt^\c)^{\kbf^\c}}\!(\Sigma\setminus\zbf,\TR),[(Y,\psi)]\mapsto[(Y,\psi+P_\Sigma(\varphi\times\mathbf{0})+c\aux|_\Sigma)]$, and likewise for $\Sigma'$. Then the left-hand side is in bijection with
\[\TR\times\Dc_0'(\Tb,\Rb)\times\bigsqcup_{k\in\Zb}\big(\Dc'^{\M,0^k\times(\bt^\c)^{\kbf^\c}}\!(\Sigma\setminus\zbf,\TR)\times\Dc'^{\M,0^k\times(\bt'^\c)^{\kbf'^\c}}\!(\Sigma'\setminus\zbf',\TR)\big),\]
and we equip it with the pullback of the natural topology on this space. Clearly this topology does not depend on the choice of $\aux$.

To compare the pullback measure and the product measure $\d\phi_{g|_\Sigma}\,\d\phi_{g|_{\Sigma'}}\,\d\wt{\varphi}^k$ on the left-hand side, we make the following choices on $\Sigma$, $\Sigma'$, $\Sigma\#\Sigma'$ according to \cref{measure2} (we add $'$ for $\Sigma'$ and $^\#$ for $\Sigma\#\Sigma'$):
\begin{itemize}
\item We choose $x_0=\zeta_1^\D(1)\in\Sigma$, $x_0'={\zeta_1^\D}'(1)\in\Sigma'$, $x_0^\#\in\Sigma\#\Sigma'$ the common image of $x_0$, $x_0'$ on the glued circle $\Cc$.
\item We choose $c_0=c_0'=c$ (the zero mode of $\wt{\varphi}^k$), $c_0^\#=0$.
\item  In each cohomology class $[\omega^\#]\in H_{\kbf^\c\times\kbf'^\c}^1(\Sigma\#\Sigma'\setminus\zbf\cup\zbf')$ with $\int_\Cc\omega^\#=k$, we choose an admissible representative $\omega^\#$ such that $\zetabf_0^*I_{x_0^\#}(2\pi R\omega^\#)|_{\paD(\Sigma\#\Sigma')}=(\cbf^\c)^{\kbf^\c}\times(\cbf'^\c)^{\kbf'^\c}$, $\zeta^*I_{x_0^\#}(2\pi R\omega^\#)|_{\Cc}=0^k$, and we take it as the representative for defining the measure on $\Dc'^{\M,(\bt^\c)^{\kbf^\c}\times(\bt'^\c)^{\kbf'^\c}}\!(\Sigma\#\Sigma'\setminus\zbf\cup\zbf',\TR)$. For $c\in[0,2\pi R)\cong\TR$ (as measure spaces), $\omega_c=\omega^\#|_\Sigma+\frac{1}{2\pi R}c\;\!d\aux|_{\Sigma}$ is an admissible $1$-form on $\Sigma$ with $\zetabf^*I_{x_0}(2\pi R\omega_c)|_{\paD\Sigma}+c_0=c^k\times(\cbf^\c)^{\kbf^\c}$, and we take it as the representative for defining the measure on $\Dc'^{\M,\wt{\varphi}^k\times(\bt^\c)^{\kbf^\c}}\!(\Sigma\setminus\zbf,\TR)$. Likewise for $\Sigma'$. This is well-defined because $H_{\kbf^\c\times\kbf'^\c}^1(\Sigma\#\Sigma'\setminus\zbf\cup\zbf')\xrightarrow{\sim}\bigsqcup_{k\in\Zb}H_{k\times\kbf^\c}^1(\Sigma\setminus\zbf)\times H_{k\times\kbf'^\c}^1(\Sigma'\setminus\zbf')$ is a bijection.
\item For each cohomology class $[df^\#]\in H_\D^1(\Sigma\#\Sigma')$, we choose an admissible primitive $f^\#$ such that $f^\#=0$ near $\Cc$, and we take it as the representative for defining the measure on $\Dc'^{\M,(\bt^\c)^{\kbf^\c}\times(\bt'^\c)^{\kbf'^\c}}\!(\Sigma\#\Sigma'\setminus\zbf\cup\zbf',\TR)$. For $n\in\Zb$, $f_n=f^\#|_\Sigma+n(\aux|_\Sigma-1)$ is an admissible function on $\Sigma$, and we take it as the representative for defining the measure on $\Dc'^{\M,\wt{\varphi}^k\times(\bt^\c)^{\kbf^\c}}(\Sigma\setminus\zbf,\TR)$. Likewise for $\Sigma'$. Here $\Zb$ comes from the kernel $K=\Zb([d\aux|_\Sigma],[d\aux|_{\Sigma'}])$ of the surjective homomorphism $H_\D^1(\Sigma)\times H_\D^1(\Sigma')\to H_\D^1(\Sigma\#\Sigma')$ discussed in \cref{glue1}. Choosing primitives for $\Sigma\#\Sigma'$ in this way amounts to choosing a section for it. Then primitives for $\Sigma$, $\Sigma'$ are chosen using the induced bijection $(H_\D^1(\Sigma)\times H_\D^1(\Sigma'))/K\times\Zb\xrightarrow{\sim}H_\D^1(\Sigma)\times H_\D^1(\Sigma')$.
\end{itemize}

With these choices, the map in question becomes
\begin{align*}
&(\wt{\varphi}^k,[(X_{g|_\Sigma},\ol{P_\Sigma(\varphi\times\varphibf^\c)}+\ol{2\pi Rf_n}+I_{x_0}(2\pi R\omega_c)+c_0)],[(X_{g|_{\Sigma'}},\ol{P_{\Sigma'}(\varphi\times\varphibf'^\c)}+\ol{2\pi Rf_n'}+I_{x_0'}(2\pi R\omega_c')+c_0'])\\
&\qquad\mapsto[(X_{g|_\Sigma}+X_{g|_{\Sigma'}}+P_{\Sigma\#\Sigma',\Cc}(\varphi-h_\Cc),\ol{P_{\Sigma\#\Sigma'}(\varphibf^\c\times\varphibf'^\c)}+\ol{2\pi Rf^\#_n}+I_{x_0^\#}(2\pi R\omega^\#_c)+c_0)]\\
&\qquad=[(X_{g|_\Sigma}+X_{g|_{\Sigma'}}+P_{\Sigma\#\Sigma',\Cc}(\varphi-h_\Cc)+(2\pi Rn+c)\aux,\ol{P_{\Sigma\#\Sigma'}(\varphibf^\c\times\varphibf'^\c)}+\ol{2\pi Rf^\#}+I_{x_0^\#}(2\pi R\omega^\#))]
\end{align*}
which decomposes into three independent parts:
\begin{gather*}
\Dc_0'(\Tb,\Rb)\times[0,2\pi R)\times\Zb\times\Dc'(\Sigma,\Rb)\times\Dc'(\Sigma',\Rb)\to\Dc'(\Sigma\#\Sigma',\Rb)\\
\bigsqcup_{k\in\Zb}H_{k\times\kbf^\c}^1(\Sigma\setminus\zbf)\times H_{k\times\kbf'^\c}^1(\Sigma'\setminus\zbf')
\xrightarrow{\sim} H^1_{\kbf^\c\times\kbf'^\c}(\Sigma\#\Sigma'\setminus\zbf\cup\zbf')\\
(H_\D^1(\Sigma)\times H_\D^1(\Sigma'))/K\xrightarrow{\sim}H_\D^1(\Sigma\#\Sigma')
\end{gather*}
The second and third maps are bijective, so they trivially preserve the counting measures. Since $[0,2\pi R)\times\Zb\xrightarrow{\sim}\Rb,(c,n)\mapsto2\pi Rn+c$ as measure spaces, the first map is equivalent to $\Dc'(\Sigma,\Rb)\times\Dc'(\Sigma',\Rb)\times\Dc'(\Tb,\Rb)\xrightarrow{\sim}\Dc'(\Sigma\#\Sigma',\Rb),(X_{g|_\Sigma},X_{g|_{\Sigma'}},\wt{\varphi})\mapsto X_{g|_\Sigma}+X_{g|_{\Sigma'}}+P_{\Sigma\#\Sigma',\Cc}(\varphi-h_\Cc)+c\aux$, where the zero mode $c\in\Rb$. This differs from the map in \cref{Markov-reformulation} by a translation. It remains to collect the factors, i.e., the functions before the product measure on the left-hand side. Write $X_g=X_{g|_\Sigma}+X_{g|_{\Sigma'}}+P_{\Sigma\#\Sigma',\Cc}\wt{\varphi}$, $\aux'=P_{\Sigma\#\Sigma'}1-\aux$. By \cref{Markov-reformulation} and Cameron--Martin, the factor of the pullback measure is
\begin{align*}
&\frac{1}{\sqrt{\pi}}\sqrt{\textstyle\det_\Fr(\Dbf_{\Sigma\#\Sigma',\Cc}(2\Dbf_0)^{-1})}\,e^{-\frac{1}{4\pi}\la\wt{\varphi},(\Dbf_{\Sigma\#\Sigma',\Cc}-2\Dbf)\wt{\varphi}\ra}\times
\frac{e^{\frac{1}{8\pi}\int_{\paD(\Sigma\#\Sigma')}k_g\,\d\ell_g}}{\sqrt{\det'\Delta_g}}\,e^{-\frac{1}{4\pi}\la\varphibf^\c\times\varphibf'^\c,(\Dbf_{\Sigma\#\Sigma'}-\Dbf)(\varphibf^\c\times\varphibf'^\c)\ra}\\
&\times\exp\left(-\pi R^2\int_{\Sigma\#\Sigma'}^\reg|\omega^\#+df^\#|_g^2\,\d v_g-R\int_{\Sigma\#\Sigma'}(X_g-P_{\Sigma\#\Sigma',\Cc}h_\Cc-c\aux'+P_{\Sigma\#\Sigma'}(\varphibf^\c\times\varphibf'^\c))\,d^*(\omega^\#+df^\#)\,\d v_g\right)\\
&\times e^{-\frac{1}{4\pi}\la h_\Cc,\Dbf_{\Sigma\#\Sigma',\Cc}h_\Cc\ra+\frac{1}{2\pi}\la\wt{\varphi},\Dbf_{\Sigma\#\Sigma',\Cc}h_\Cc\ra}\exp\left(-\frac{c^2}{4\pi}\int_{\Sigma\#\Sigma'}|d\aux'|_g^2\,\d v_g+\frac{c}{2\pi}\int_{\Sigma\#\Sigma'}\la d(X_g-P_{\Sigma\#\Sigma',\Cc}h_\Cc),d\aux'\ra_g\,\d v_g\right).
\end{align*}
By definition, the factor of the product measure $\d\phi_{g|_\Sigma}\,\d\phi_{g|_{\Sigma'}}\,\d\wt{\varphi}^k$ is
\begin{align*}
&\frac{e^{\frac{1}{8\pi}(\int_{\paD\Sigma}k_g\,\d\ell_g+\int_{\paD\Sigma'}k_g\,\d\ell_g)}}{\sqrt{\det'\Delta_{g|_\Sigma}\det'\Delta_{g|_{\Sigma'}}}}\,e^{-\frac{1}{4\pi}\la\varphi\times\varphibf^\c,(\Dbf_\Sigma-\Dbf)(\varphi\times\varphibf^\c)\ra-\frac{1}{4\pi}\la\varphi\times\varphibf'^\c,(\Dbf_{\Sigma'}-\Dbf)(\varphi\times\varphibf'^\c)\ra)}\\
&\times\exp\bigg({-}\pi R^2\int_{\Sigma\#\Sigma'}^\reg\big|\omega^\#+df^\#+\tfrac{c}{2\pi R}du\big|_g^2\,\d v_g\\
&\hspace{5em}-R\int_{\Sigma\#\Sigma'}(X_g-P_{\Sigma\#\Sigma',\Cc}h_\Cc-cP_{\Sigma\#\Sigma',\Cc}1+P_{\Sigma\#\Sigma'}(\varphibf^\c\times\varphibf'^\c))\,d^*(\omega_c^\#+df_n^\#+\tfrac{c}{2\pi R}du)\,\d v_g\bigg).
\end{align*}
The difference between the integrals over $\Sigma\#\Sigma'$ is
\begin{align*}
&\frac{c^2}{4\pi}\int_{\Sigma\#\Sigma'}(|du|_g^2-|d\aux'|_g^2)\,\d v_g+\frac{c}{2\pi}\int_{\Sigma\#\Sigma'}\la d(X_g-P_{\Sigma\#\Sigma',\Cc}h_\Cc),d\aux'\ra_g\,\d v_g\\
&\hspace{5em}+\frac{c}{2\pi}\int_{\Sigma\#\Sigma'}\la d(X_g-P_{\Sigma\#\Sigma',\Cc}h_\Cc-cP_{\Sigma\#\Sigma',\Cc}1+P_{\Sigma\#\Sigma'}(\varphibf^\c\times\varphibf'^\c)),d\aux\ra_g\,\d v_g\\
&=-\frac{c^2}{4\pi}\int_{\Sigma\#\Sigma'}|dP_{\Sigma\#\Sigma',\Cc}1|_g^2\,\d v_g+\frac{c}{2\pi}\int_{\Sigma\#\Sigma'}\la dP_{\Sigma\#\Sigma',\Cc}(\wt{\varphi}-h_\Cc),dP_{\Sigma\#\Sigma',\Cc}1\ra_g\,\d v_g\\
&=\frac{c^2}{4\pi}\la1,\Dbf_{\Sigma\#\Sigma',\Cc}1\ra+\frac{c}{2\pi}\la\varphi-h_\Cc,\Dbf_{\Sigma\#\Sigma',\Cc}1\ra,
\end{align*}
where we used integration by parts several times. The other terms cancel out by \cref{gluing-det,lem-free-field-amplitude}.

\paragraph{The case $\paD(\Sigma\#\Sigma')=\varnothing$.}
The map in this case is
\[\setlength\arraycolsep{1pt}
\begin{array}[t]{cccccc}
\displaystyle\bigsqcup_{\wt{\varphi}^k\in\Dc'(\Tb,\TR)} &  \Dc'^{\M,\wt{\varphi}^k}(\Sigma\setminus\zbf,\TR) & \times & \Dc'^{\M,\wt{\varphi}^k}(\Sigma'\setminus\zbf',\TR) & \xrightarrow{\sim} & \Dc'^\N(\Sigma\#\Sigma'\setminus\zbf\cup\zbf',\TR)\\
&[(Y,\psi)]&&[(Y',\psi')]&\mapsto&[(Y+Y'+P_{\Sigma\#\Sigma',\Cc}\varphi,\psi^\#-\ol{P_{\Sigma\#\Sigma',\Cc}\varphi})]
\end{array}\]
where $\psi^\#$ is defined to be $\psi$ on $\Sigma$ and $\psi'$ on $\Sigma'$. We make the same choices as in the previous case. The difference now is that $K=0$. In fact, $H_\D^1(\Sigma)=H_\D^1(\Sigma')=H_\D^1(\Sigma\#\Sigma')=0$. The map decomposes into two independent parts:
\begin{gather*}
\Dc_0'(\Tb,\Rb)\times[0,2\pi R)\times\Dc'(\Sigma,\Rb)\times\Dc'(\Sigma',\Rb)\to\Dc'(\Sigma\#\Sigma',\Rb)\\
\bigsqcup_{k\in\Zb}H_k^1(\Sigma\setminus\zbf)\times H_k^1(\Sigma'\setminus\zbf')
\xrightarrow{\sim} H^1(\Sigma\#\Sigma'\setminus\zbf\cup\zbf')
\end{gather*}
Since $P_{\Sigma\#\Sigma'}1=1$, we can restrict the zero mode on both sides in \cref{Markov-reformulation} to $[0,2\pi R)$ to get an isomorphism of measure spaces
\[\Dc'(\Sigma,\Rb)\times\Dc'(\Sigma',\Rb)\times(\TR\times\Dc_0'(\Tb,\Rb))\to\TR\times\Dc_0'(\Sigma\#\Sigma',\Rb)\]
which is what we want, since $\Dc'^\N(\Sigma\#\Sigma'\setminus\zbf\cup\zbf',\TR)\cong\TR\times\Dc_0'(\Sigma\#\Sigma',\Rb)\times H^1(\Sigma\#\Sigma'\setminus\zbf\cup\zbf')$. One checks similarly that the factors match.
\end{proof}

\section{Curvature terms}\label{curv-term-section}

This section generalizes \cite[Section 4]{CILT}.

Let $\Sigma$ be an extended surface and $g$ a conformal metric on $\Sigma$. For the definition of BCILT, we need to make sense of the integrals
\[\frac{1}{2}\int_{\Sigma\setminus\zbf}K_gI_{x_0}(\omega)\,\d v_g+\int_{\pa\Sigma}k_gI_{x_0}(\omega)\,\d\ell_g,\]
where $\omega$ is a closed $1$-form on $\Sigma\setminus\zbf$, $x_0\in\Sigma\setminus\zbf$, $I_{x_0}(\omega)$ is as defined in \cref{section:measure-space2}. The problem is that $I_{x_0}(\omega)$ is multi-valued. To circumvent this, the idea is to integrate over a domain of full measure on which $I_{x_0}(\omega)$ is single-valued. By (2) of \cref{null-homologous3}, such a domain can be obtained by removing a separating family. Then one introduces suitable regularization terms to ensure that the resulting quantity is well-behaved under change of separating family.

The precise definition is as follows. For notational convenience, we relabel the set $C=\{c_i^\N\}_{i=1}^{\bN}\sqcup\{c_i^\D\}_{i=1}^{\bD}\sqcup\{c_i^\M\}_{i=1}^{\bM}\sqcup\{z_i\}_{i=1}^\s$ as $\{c_i\}_{i=1}^{\btilde}$, where $\btilde=\bN+\bD+\bM+\s$. For $c\in C$, we denote by $[c]$ the homology class of $c$ in $H_1(\Sigma\setminus\zbf)$, where $[z_i]$ is defined to be $-[\pc_i]$. In much of what follows, the type of an element in $C$ is irrelevant.

Let $\deltabf=(a_1,b_1,\ldots,a_\g,b_\g,d_1,\ldots,d_{\btilde-1})$ be a separating family of $\Sigma$ (with respect to an implicitly fixed $\vbf$). For each $d_i$, let $e_i$ be a simple circle on $(\Sigma\setminus\deltabf)\cup d_i^\circ$ such that $e_i$ intersects $d_i$ exactly once and the intersection is transversal and positively oriented. The homology class $[e_i]$ of $e_i$ does not depend on the choice of $e_i$. In fact, we have the following combinatorial description of $[e_i]$. Suppose $d_i$ goes from $c_s$ to $c_t$. Deleting $d_i$ from the tree formed by the $d_j$, it splits into two connected components. Let $C_s$ be the set of elements in $C$ in the connected component of $c_s$, and likewise $C_t$ for $c_t$, so that $C=C_s\sqcup C_t$. Then $[e_i]=\sum_{c\in C_s}[c]=-\sum_{c\in C_t}[c]$. Note that reversing the orientation of $d_i$ amounts to reversing the orientation of $e_i$.

Let $\omega$ be a closed $1$-form on $\Sigma\setminus\zbf$. For $x_0\in\Sigma\setminus\deltabf$, we define $I_{x_0}^\deltabf(\omega):\Sigma\setminus\deltabf\to\Rb,x\mapsto\int_\gamma\omega$ where $\gamma$ is any path on $\Sigma\setminus\deltabf$ from $x_0$ to $x$. By \cref{null-homologous3}, it is well-defined, i.e., single-valued. It is smooth on $\Sigma\setminus\deltabf$ and has jumps at $\deltabf$ given by the cycles of $\omega$. We assume that $I_{x_0}^\deltabf(\omega)$ is bounded, so that $I_{x_0}^\deltabf(\omega)\in L^\infty(\Sigma,\Rb)$, $I_{x_0}^\deltabf(\omega)|_{\pa\Sigma}\in L^\infty(\pa\Sigma,\Rb)$, since $\deltabf\cap\pa\Sigma$ is finite. Then we define
\begin{align*}
K_{\Sigma,g,x_0}^\deltabf(\omega)&=
\frac{1}{2}\int_\Sigma K_gI_{x_0}^\deltabf(\omega)\,\d v_g+\int_{\pa\Sigma}k_gI_{x_0}^{\deltabf}(\omega)\,\d\ell_g\\
&\phantom{{}={}}+\sum_{i=1}^\g\left(\int_{a_i}\omega\int_{b_i}k_g\,\d\ell_g-\int_{b_i}\omega\int_{a_i}k_g\,\d\ell_g\right)+\sum_{i=1}^{\btilde-1}\int_{e_i}\omega\int_{d_i}k_g\,\d\ell_g.
\end{align*}
Note that the last term does not depend on the orientation of $d_i$.

\begin{thm}[\textbf{Change of separating family}]\label{curv-inv}
Suppose $\int_c\omega\in\Zb$ for any circle $c$ on $\Sigma\setminus\zbf$. For two separating families $\deltabf$, $\deltabf'$ of $\Sigma$ and $x_0\in\Sigma\setminus(\deltabf\cup\deltabf')$, we have
\[K_{\Sigma,g,x_0}^{\deltabf'}(\omega)-K_{\Sigma,g,x_0}^\deltabf(\omega)\in\begin{cases}\pi\Zb,&\Sigma\textrm{ has no corners},\\\tfrac{1}{2}\pi\Zb,&\textrm{in general}.\end{cases}\]
In particular, this holds for admissible $1$-forms.
\end{thm}

\begin{proof}
Recall the Gauss--Bonnet formula: For a compact Riemannian surface $(\Omega,g)$ with piecewise smooth boundary,
\[\frac{1}{2}\int_\Omega K_g\,\d v_g+\int_{\pa\Omega}k_g\,\d\ell_g+\sum\alpha=2\pi\chi(\Omega),\]
where $\sum\alpha$ is the sum of the turning angles along $\pa\Omega$.

We first study the effect of changing the $d_i$. It suffices to consider the case of changing one $d_i$, i.e., suppose that $\deltabf'$ differs from $\deltabf$ only in $d_i$. We add a prime $'$ to denote the corresponding object for $\deltabf'$. For notational convenience, suppose $i=1$ and $d_1$ goes from $c_1$ to $c_2$.

\paragraph{Step 1: Changing $d_1$ to a homologous $d_1'$.}
First suppose $d_1'^\circ\cap d_1^\circ=\varnothing$. By (3) of \cref{null-homologous3}, let $\Omega$ be the connected component of $\Sigma^\circ\setminus(\deltabf\cup d_1')$ not containing $x_0$. For $i=1,2$, let $y_i=d_1\cap c_i$, $y_i'=d_1'\cap c_i$, and let $s_i$ be the segment $\pa\Omega\cap c_i$ on $c_i$ from $y_i$ to $y_i'$. (If $y_i=y_i'$, $s_i$ is degenerate.) Without loss of generality, suppose the orientations of $d_1'$, $s_1$ coincide with those induced by $\Omega$, so that the orientations of $d_1$, $s_2$ are reverse to those induced by $\Omega$. See \cref{gb-homologous-change} for a possible situation.

By definition, $I_{x_0}^{\deltabf'}(\omega)=I_{x_0}^\deltabf(\omega)$ except on $\Omega$, where $I_{x_0}^{\deltabf'}(\omega)=I_{x_0}^\deltabf(\omega)+\int_{e_1}\omega$, so
\begin{align*}
\int_\Sigma K_gI_{x_0}^{\deltabf'}(\omega)\,\d v_g-\int_\Sigma K_gI_{x_0}^\deltabf(\omega)\,\d v_g&=\int_{e_1}\omega\int_\Omega K_g\,\d v_g,\\
\int_{\pa\Sigma}k_gI_{x_0}^{\deltabf'}(\omega)\,\d\ell_g-\int_{\pa\Sigma}k_gI_{x_0}^\deltabf(\omega)\,\d\ell_g&=\int_{e_1}\omega\left(\int_{s_1}k_g\,\d\ell_g-\int_{s_2}k_g\,\d\ell_g+\sum_{c\in C_\Omega}\int_ck_g\,\d\ell_g\right),
\end{align*}
where $C_\Omega=\{c\in C\setminus\{z_i\}_i:c\subset\pa\Omega\}$. Since the $e_i$ are unchanged,
\[\sum_{i=1}^{\btilde-1}\int_{e_i'}\omega\int_{d_i'}k_g\,\d\ell_g-\sum_{i=1}^{\btilde-1}\int_{e_i}\omega\int_{d_i}k_g\,\d\ell_g=\int_{e_1}\omega\left(\int_{d_1'}k_g\,\d\ell_g-\int_{d_1}k_g\,\d\ell_g\right).\]
By the Gauss--Bonnet formula on $\Omega$,
\[\frac{1}{2}\int_\Omega K_g\,\d v_g+\int_{s_1}k_g\,\d\ell_g+\int_{d_1'}k_g\,\d\ell_g-\int_{s_2}k_g\,\ell_g-\int_{d_1}k_g\,\ell_g+\sum_{c\in C_\Omega}\int_ck_g\,\d\ell_g=2\pi\chi(\Omega)-\sum\alpha.\]
Note that the terms for the $d_i$ contained in $\pa\Omega$ cancel out because $\pa\Omega$ traverses them twice with opposite directions. Summing up, we get
\[K_{\Sigma,g,x_0}^{\deltabf'}(\omega)-K_{\Sigma,g,x_0}^\deltabf(\omega)=\int_{e_1}\omega\left(2\pi\chi(\Omega)-\sum\alpha\right).\]
Since $\int_{e_1}\omega\in\Zb$ by assumption and $\chi(\Omega)\in\Zb$, it remains to count $\sum\alpha$:
\begin{itemize}
\item The angles at $y_1$, $y_1'$, $y_2$, $y_2'$ contribute $2\pi$, regardless of whether $s_1$, $s_2$ are degenerate.
\item There may be right angles on $s_1$, $s_2$ or the circles in $C_\Omega$ coming from the corners of $\Sigma$, which contribute an integral multiple of $\pi/2$. If $\pa_\M\Sigma=\varnothing$, this part does not exist.
\item The ends of the $d_i$ contained in $\pa\Omega$ contribute an integral multiple of $\pi$: An end in a boundary circle contributes two right angles. At puncture ends, all the turning angles are $\pi$, since the tangent vectors of the $d_i$ are fixed.
\end{itemize}
Thus $K_{\Sigma,g,x_0}^{\deltabf'}(\omega)-K_{\Sigma,g,x_0}^\deltabf(\omega)\in\pi\Zb$ if $\pa_\M\Sigma=\varnothing$, $\frac{1}{2}\pi\Zb$ in general.

\begin{figure}
\centering
\includegraphics{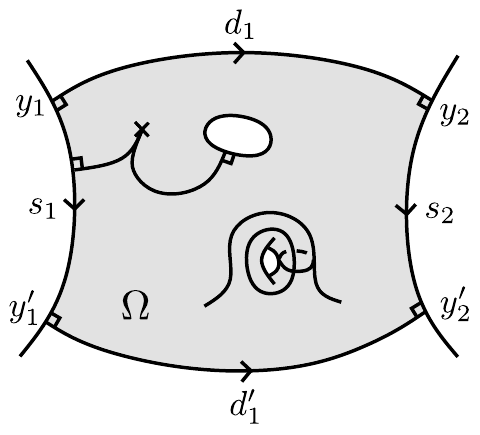}
\caption{Changing $d_1$ to a disjoint homologous $d_1'$}
\label{gb-homologous-change}
\end{figure}

For the general case, we use the following argument. First, for $d_1'$ close to $d_1$, it is easy to find $d_1''$ with $d_1''^\circ\cap d_1^\circ=d_1''^\circ\cap d_1'^\circ=\varnothing$, so that we can apply the above to the pairs $d_1,d_1''$ and $d_1',d_1''$ to get the result. This allows us to perturb $d_1$, $d_1'$. In particular, in the general case, we may assume that $d_1$, $d_1'$ intersect transversally. Then they intersect at finitely many points, giving a finite number of regions on which we apply the Gauss--Bonnet formula. Summing up, the angles at the intersections in $d_1^\circ\cap d_1'^\circ$ cancel out or give $\pi$, so the conclusion is the same.

\paragraph{Step 2: Changing $d_1$ to a nonhomologous $d_1'$.}
It suffices to consider the case of changing an end of $d_1$ to a neighboring element in $C$. For notational convenience, suppose $d_2$ goes from $c_2$ to $c_3$ and $d_1'$ goes from $c_1$ to $c_3$. Again, consider the case $d_1^\circ\cap d_1'^\circ=\varnothing$ for simplicity. As in the previous step, let $\Omega$ be the connected component of $\Sigma^\circ\setminus(\deltabf\cup d_1')$ not containing $x_0$. Let $y_1=d_1\cap c_1$, $y_2=d_1\cap c_2$, $y_2'=d_2\cap c_2$, $y_3=d_2\cap c_3$, $y_1'=d_1'\cap c_1$, $y_3'=d_1'\cap c_3$. For $i=1,2,3$, let $s_i$ be the segment $\pa\Omega\cap c_i$ on $c_i$ from $y_i$ to $y_i'$, which may be degenerate. Without loss of generality, suppose the orientations of $d_1'$, $s_1$ coincide with those induced by $\Omega$, so that the orientations of $d_1$, $s_2$, $d_2$, $s_3$ are reverse to those induced by $\Omega$. See \cref{gb-homology-change} for a possible situation.

Again, $I_{x_0}^{\deltabf'}(\omega)=I_{x_0}^\deltabf(\omega)$ except on $\Omega$, where $I_{x_0}^{\deltabf'}(\omega)=I_{x_0}^\deltabf(\omega)+\int_{e_1}\omega$, so
\begin{align*}
\int_\Sigma K_gI_{x_0}^{\deltabf'}(\omega)\,\d v_g-\int_\Sigma K_gI_{x_0}^\deltabf(\omega)\,\d v_g&=\int_{e_1}\omega\int_\Omega K_g\,\d v_g,\\
\int_{\pa\Sigma}k_gI_{x_0}^{\deltabf'}(\omega)\,\d\ell_g-\int_{\pa\Sigma}k_gI_{x_0}^\deltabf(\omega)\,\d\ell_g&=\int_{e_1}\omega\left(\int_{s_1}k_g\,\d\ell_g-\int_{s_2}k_g\,\d\ell_g-\int_{s_3}k_g\,\d\ell_g+\sum_{c\in C_\Omega}\int_ck_g\,\d\ell_g\right),
\end{align*}
where $C_\Omega=\{c\in C\setminus\{z_i\}_i:c\subset\pa\Omega\}$. Now the $e_i$ change. More precisely, we have the following relations. Deleting $d_1$, $d_2$ from the tree formed by the $d_i$, it splits into three connected components. For $i=1,2,3$, let $C_i$ be the set of elements in $C$ in the connected component of $c_i$, so that $C=C_1\sqcup C_2\sqcup C_3$. Then $[e_1]=\sum_{c\in C_1}[c]$, $[e_2]=\sum_{c\in C_1\sqcup C_2}[c]=-\sum_{c\in C_3}[c]$, $[e_1']=\sum_{c\in C_1}[c]=[e_1]$, $[e_2']=\sum_{c\in C_2}[c]=[e_2]-[e_1]$, and the other $e_i$ are unchanged. Thus
\[\sum_{i=1}^{\btilde-1}\int_{e_i'}\omega\int_{d_i'}k_g\,\d\ell_g-\sum_{i=1}^{\btilde-1}\int_{e_i}\omega\int_{d_i}k_g\,\d\ell_g=\int_{e_1}\omega\left(\int_{d_1'}k_g\,\d\ell_g-\int_{d_1}k_g\,\d\ell_g-\int_{d_2}k_g\,\d\ell_g\right).\]
By the Gauss--Bonnet formula on $\Omega$,
\[\frac{1}{2}\int_\Omega K_g\,\d v_g+\int_{s_1}k_g\,\d\ell_g+\int_{d_1'}k_g\,\d\ell_g-\int_{s_3}k_g\,\ell_g-\int_{d_2}k_g\,\ell_g-\int_{s_2}k_g\,\ell_g-\int_{d_1}k_g\,\ell_g+\sum_{c\in C_\Omega}\int_ck_g\,\d\ell_g=2\pi\chi(\Omega)-\sum\alpha.\]
Summing up, we get
\[K_{\Sigma,g,x_0}^{\deltabf'}(\omega)-K_{\Sigma,g,x_0}^\deltabf(\omega)=\int_{e_1}\omega\left(2\pi\chi(\Omega)-\sum\alpha\right).\]
The angles are counted similarly, the main difference being that the angles at $y_1$, $y_1'$, $y_2$, $y_2'$, $y_3$, $y_3'$ contribute $3\pi$.

\begin{figure}
\centering
\includegraphics{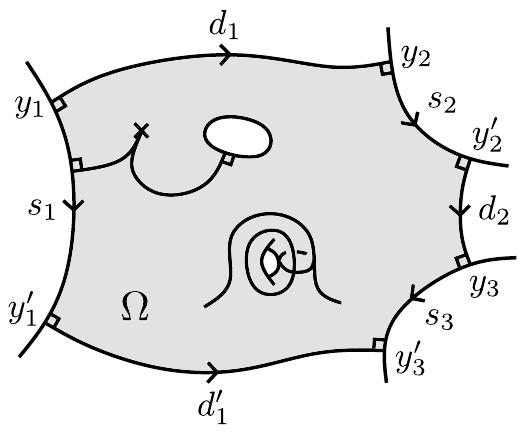}
\caption{Changing $d_1$ to a disjoint nonhomologous $d_1'$}
\label{gb-homology-change}
\end{figure}

\paragraph{Step 3: The general case.}
We turn to the general case. Our goal is to reduce to the case of closed surfaces, which was treated in detail in \cite[Appendix A]{CILT}. Let $\deltabf$, $\deltabf'$ be separating families of $\Sigma$.

Removing a small triangular region at each corner of $\Sigma$ as in \cref{smooth-corners}, we obtain a compact surface without corners $\Sigma'$ such that $(\Sigma\setminus\Sigma')\cap(\deltabf\cup\deltabf'\cup x_0)=\varnothing$. We view $\Sigma'$ as an extended surface with Dirichlet boundary, so that $\deltabf$, $\deltabf'$ are also separating families of $\Sigma'$. The point is that the theorem for $\Sigma$ is equivalent to the theorem for $\Sigma'$. Indeed, since the regularization terms are the same for $\Sigma$ and $\Sigma'$, we have
\[K_{\Sigma,g,x_0}^\deltabf(\omega)-K_{\Sigma',g,x_0}^\deltabf(\omega)=\frac{1}{2}\int_{\Sigma^\circ\setminus\Sigma'}K_gI_{x_0}^\deltabf(\omega)\,\d v_g+\int_{\pa(\Sigma^\circ\setminus\Sigma')}k_gI_{x_0}^\deltabf(\omega)\,\d\ell_g,\]
and likewise for $\deltabf'$. Thus
\begin{align*}
&(K_{\Sigma,g,x_0}^{\deltabf'}(\omega)-K_{\Sigma',g,x_0}^{\deltabf'}(\omega))-(K_{\Sigma,g,x_0}^\deltabf(\omega)-K_{\Sigma',g,x_0}^\deltabf(\omega))\\
&\hspace{5em}=\frac{1}{2}\int_{\Sigma^\circ\setminus\Sigma'}K_g(I_{x_0}^{\deltabf'}(\omega)-I_{x_0}^\deltabf(\omega))\,\d v_g+\int_{\pa(\Sigma^\circ\setminus\Sigma')}k_g(I_{x_0}^{\deltabf'}(\omega)-I_{x_0}^\deltabf(\omega))\,\d\ell_g.
\end{align*}
Now, on each connected component of $\Sigma^\circ\setminus\Sigma'$, i.e., on the removed region bounded by $\pa\Sigma$ and $\pa\Sigma'$ at a corner of $\Sigma$, $I_{x_0}^{\deltabf'}(\omega)-I_{x_0}^\deltabf(\omega)$ is constant with value given by a cycle of $\omega$, which is in $\Zb$ by assumption. We take it out of the integral and apply the Gauss--Bonnet formula to get $(K_{\Sigma,g,x_0}^{\deltabf'}(\omega)-K_{\Sigma',g,x_0}^{\deltabf'}(\omega))-(K_{\Sigma,g,x_0}^\deltabf(\omega)-K_{\Sigma',g,x_0}^\deltabf(\omega))=(K_{\Sigma,g,x_0}^{\deltabf'}(\omega)-K_{\Sigma,g,x_0}^\deltabf(\omega))-(K_{\Sigma',g,x_0}^{\deltabf'}(\omega)-K_{\Sigma',g,x_0}^\deltabf(\omega))\in\frac{1}{2}\pi\Zb$. Thus we may assume that $\Sigma$ has no corners.

In the same vein, one observes that the theorem is invariant after removing a small disk around each puncture $z_i$ (such that the boundary of the disk is orthogonal to the $d_j$ incident to $z_i$, see \cref{remove-puncture}). Thus we may assume that $\Sigma$ has no punctures.

\begin{figure}
\centering
\begin{subfigure}[t]{0.3\textwidth}
    \centering
    \includegraphics{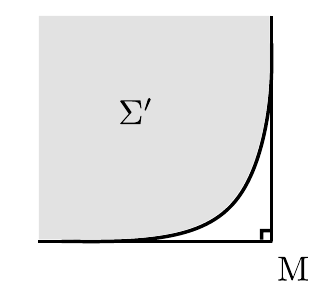}
    \caption{Eliminating a corner}
    \label{smooth-corners}
\end{subfigure}
\begin{subfigure}[t]{0.3\textwidth}
    \centering
    \raisebox{2ex}{\includegraphics{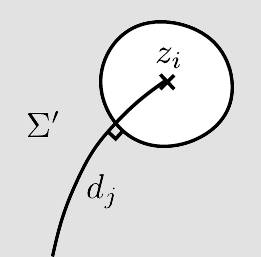}}
    \caption{Eliminating a puncture}
    \label{remove-puncture}
\end{subfigure}
\caption{Reducing to the case of closed surfaces}
\end{figure}

Now suppose $\Sigma$ has no corners or punctures. We use the same argument as in \cite[Lemma 8.3]{CILT} to remove the boundary. We label the boundary $\pa\Sigma$ as Neumann. By \cref{curv-term-calculations} below, we can take $g$ Neumann extendible and $x_0\in\pa\Sigma$. Let $\deltabf=(a_i,b_i)_{i=1}^{\g}\times(d_j)_{j=1}^{\btilde-1}$ be a separating family of $\Sigma$. By Step 2, we may change the $d_j$ so that $d_j$ goes from $c_j$ to $c_{\btilde}$ for all $j$. Consider the doubling $\Sigma^{\#2}$ of $\Sigma$ and the reflection $\tau$ on $\Sigma^{\#2}$. For all $j$, $d_j^{\#2}=d_j\cup\tau(d_j)$ glues to a simple circle on $\Sigma^{\#2}$, and it is easy to see that $\deltabf^{\#2}=(a_i,b_i,\tau(b_i),\tau(a_i))_{i=1}^\g\times(d_j^{\#2},c_j)_{j=1}^{\btilde-1}$ is a separating family (i.e., interior topological basis) of $\Sigma^{\#2}$, and $K_{\Sigma^{\#2},g^{\#2},x_0}^{\deltabf^{\#2}}(\omega^{\#2})=\frac{1}{2}K_{\Sigma,g,x_0}^{\deltabf}(\omega)$. One does the same with any other separating family $\deltabf'$ of $\Sigma$ to get $\deltabf'^{\#2}$, so $K_{\Sigma,g,x_0}^{\deltabf'}(\omega)-K_{\Sigma,g,x_0}^\deltabf(\omega)=\frac{1}{2}\big(K_{\Sigma^{\#2},g^{\#2},x_0}^{\deltabf'^{\#2}}(\omega^{\#2})-K_{\Sigma^{\#2},g^{\#2},x_0}^{\deltabf^{\#2}}(\omega^{\#2})\big)\in\pi\Zb$ by \cite[Appendix A]{CILT}.
\end{proof}

\begin{rem*}
Recall that we have a different convention for the $d_i$ in a separating family, which results in a different anomaly here. In the convention of \cite{CILT}, the anomaly is in $2\pi\Zb$. In our convention, the anomaly is at best in $\pi\Zb$, even for closed surfaces. We remark also that an anomaly in $\frac{\pi}{2}\Zb$ seems inevitable in the presence of corners.
\end{rem*}

The following are direct calculations:

\begin{prop}\label{curv-term-calculations}
The following hold:
\begin{enumerate}[label=\emph{(\arabic*)}]
\item $K_{\Sigma,g,x_0}^\deltabf(\omega)$ is linear in $\omega$.
\item For $f\in C^1(\Sigma\setminus\zbf,\Rb)\cap L^1(\Sigma,\Rb)$,
\[K_{\Sigma,g,x_0}^\deltabf(df)=\frac{1}{2}\int_\Sigma K_gf\,\d v_g+\int_{\pa\Sigma}k_gf\,\d\ell_g-(2\pi\chi(\Sigma)-\pi\bMD)f(x_0).\]
\item \emph{(\textbf{Change of base point})} For $x_0,x_0'\in\Sigma\setminus\deltabf$,
\[K_{\Sigma,g,x_0'}^\deltabf(\omega)-K_{\Sigma,g,x_0}^\deltabf(\omega)=(2\pi\chi(\Sigma)-\pi\bMD)\int_\gamma\omega,\]
where $\gamma$ is any path on $\Sigma\setminus\deltabf$ from $x_0$ to $x_0'$.
\item \emph{(\textbf{Change of conformal metric})} For $\rho\in C^\infty(\Sigma,\Rb)$,
\[K_{\Sigma,e^\rho g,x_0}^\deltabf(\omega)-K_{\Sigma,g,x_0}^\deltabf(\omega)=\frac{1}{2}\int_\Sigma\la d\rho,\omega\ra_g\,\d v_g.\]
\end{enumerate}
\end{prop}
\begin{proof}
Since $\Sigma$ has $2\bMD$ corners, the Gauss--Bonnet formula on $\Sigma$ reads
\[\frac{1}{2}\int_\Sigma K_g\,\d v_g+\int_{\pa\Sigma}k_g\,\d\ell_g=2\pi\chi(\Sigma)-\pi\bMD.\]
(1-3) are clear. As for (4), recall that
\[K_{e^\rho g}=e^{-\rho}(K_g+\Delta_g\rho),\qquad
\d v_{e^\rho g}=e^\rho\;\!\d v_g,\qquad
k_{e^\rho g}=e^{-\frac{1}{2}\rho}(k_g-\tfrac{1}{2}\pa_\nu\rho),\qquad
\d\ell_{e^\rho g}=e^{\frac{1}{2}\rho}\;\!\d\ell_g.\]
Thus
\begin{align*}
\int_\Sigma K_{e^\rho g}I_{x_0}^\deltabf(\omega)\,\d v_{e^\rho g}&=\int_\Sigma(K_g+\Delta_g\rho)I_{x_0}^\deltabf(\omega)\,\d v_g,\\
\int_{\pa\Sigma}k_{e^\rho g}I_{x_0}^\deltabf(\omega)\,\d\ell_{e^\rho g}&=\int_{\pa\Sigma}(k_g-\tfrac{1}{2}\pa_\nu\rho)I_{x_0}^\deltabf(\omega)\,\d\ell_g.
\end{align*}
By integration by parts on $\Sigma\setminus\sigmabf$,\footnote{We are really doing integration by parts on $\Sigma_\deltabf$ from the proof of \cref{null-homologous3}.}
\[\int_{\Sigma\setminus\sigmabf}\Delta_g\rho\,I_{x_0}^\deltabf(\omega)\,\d v_g=\int_{\sigmabf^\circ}\pa_\nu\rho\,(I_{x_0}^\deltabf(\omega)^+-I_{x_0}^\deltabf(\omega)^-)\,\d\ell_g+\int_{\pa\Sigma}\pa_\nu\rho\,I_{x_0}^\deltabf(\omega)\,\d\ell_g+\int_{\Sigma\setminus\sigma}\la d\rho,\omega\ra_\rho\,\d v_g,\]
where $I_{x_0}^\deltabf(\omega)^\pm$ is the limit of $I_{x_0}^\deltabf(\omega)$ as one approaches from the right (resp.\ left). By definition, $I_{x_0}^\deltabf(\omega)^+-I_{x_0}^\deltabf(\omega)^-$ is equal to $\int_{b_i}\omega$ on $a_i$, $-\int_{a_i}\omega$ on $b_i$, $-\int_{e_i}\omega$ on $d_i$. Thus
\[\int_{\sigmabf^\circ}\pa_\nu\rho\,(I_{x_0}^\deltabf(\omega)^+-I_{x_0}^\deltabf(\omega)^-)\,\d\ell_g=-\sum_{i=1}^\g\left(\int_{a_i}\omega\int_{b_i}\pa_\nu\rho\,\d\ell_g-\int_{b_i}\omega\int_{a_i}\pa_\nu\rho\,\d\ell_g\right)
-\sum_{i=1}^{\btilde-1}\int_{e_i}\omega\int_{d_i}\pa_\nu\rho\,\d\ell_g.\]
Summing up, we get the desired formula.
\end{proof}

\begin{prop}[\textbf{Change under gluing}]\label{curv-gluing}
We resume the notations in \cref{gluing-top}.
\begin{itemize}
\item In the case of gluing two surfaces, let $g$ be a conformal metric on $\Sigma\#\Sigma'$ and $x_0$ a point on the glued (semi)circle. For a closed $1$-form $\omega$ on $\Sigma\#\Sigma'$, $K_{\Sigma\#\Sigma',g,x_0}^{\deltabf\#\deltabf'}(\omega)=K_{\Sigma,g|_\Sigma,x_0}^\deltabf(\omega)+K_{\Sigma',g|_{\Sigma'},x_0}^{\deltabf'}(\omega)$.
\item In the case of self-gluing a surface, let $g$ be a conformal metric on $\#\Sigma$ and $x_0$ a point \emph{not} on the glued (semi)circle. For a closed $1$-form $\omega$ on $\#\Sigma$, $K_{\#\Sigma,g,x_0}^{\#\deltabf}(\omega)=K_{\Sigma,g|_\Sigma,x_0}^\deltabf(\omega)$ except in case 1 of the case of gluing along a semicircle, where $K_{\#\Sigma,g,x_0}^{\#\deltabf}(\omega)-K_{\Sigma,g|_\Sigma,x_0}^\deltabf(\omega)\in(\int_{d_1^\#}\omega)\pi\Zb$.
\end{itemize}
\end{prop}
\begin{proof}
The only nontrivial case is case 1 of the case of gluing along a semicircle (\cref{glue4-pic}), where we moved $\Cc$ to $c\not\subset\bigcup\deltabf$. We have $I_{x_0}^{\#\deltabf}(\omega)=I_{x_0}^\deltabf(\omega)$ except on the region $\Omega$ bounded by $c_2^\M$ and $c$, where $I_{x_0}^{\#\deltabf}(\omega)=I_{x_0}^\deltabf(\omega)-\int_{d_1^\#}\omega$. Thus
\begin{align*}
\int_{\#\Sigma}K_gI_{x_0}^{\#\deltabf}(\omega)\,\d v_g-\int_\Sigma K_gI_{x_0}^\deltabf(\omega)\,\d v_g&=-\int_{d_1^\#}\omega\int_\Omega K_g\,\d v_g,\\
\int_{\pa\#\Sigma}k_gI_{x_0}^{\#\deltabf}(\omega)\,\d\ell_g-\int_{\pa\Sigma}k_gI_{x_0}^\deltabf(\omega)\,\d\ell_g&=-\int_{d_1^\#}\omega\int_{c_2^\MN}k_g\,\d\ell_g-\left(\int_{c_2^\MD}k_gI_{x_0}^\deltabf(\omega)\,\d\ell_g-\int_{c_1^\MD}k_gI_{x_0}^\deltabf(\omega)\,\d\ell_g\right)\\
&=-\int_{d_1^\#}\omega\int_{c_2^\MN}k_g\,\d\ell_g-\int_{d_1^\#}\omega\int_{c_2^\MD}k_g\,\d\ell_g=-\int_{d_1^\#}\omega\int_{c_2^\M}k_g\,\d\ell_g.
\end{align*}
The difference between the regularization terms is
\[\int_c\omega\int_{d_1^\#}k_g\,\d\ell_g-\int_{d_1^\#}\omega\int_{c}k_g\,\d\ell_g-\int_{c_2^\M}\omega\int_{d_1}k_g\,\d\ell_g=-\int_{d_1^\#}\omega\int_ck_g\,\d\ell_g.\]
By the Gauss--Bonnet formula on $\Omega$,
\[\frac{1}{2}\int_\Omega K_g\,\d v_g+\int_ck_g\,\d\ell_g+\int_{c_2^\M}k_g\,\d\ell_g=-k_2\pi,\]
where $k_2$ is the number of Dirichlet (or Neumann) semicircles on $c_2^\M$, so that $c_2^\M$ has $2k_2$ corners. Summing up, we get the desired formula.
\end{proof}

Finally, we discuss the dependence on $\vbf$, which we have implicitly fixed up to this point. We view the tangent bundle of $\Sigma$ as a complex line bundle, so that $\lambda v$ is defined for $\lambda\in\Cb$ and $v$ a tangent vector on $\Sigma$. For $\thetabf=(\theta_1,\ldots,\theta_\s)\in\Rb^\s$, we write $e^{\i\thetabf}\vbf=(e^{\i\theta_1}v_1,\ldots,e^{\i\theta_\s}v_\s)$.

\begin{prop}[\textbf{Change of tangent vectors}]\label{curv-term-spin}
Let $\deltabf$ be a separating family with respect to $\vbf$. For $\thetabf\in[0,2\pi)^\s$, there exists a separating family $\deltabf'$ with respect to $e^{\i\thetabf}\vbf$ such that $K_{\Sigma,g,x_0}^{\deltabf'}(\omega)-K_{\Sigma,g,x_0}^\deltabf(\omega)=\sum_{i=1}^\s\theta_im_i$ for any $\omega$, where $m_i$ is the winding number of $\omega$ around $z_i$.
\end{prop}
\begin{proof}
Let $\chi\in C^\infty(\Rb_+,\Rb)$ with $\chi|_{[0,1/3]}=1$, $\chi|_{[2/3,\infty)}=0$. For $\alpha\in\Rb$, $R_\alpha:re^{\i\theta}\mapsto re^{\i(\theta+\alpha\chi(r))}$ is a diffeomorphism of $\Db$ with $R_\alpha=\id$ near $\pa\Db$, $R_\alpha=e^{\i\alpha}\,\id$ near $0$. For each $z_i$, let $\Psi_i:D_i\to\Db$ be a chart centered at $z_i$ such that the $D_i$ are disjoint and $x_0\not\in\bigsqcup_{i=1}^\s D_i$. Consider the diffeomorphism $\Phi$ of $\Sigma$ defined to be $\Psi_i^{-1}R_{\theta_i}\Psi_i$ on $D_i$ and $\id$ on $\Sigma\setminus\bigsqcup_{i=1}^\s D_i$. Then $\deltabf$ is a separating family with respect to $\vbf$ if and only if $\Phi_*\deltabf$ is a separating family with respect to $e^{\i\thetabf}\vbf$. We claim that $K_{\Sigma,g,x_0}^{\Phi_*\deltabf}(\omega)-K_{\Sigma,g,x_0}^\deltabf(\omega)=\sum_{i=1}^\s\theta_i\int_{\pc_i}\omega$. If $\deltabf$ is such that $z_i\cap d_j\neq\varnothing$ if and only if $i=j$, then this follows from a simple application of the Gauss--Bonnet formula on the region in each $D_i$ bounded by $d_i$ and $\Phi_*d_i$. Let $\deltabf'$ be another separating family with respect to $\vbf$. Since $\Phi$ preserves angles near $\pa\Sigma\cup\zbf$, it does not change the turning angles in the Gauss--Bonnet formulas in the proof of \cref{curv-inv}, so $K_{\Sigma,g,x_0}^{\deltabf'}(\omega)-K_{\Sigma,g,x_0}^\deltabf(\omega)=K_{\Sigma,g,x_0}^{\Phi_*\deltabf'}(\omega)-K_{\Sigma,g,x_0}^{\Phi_*\deltabf}(\omega)$, i.e., $K_{\Sigma,g,x_0}^{\Phi_*\deltabf'}(\omega)-K_{\Sigma,g,x_0}^{\deltabf'}(\omega)=K_{\Sigma,g,x_0}^{\Phi_*\deltabf}(\omega)-K_{\Sigma,g,x_0}^\deltabf(\omega)$. Thus the claim holds for any separating family $\deltabf$.
\end{proof}

\section{Imaginary Gaussian Multiplicative Chaos}\label{GMC}
In this section we focus on the GFF on $\Sigma$ with at least one Neumann boundary, i.e. $\paN\Sigma\not=\varnothing$. As the metric is assumed to be Neumann extendible, we have the doubled surface $\Sigma^{\#2}$ with a doubled metric $g^{\#2}$. Moreover, recall $\tau$ is the involution on $\Sigma^{\#2}$, we have the following relation between the green function $G_g^{\Sigma}$ and $G_{g^{\#2}}^{\Sigma^{\#2}}$
\begin{equation*}
    G^{\Sigma}_g(x,y)=G^{\Sigma^{\#2}}_{g^{\#2}}(x,y)+G^{\Sigma^{\#2}}_{g^{\#2}}(x,\tau(y)).
\end{equation*}
Similar to \cite[Section 5]{CILT}, based on certain regularized sequences $(X_{g,\epsilon})_{\epsilon>0}$, we can define the imaginary bulk (resp. boundary) Gaussian multiplicative chaos (GMC for short) $M_g(\d v_g)$ (resp. $L_g(\d \ell_g)$) as the following $L^2(\Omega)$ limit (see \cite[Section 5.4]{CILT} e.g.)
\begin{equation*}
    M_g(X_g,\d v_g):=\lim_{\epsilon\to0}\epsilon^{-\beta^2/2}e^{\i\beta X_{g,\epsilon}(x)} \d v_g(x) ,\ \ \  L_g(X_g,\d \ell_g):=\lim_{\epsilon\to0}\epsilon^{-\beta^2/4}e^{\i\beta /2X_{g,\epsilon}(x)} \d \ell_g(x).
\end{equation*}
For some $\rho\in C^{\infty}(\Sigma)$, we consider the metric $g':=e^\rho g$ conformal to $g$. By comparing the covariances under two metrics, we see that 
\[X_{g'}\overset{d}{=}X_{g}-m_{g'}(X_g),\]
where $m_g(\cdot)=\frac{1}{v_{g}(\Sigma)}\int\cdot\d v_{g}$. Moreover, we have the following rules under change of metric.
\begin{prop}
    For $g':=e^\rho g$, we have the following change of metric formula for imaginary GMC:
    \begin{equation*}
        M_{g'}(X_{g'},\d v_{g'})\overset{d}{=}e^{\i\beta(\i Q\rho/2-m_{g'}(X_g))}M_{g}(X_{g},\d v_{g}),
    \end{equation*}
        and
        \begin{equation*}
        L_{g'}(X_{g'},\d \ell_{g'})\overset{d}{=}e^{\i\beta/2(\i Q\rho/2-m_{g'}(X_g))}L_{g}(X_{g},\d \ell_{g}).
    \end{equation*}
\end{prop}
Assume first that $\Sigma$ has at least one Dirichlet boundary (including mixed boundaries). Suppose $p^{\Sigma^{\#2}}_t(x,y)$ is the transition kernel (depending on $g^{\#2}$) of a speed 2 Brownian motion on $\Sigma^{\#2}$ killed upon hitting the boundary. Then the covariance of $X_g$ (or Green function) is written as
\begin{equation*}
    2\pi\int_0^\infty p^{\Sigma^{\#2}}_{t}(x,x')+p^{\Sigma^{\#2}}_{t}(x,\tau(x'))\d t, x,x'\in \Sigma.
\end{equation*}
Furthermore, we have $X_{g,\epsilon}$ the white noise cut off 
\begin{equation*}
    X_{g,\epsilon}(x):=\sqrt{2\pi}\int_{\epsilon^2}^\infty\int_{\Sigma}p^{\Sigma^{\#2}}_{t/2}(x,y)+p^{\Sigma^{\#2}}_{t/2}(x,\tau(y))W(\d t,\d y),
\end{equation*}
whose covariance is
\begin{equation}\label{eq_variance_neuman_GFF_white_noise}
    \mathbb E [X_{g,\epsilon}(x)X_{g,\epsilon'}(x')]=2\pi\int_{\epsilon^2\vee\epsilon'^2}^\infty p^{\Sigma^{\#2}}_{t}(x,x')+p^{\Sigma^{\#2}}_{t}(x,\tau(x'))\d t.
\end{equation}

\begin{prop}
\label{prop_expo_momen}
Consider a measurable function $f:{\Sigma}\cup\partial \Sigma\to \mathbb C$ and an arbitrary real-valued random variable $Z$, for some constant $C$ depending on $\beta$, we have
    \begin{equation*}
    \begin{split}
        &\mathbb E\left[\exp\left(\left|\int_{\Sigma} f(x)M_g(X_g+Z,\mathrm{ d}v_g(x))\right|+\left|\int_{\paN\Sigma} f_g(x)L(X_g+Z,\d \ell_g(x))\right|\right)\right]\\
        \leq& e^{C V(f)}\sqrt{\left(1+C\sqrt{U_1(\Sigma,f)}e^{CU_1(\Sigma,f)}\right)\left(1+C\sqrt{U_2(\paN\Sigma,f)}e^{CU_2(\paN\Sigma,f)}\right)},
    \end{split}
    \end{equation*}
    where
    \begin{equation}\label{eq_defn_U1U2V}
    \begin{split}
        &U_1(\Sigma,f):=\iint_{\Sigma^2}|f(x)f(y)|(d(x,y)^{-\beta^2}+d(x,\tau(y))^{-\beta^2})\d v_g(x)\d v_g(y),\\
        &U_2(\paN\Sigma,f):=\iint_{\paN\Sigma^2}|f(x)f(y)|d(x,y)^{-\beta^2/2}\d \ell_g(x)\d \ell_g(y),\\
        &V(f):=\int_{\Sigma}|f(x)|d(x,\pa\Sigma)^{-\beta^2/2}\d v_g(x)+\int_{\paN\Sigma}|f(x)|d(x,\paD\Sigma)^{-\beta^2/4}\d \ell_g(x).
    \end{split}
    \end{equation}  
\end{prop}
\begin{proof}
    Similar to the proof in \cite[Proposition 5.3]{CILT}, we can throw away $Z$ and only need to consider functions s.t. $\mathrm{ Im}(f)=0$, i.e. we just need to prove the exponential moment for
    \begin{equation}\label{eq_prop_expo_momen_1}
        \int_{\Sigma}\epsilon^{-\beta^2/2}f(x)\cos(\beta X_{g,\epsilon}(x)) \d v_g(x)+ \int_{\paN\Sigma}\epsilon^{-\beta^2/4}f(x)\cos(\beta/2 X_{g,\epsilon}(x)) \d \ell_g(x).
    \end{equation}
    Replacing $\epsilon$ with $e^{-t}$, we have the following semi-martingale
    \begin{equation*}
        \mathcal{S}_t:= \int_{\Sigma}e^{\beta^2t/2}f(x)\cos(\beta X_{g,e^{-t}}(x)) \d v_g(x)+ \int_{\paN\Sigma}e^{\beta^2t/4}f(x)\cos(\beta/2 X_{g,e^{-t}}(x)) \d \ell_g(x),
    \end{equation*}
    whose limit when $t\to\infty$ is the limit of \eqref{eq_prop_expo_momen_1} as $\epsilon\to0$.

By the knowledge of It\^{o}'s calculus, we have $ \d\mathcal{S}_t=\d A_t+\d \mathcal M_t$, where
     \begin{equation*}
         \begin{split}
             \d \mathcal M_t:=&-\beta\int_{\Sigma} f(x)\sin(\beta X_{g,e^{-t}}(x))e^{\frac{\beta^2t}2}\d v_g(x)\mathrm{ d}X_{g,e^{-t}}(x)\\
         &-\frac\beta2 \int_{\paN\Sigma}f(x)\sin(\beta/2 X_{g,e^{-t}}(x))e^{\frac{\beta^2t}4}\d \ell_g(x)\mathrm{ d}X_{g,e^{-t}}(x),
         \end{split}
     \end{equation*}
     and
         \begin{equation*}
         \begin{split}
             \d A_t:=&\frac{\beta^2}{2}\int_{\Sigma} f(x)\cos(\beta X_{g,e^{-t}}(x))e^{\frac{\beta^2t}2}\d v_g(x)(\d t-\mathrm{ d}\langle X_{g,e^{-t}}(x)\rangle_t)\\
         +&\frac{\beta^2}4 \int_{\paN\Sigma}f(x)\cos(\beta/2 X_{g,e^{-t}}(x))e^{\frac{\beta^2t}4}\d \ell_g(x)(\d t-\frac12\mathrm{ d}\langle X_{g,e^{-t}}(x)\rangle_t).
         \end{split}
     \end{equation*}
      The quadratic variation of the local martingale part is
      \begin{equation*}
        \begin{split}
          \langle \mathcal{M}\rangle_t- \langle \mathcal{M}\rangle_0=\beta^2&\int_0^t\iint_{\Sigma^2}f(x)f(y)\sin(\beta X_{g,e^{-u}}(x))\sin(\beta X_{g,e^{-u}}(y))\\&e^{{\beta^2u}}\d v_g(x)\d v_g(y)\mathrm{ d}\langle X_{g,e^{-u}}(x), X_{g,e^{-u}}(y)\rangle_u\\
         +\beta^2/4 &\int_0^t\iint_{\paN\Sigma^2}f(x)f(y)\sin(\beta/2 X_{g,e^{-u}}(x))\sin(\beta/2 X_{g,e^{-u}}(y))\\
         &e^{\frac{\beta^2u}2}\d \ell_g(x)\d \ell_g(y)\mathrm{ d}\langle X_{g,e^{-u}}(x), X_{g,e^{-u}}(y)\rangle_u\\
         +\beta^2&\int_0^t\iint_{\paN\Sigma\times \Sigma}f(x)f(y)\sin(\beta X_{g,e^{-u}}(x))\sin(\beta/2 X_{g,e^{-u}}(y))\\
         &e^{\frac{3\beta^2u}4}\d v_g(x)\d \ell_g(y)\mathrm{ d}\langle X_{g,e^{-u}}(x), X_{g,e^{-u}}(y)\rangle_u.
    \end{split}
    \end{equation*}
    Recall that we have \eqref{eq_variance_neuman_GFF_white_noise}, and taking derivative w.r.t. $u$, we get
\begin{equation}\label{eq_qua_X}
    \mathrm{ d}\langle X_{g,e^{-u}}(x), X_{g,e^{-u}}(y)\rangle_u=4\pi e^{-2u}(p_{e^{-2u}}^{{\Sigma^{\#2}}}(x,y)+p^{{\Sigma^{\#2}}}_{e^{-2u}}(x,\tau(y)))\mathrm{ d}u.
\end{equation}
According to \cite{HKUpB}, we have the following estimate for the Dirichlet heat kernel $p^{\Sigma^{\#2}}_{e^{-2u}}(x,y)$
\begin{equation}\label{eq_heatK_ineq}
     e^{-2u}p_{e^{-2u}}^{\Sigma^{\#2}}(x,y)\leq C e^{-d(x,y)^2e^{2u}/C},
\end{equation}
where $C$ is a constant that may change from line to line. Combining everything and replacing all sine functions by 1, we see that
\begin{equation*}
\begin{split}
    \sup_{t\geq0}\langle \mathcal{M}\rangle_t- \langle \mathcal{M}\rangle_0&\leq C\int_0^\infty\iint_{\Sigma^2} |f(x)f(y)|e^{\beta^2t} (e^{-d(x,y)^2e^{2t}/C}+e^{-d(x,\tau(y))^2e^{2t}/C})\d v_g(x)\d v_g(y)\mathrm{ d}t\\
    &+C \int_0^\infty\iint_{\paN\Sigma^2}|f(x)f(y)|e^{\beta^2t/2}e^{-d(x,y)^2e^{2t}/C}\d \ell_g(x)\d \ell_g(y)\mathrm{ d}t\\
    &+C\int_0^\infty\iint_{\paN\Sigma\times \Sigma}|f(x)f(y)|e^{3\beta^2t/4}e^{-d(x,y)^2e^{2t}/C}\d v_g(x)\d \ell_g(y)\mathrm{ d}t.
\end{split}
\end{equation*}
Then we first integrate over time $t$ (and perform the change of variables $s=d(x,y)e^t$ or $s=d(x,\tau(y))e^t$), the above quantity becomes
\begin{equation}\label{eq_exp_momen_cross}
\begin{split}
    &C\iint_{\Sigma^2} |f(x)f(y)|(d(x,y)^{-\beta^2}+d(x,\tau(y))^{-\beta^2})\d v_g(x)\d v_g(y)\\
    &+C \iint_{\paN\Sigma^2}|f(x)f(y)|d(x,y)^{-\beta^2/2}\d \ell_g(x)\d \ell_g(y)\\
    &+C\iint_{\paN\Sigma\times \Sigma}|f(x)f(y)|d(x,y)^{-3\beta^2/4}\d v_g(x)\d \ell_g(y).
\end{split}
\end{equation}
Now we are back to $\mathcal M_t^0:=\mathcal M_t-\mathcal M_0$. By Fatou's lemma
\begin{equation*}
    \mathbb E[e^{\mu\mathcal M^0_{\infty}}]\leq e^{\frac{\mu^2}2\sup_{t\geq0} \langle \mathcal{M}^0\rangle_t},
\end{equation*}
because $\mathcal M^0$ is a martingale. And by Chernoff's bound, it is easy to see that
\begin{equation*}
    \mathbb P(|\mathcal M^0_{\infty}|\geq x)\leq2\exp\left(-\frac{x^2}{4C\sup_{t\geq0} \langle \mathcal{M}^0\rangle_t}\right).
\end{equation*}
Using Fubini's theorem, we have
\begin{equation*}
    \mathbb E[e^{\mathcal |\mathcal M^0_\infty|}]=1+\int_{0}^\infty \mathbb P(|\mathcal M^0_{\infty}|\geq x)e^x\mathrm{ d}x\leq1+C\sqrt{\sup_{t\geq0} \langle \mathcal{M}^0\rangle_t}\  e^{C\sup_{t\geq0} \langle \mathcal{M}^0\rangle_t}.
\end{equation*}
The $U_1(\Sigma,f),U_2(\paN\Sigma,f)$ in the Proposition \ref{prop_expo_momen} is the bulk and boundary part of $ \langle \mathcal{M}^0\rangle_t$ respectively. Note that we can avoid the cross term in the last line of \eqref{eq_exp_momen_cross} if we deal with the bulk and the boundary separately. Indeed, by Cauchy-Schwarz inequality,
\begin{equation*}
\begin{split}
    &\mathbb E\left[\exp\left(\left|\int_{\Sigma} f(x)M(X_g+Z,\mathrm{ d}v_g(x))\right|+\left|\int_{\paN\Sigma} f(x)L(X_g+Z,\d \ell_g(x))\right|\right)\right]\\
    \leq &\sqrt{\mathbb E\left[\exp\left(2\left|\int_{\Sigma} f(x)M(X_g+Z,\mathrm{ d}v_g(x))\right|\right)\right]\mathbb E\left[\exp\left(2\left|\int_{\paN\Sigma} f(x)L(X_g+Z,\d \ell_g(x))\right|\right)\right]}.
\end{split}
\end{equation*}
On the other hand, using the expression of \eqref{eq_qua_X}, we get that
\begin{equation}\label{eq_A_bound_1}
\begin{split}
    |A_{\infty}-A_0|\leq&\frac{\beta^2}{2}\int_0^\infty\int_{\Sigma} |f(x)|e^{\frac{\beta^2t}2}|1-4\pi e^{-2u}(p_{e^{-2u}}^{{\Sigma^{\#2}}}(x,x)+p^{{\Sigma^{\#2}}}_{e^{-2u}}(x,\tau(x)))|\d v_g(x)\d t\\
         +&\frac{\beta^2}4 \int_{\paN\Sigma}|f(x)|e^{\frac{\beta^2t}4}|1-4\pi e^{-2u}p_{e^{-2u}}^{{\Sigma^{\#2}}}(x,x)|\d \ell_g(x)\d t
\end{split}    
\end{equation}
The first line in the above equation can be further bounded by
\begin{equation}\label{eq_A_bound_2}
\begin{split}
        &\frac{\beta^2}{2}\int_0^\infty\int_{\Sigma} |f(x)|e^{\frac{\beta^2t}2}|1-4\pi e^{-2u}p_{e^{-2u}}^{{\Sigma^{\#2}}}(x,x)|\d v_g(x)\d t\\
        +&\frac{\beta^2}{2}\int_0^\infty\int_{\Sigma} |f(x)|e^{\frac{\beta^2t}2}|4\pi e^{-2u}p^{{\Sigma^{\#2}}}_{e^{-2u}}(x,\tau(x))|\d v_g(x)\d t,
\end{split}
\end{equation}
where, according to \eqref{eq_heatK_ineq}, the second term is readily controlled by $\int_{\Sigma}|f(x)|d(x,\tau(x))^{-\beta^2/2}\d v_g(x)$. The remaining parts of $|A_{\infty}-A_0|$ require the following estimate of the heat kernel by \cite{McKean:1967xf}
\begin{equation*}
    |4\pi e^{-2u} p_{e^{-2u}}^{\Sigma^{\#2} }(x,x)-1|\leq C(e^{-2u}+e^{-d(x,\paD\Sigma)^2e^{2u}/C}),
\end{equation*}
which is true for some $C>0$ and for all $u\in(0,\infty)$. Thus, plugging the estimate into \eqref{eq_A_bound_1} and \eqref{eq_A_bound_2}, we have
\begin{equation*}
    |A_{\infty}-A_0|\leq C\int_{\Sigma}|f(x)|d(x,\pa\Sigma)^{-\beta^2/2}\d v_g(x)+C\int_{\paN\Sigma}|f(x)|d(x,\paD\Sigma)^{-\beta^2/4}\d \ell_g(x).
\end{equation*}
Finally, we only need to control $|\mathcal S_0|=|\mathcal M_0+A_0|$, which is just
 \begin{equation*}
         |\mathcal{S}_0|\leq C(\int_{\Sigma} |f(x)|\d v_g(x)+ \int_{\paN\Sigma}|f(x)|\d \ell_g(x)).
    \end{equation*}
    This part is bounded by the estimate for $|A_{\infty}-A_0|$, which corresponds to $V(f)$.
\end{proof}

The renormalization of the bulk GMC produces singularities on the Dirichlet boundaries of $\Sigma$ (but still has an exponential moment), as is stressed in \cite[Proposition 5.2]{CILT}. Now, we would focus on surfaces with only Neumann boundaries, but we have to create a Dirichlet boundary (so that we have the white noise decomposition of GFF) and remove it later. 

Specifically, we cut the surface $\Sigma$ using three geodesic circles in its interior, which would be helpful later. Abusively, let $d(\cdot,\cdot)$ be the distance induced by the metric $g$ or $g^{\#2}$. Denote by $D_1$, $D_2$ and $D_3$ three geodesic closed balls s.t. 
\begin{itemize}
    \item $D_1\subset D_2\subset D_3$ and $d(D_1,\partial D_2)>0,d(D_2,\partial D_3)>0$, i.e. $D_i$ is strictly contained in $D_{i+1}$;
    \item $d(D_3,\partial \Sigma)>0$, i.e. $D_3$ stays away from the boundary of $\Sigma$.
\end{itemize}

For each of these balls, according to the domain Markov property of a GFF (see \cite[Section 5.3]{CILT} e.g.) on $\Sigma$, we can decompose $X_g$ into three independent parts
\begin{equation*}
    X^{\Sigma\backslash D_i}_g+X^{D_i}_g+P_i, i=1,2,3,
\end{equation*}
where
\begin{itemize}
    \item $ X^{\Sigma\backslash D_i}_g$ is a GFF with Neumann conditions on $\partial\Sigma$ and Dirichlet condition on $\partial D_i$;
    \item $X^{D_i}_g$ is a GFF with Dirichlet condition on $\partial D_i$;
    \item  $P$ is the sum of a harmonic extension of $X|_{\partial D_i}$ onto $\Sigma$ satisfying Neumann condition on $\partial \Sigma$ and a Gaussian random variable.
\end{itemize}
Let $A_i:={\Sigma\backslash D_i}$. Then using the above proposition and Cauchy-Schwarz inequality, we have
\begin{cor}
\label{coro_expo_momen}
    Consider a measurable function $f:\Sigma\cup\partial\Sigma\to \mathbb C$ and an arbitrary real-valued random variable $Z$, for some constant $C$ depending on $\beta$ and the way we choose $D_1,D_2$ and $D_3$, we have
    \begin{equation*}
    \begin{split}
        &\mathbb E\left[\exp\left(\left|\int_{\Sigma} f(x)M(X_g+Z,\mathrm{ d}v_g(x))\right|+\left|\int_{\paN\Sigma} f(x)L(X_g+Z,\d \ell_g(x))\right|\right)\right]\\
        \leq& e^{C V(f)}\sqrt[3]{\left(1+C\sqrt{U_1(D_2,f)}e^{CU_1(D_2,f)}\right)\left(1+C\sqrt{U_1(A_2,f)}e^{CU_1(A_2,f)}\right)}\\
        &\sqrt[3]{\left(1+C\sqrt{U_2(\paN\Sigma,f)}e^{CU_2(\paN\Sigma,f)}\right)},
    \end{split}
    \end{equation*}
    where $U_1,U_2,V$ are the same as the one in Proposition \ref{prop_expo_momen}.
\end{cor}
Finally, we prove the Imaginary version of Cameron-Martin theorem for imaginary GMC on an arbitrary surface.
\begin{prop}\label{prop_im_CM}
    Consider a bounded functional $F$ of GFF $X_g$ s.t. $F(X_g+wY)$ has an analytic continuation in $w$ for the field $Y:=2\pi\int G_g(x,y)\d m(y)$, where $m$ is an arbitrary measure s.t. $Y$ is continuous on $\Sigma\cup\partial\Sigma$. $f(x)$ is a bounded measurable function on $\Sigma\cup\partial\Sigma$ and $Z$ is a random variable. Then we have for $\alpha\in\mathbb R$
    \begin{equation*}
    \begin{split}
        &\mathbb E\Bigg[F(X_g)\exp\Bigg(\i\alpha\int X\d m+\frac{\alpha^2}{2}\int 2\pi G_g(x,y)\d m(x)\d m(y)+\int_{\Sigma} f(x)M_g(Z+X_g,\mathrm{ d}v_g(x))\\
        &\quad\quad\quad\quad\quad\quad+\int_{\partial \Sigma} f(x)L_g(Z+X_g,\mathrm{ d}\ell_g(x))\Bigg)\Bigg]\\
        =&\mathbb E\left[F\left(X_g+\i\alpha Y\right)\exp\left(\int_{\Sigma} f(x)M_g\left(Z+X_g+\i\alpha Y,\mathrm{ d}v_g(x)\right)+\int_{\partial \Sigma} f(x)L_g\left(Z+X_g+\i\alpha Y,\mathrm{ d}\ell_g(x)\right)\right)\right].
    \end{split}
    \end{equation*}
\end{prop}
\begin{proof}
   We introduce the function
    \begin{equation*}
    \begin{split}
        A(w):=\mathbb E\Bigg[F(X_g)\exp\Bigg(&w\int X\d m-\frac{w^2}{2}\int 2\pi G_g(x,y)\d m(x)\d m(y)\\
        &+\int_{\Sigma} f(x)M_g(Z+X_g,\mathrm{ d}v_g(x))
        +\int_{\partial \Sigma} f(x)L_g(Z+X_g,\mathrm{ d}\ell_g(x))\Bigg)\Bigg],
    \end{split}
    \end{equation*}
    which is holomorphic for $w\in\mathbb C$. When $w\in \i \mathbb R$, we can apply Cameron-Martin theorem to $A(w)$ and obtain
    \begin{equation}\label{eq_correlation_1}
        A(w)=\mathbb E\left[F\left(X_g+\i wY\right)\exp\left(\int_{\Sigma} f(x)M_g\left(Z+X_g+\i w Y,\d v_g(x)\right)+\int_{\partial \Sigma} f(x)L_g\left(Z+X_g+\i w Y,\d\ell_g(x)\right)\right)\right].
    \end{equation}
We want to show that \eqref{eq_correlation_1} is holomorphic as well, so as to apply analytic continuation to $A(w)$ before and after the shift by Cameron-Martin theorem. Hence we will need to approximate $\int f(x)M_g(Z+X_g+iwY,\d v_g(x))$ and $\int f(x)L_g(Z+X_g+iwY,\d\ell_g(x))$ by holomorphic functions. Note that the definition of $F$ is important to the proof. To start with, we claim they are Schwarz distributions with order 2 before the shift.
 \begin{lem}
    \label{lemma_M_is_dist2}
        $ f(x)M_g(Z+X_g,\d v_g(x))$ and $f(x)L_g(Z+X_g,\d\ell_g(x))$ are two random (Schwarz) distributions with order 2 on $\Sigma$. Moreover, there exists a $L^2$ random variable $D_{\Sigma}$ such that for a smooth function $h$ on $\bar\Sigma$,
        \begin{equation*}
            \left|\int_{\Sigma}h(x)f(x)M_g(Z+X_g,\d v_g(x))+\int_{\partial\Sigma}h(x)f(x)L_g(Z+X_g,\d\ell_g(x))\right|\leq D_{\Sigma}(\sup_{\bar\Sigma}|h|+\sup_{\Sigma}|\Delta_g h|+\sup_{\partial\Sigma}|\partial_\nu h|).
        \end{equation*}
    \end{lem}
    \begin{proof}
     By Green's identity, 
     \begin{equation*}
         \int_{\Sigma}G_g(x,y)\Delta_{g}h(y)\d v_g(y)-h(x)+\frac{1}{v_g(\Sigma)}\int_{\Sigma}h(y)\d v_g(y)=\int_{\partial\Sigma}G(x,y)\partial_\nu h(y)\d\ell_g(y).
     \end{equation*}
     Hence the smooth function $h$ can be written as \\$h(x)=\frac{1}{v_g(\Sigma)}\int_{\Sigma}h(y)\d v_g(y)+\int_{\Sigma}G_g(x,y)\Delta_gh(y)\d v_g(y)-\int_{\partial\Sigma}G_g(x,y)\partial_\nu h(y)\d\ell_g(y)$. Similar to \cite[Lemma 6.12]{CILT}, we have
        \begin{equation*}
        \begin{split}
&\left|\int_{\Sigma}h(x)f(x)M_g(Z+X_g,\d v_g(x))+\int_{\partial\Sigma}h(x)f(x)L_g(Z+X_g,\d\ell_g(x))\right|\leq(\sup_{\bar\Sigma}|h|+\sup_{\Sigma}|\Delta_g h|+\sup_{\partial\Sigma}|\partial_\nu h|)\\
&\bigg( \left|\int f(x)M_g(Z+X_g,\d v_g(x))\right| +\left|\int f(x)L_g(Z+X_g,\d\ell_g(x))\right|\\
+&\int\left|\int G_g(x,y)f(x)M_g(Z+X_g,\d v_g(x))\right|\d v_g(y)+\int\left|\int G_g(x,y)f(x)L_g(Z+X_g,\d\ell_g(x))\right|\d v_g(y)\\
&+\int\left|\int G_g(x,y)f(x)M_g(Z+X_g,\d v_g(x))\right|\d \ell_g(y)+\int\left|\int G_g(x,y)f(x)L_g(Z+X_g,\d\ell_g(x))\right|\d \ell_g(y)\bigg).
        \end{split}
        \end{equation*}
    We only need to check that 
    \begin{equation*}
        D_{\Sigma}:=|M|+|L|+\int|\int Gf\d M|\d v+\int|\int Gf\d L|\d v+\int|\int Gf\d M|\d \ell+\int|\int Gf\mathrm{d}L|\d \ell
    \end{equation*}
    has a finite second moment. And by Jensen's inequality, it is easy to see that we only need to control
    \begin{equation*}
        |\int Gf\mathrm{d}M|^2+|\int Gf\mathrm{d}L|^2,
    \end{equation*}
    which amounts to control locally
    \begin{equation*}
        \int_{\mathbb D}\int_{\mathbb D}\log\frac1{|x-y|}\log\frac1{|x'-y|}|x-x'|^{-\beta^2}\mathrm{d}x \mathrm{d}x'+\int_{-1}^1 \int_{-1}^1 \log\frac1{|x-y|}\log\frac1{|x'-y|}|x-x'|^{-\beta^2/2}\mathrm{d}x \mathrm{d}x'<\infty.
    \end{equation*}
    Since $\beta^2<2$, this integral over disk and interval is finite for sure.
    \end{proof}
Next, we mollify $Y$ with some smooth function $u_{\delta}$, i.e. by the convolution $Y*u_{\delta}(x):=\int_{\Sigma}u_{\delta}(x-y)Y(y)\d v_g(y)$. we require that $\int_{\Sigma}u_{\delta}(x-y)\d v_g(y)\equiv1$ for all $x\in\bar{\Sigma}$. Since $Y$ is continuous on $\bar{\Sigma}$, we can easily find a sequence of $(u_\delta)_{\delta>0}$ s.t. $\sup_{\bar{\Sigma}}|Y-Y*u_{\delta}|\to0$ as $\delta\to0$.

Then we replace the $M_g(X_g+wY,\d v_g(x))$ ($L$ resp.) in \eqref{eq_correlation_1} by $M_g(X_g+w Y*u_{\delta},\d v_g(x))$ ($L$ resp.), and the latter ones are holomorphic as a consequence of Lemma \ref{lemma_M_is_dist2} (which justifies the differentiations by Lebesgue convergence theorem). We see that
\begin{equation}\label{eq_correlation_delta}
   \mathbb E\left[F\left(Z+X_g+\i wY*u_{\delta}\right)\exp\left(\int_{\Sigma} f(x)M_g\left(Z+X_g+\i w Y*u_{\delta},\d v_g(x)\right)+\int_{\partial \Sigma} f(x)L_g\left(Z+X_g+\i w Y*u_{\delta},\d\ell_g(x)\right)\right)\right]
\end{equation}
is holomorphic for any fixed $\delta>0$, since according to Corollary \ref{coro_expo_momen}, the above quantity is uniformly bounded on any compact subset $K\subset \mathbb C$.
Now to prove \eqref{eq_correlation_1} is holomorphic, we only need to show that \eqref{eq_correlation_delta} converges locally uniformly w.r.t. $w$ to \eqref{eq_correlation_1} when $\delta\to0$. It is enough to prove that when $\delta\to0$,
\begin{equation}\label{eq_correlation_converge_d}
\begin{split}
    \mathbb E\bigg[\bigg|\exp\bigg(&\int_{\Sigma} f(x)M_g\left(Z+X_g+\i w Y,\d v_g(x)\right)-\int_{\Sigma} f(x)M_g\left(Z+X_g+\i w Y*u_{\delta},\d v_g(x)\right)\\
    +&\int_{\partial \Sigma} f(x)L_g\left(Z+X_g+\i w Y,\d\ell_g(x)\right)-\int_{\partial \Sigma} f(x)L_g\left(Z+X_g+\i w Y*u_{\delta},\d\ell_g(x)\right)\bigg)-1\bigg|^2\bigg]
\end{split}
\end{equation}
goes to zero as well. By Corollary \ref{coro_expo_momen},
\begin{equation*}\label{eq_expo_momen_dist}
\begin{split}
    \mathbb E\bigg[\exp\bigg(&\bigg|\int_{\Sigma} f(x)M_g\left(Z+X_g+\i w Y,\d v_g(x)\right)-\int_{\Sigma} f(x)M_g\left(Z+X_g+\i w Y*u_{\delta},\d v_g(x)\right)\bigg|\\
    +&\bigg|\int_{\partial \Sigma} f(x)L_g\left(Z+X_g+\i w Y,\d\ell_g(x)\right)-\int_{\partial \Sigma} f(x)L_g\left(Z+X_g+\i w Y*u_{\delta},\d\ell_g(x)\right)\bigg|\bigg)\bigg]\\
   \leq& e^{C V(\varrho_1)}\sqrt[3]{\left(1+C\sqrt{U_1(D,\varrho_2)}e^{CU_1(D,\varrho_2)}\right)\left(1+C\sqrt{U_1(\Sigma\setminus D,\varrho_2)}e^{CU_1(\Sigma\setminus D,\varrho_2)}\right)}\\
        &\sqrt[3]{\left(1+C\sqrt{U_2(\paN\Sigma,\varrho_3)}e^{CU_2(\paN\Sigma,\varrho_3)}\right)},
\end{split}
\end{equation*}
where $\varrho_2:=f(e^{-\beta wY}-e^{-\beta wY*u_\delta}),\varrho_3:=f(e^{-\beta wY/2}-e^{-\beta wY*u_\delta/2})$ and $\varrho_1:=\max(\varrho_2,\varrho_3)$ goes to 1 as $\delta\to0$ locally uniformly in $w$. Hence \eqref{eq_correlation_converge_d} converges to zero, and \eqref{eq_correlation_1} is holomorphic indeed.
\end{proof}
\begin{rem*}
    Proposition \ref{prop_im_CM} can be generalized to multiple fields $Y_1,Y_2,\cdots$ by induction. Indeed, consider $Y_1:=\int 2\pi G_g(x,y)\d m_1(y)$, then $|e^{-\beta\alpha_1 Y_1}|$ is bounded for fixed $\alpha_1$ and $F(X_g+i\alpha_1 Y_1)$ is again a desirable functional. The shift from $Y_1$ to $\int X\d m_2$ is $e^{\i\alpha_2\int \i\alpha_1Y_1\d m_2}=e^{-\mathbb E[\alpha_1\int X\d m_1\alpha_2\int X\d m_2]}$.
\end{rem*}

\section{BCILT path integrals}\label{sec_path_int}

In this section, we define the path integrals in BCILT and verify that they define a CFT satisfying Segal's axioms.

\subsection{Correlation functions}\label{bcilt-corr}

Let $\Sigma$ be an extended surface with $\paD\Sigma=\varnothing$ and $g$ a Neumann extendible conformal metric on $\Sigma$. Fix $\vbf$ as in \cref{sep-family-extended}. We define the following functionals on the Liouville field $\phi_g=[(X_g,c+\ol{I_{x_0}(2\pi R\omega)})]$ on $\Sigma$. More precisely, we define them on $\TR\times\Dc_0'(\Sigma,\Rb)\times H^1(\Sigma\setminus\zbf)$ and push forward to $\Dc'^\N(\Sigma\setminus\zbf,\TR)$ as in \cref{measure1}. Then well-definedness amounts to checking that they do not depend on the choices of $x_0,\omega$.
\begin{itemize}
    \item For the curvature terms
\[\frac{\i Q}{4\pi}\int_\Sigma K_g\phi_g\,\d v_g+\frac{\i Q}{2\pi}\int_{\pa\Sigma}k_g\phi_g\,\d\ell_g,\]
we define (recall \cref{curv-term-section})
\[K_{\Sigma,g}(\phi_g)=e^{\i Q\chi(\Sigma)c}\,e^{\i QRK_{\Sigma,g,x_0}^\deltabf(\omega)}\,e^{\frac{\i Q}{4\pi}\int_\Sigma K_gX_g\,\d v_g},\]
where $\deltabf$ is any separating family of $\Sigma$ with respect to $\vbf$ disjoint from $x_0$. If $QR\in2\Zb$, this is well-defined by \cref{curv-inv,curv-term-calculations}. (Here $\int_{\pa\Sigma}k_gX_g\,\d\ell_g=0$ since $\pa\Sigma$ is necessarily geodesic.)
\item For the potential terms
\[\mu\int_\Sigma e^{\i\beta\phi_g}\,\d v_g+\int_{\pa\Sigma}\mu_\pa\,e^{\i\frac{\beta}{2}\phi_g}\,\d\ell_g,\]
we define (recall \cref{GMC})
\begin{align*}
M_g^\beta(\phi_g,\d v_g)&=e^{\i\beta(c+I_{x_0}(2\pi R\omega))}M_g^\beta(X_g,\d v_g),\\
L_g^\beta(\phi_g,\d\ell_g)&=e^{\i\frac{\beta}{2}(c+I_{x_0}(2\pi R\omega))}L_g^\beta(X_g,\d\ell_g),
\end{align*}
and
\[M_{\Sigma,g}(\phi_g)=\mu\int_\Sigma M_g^\beta(\phi_g,\d v_g)+\int_{\pa\Sigma}\mu_\pa\,L_g^\beta(X_g,\d\ell_g).\]
If $\beta R\in2\Zb$ (or $\Zb$ if $\mu_\pa=0$), this is well-defined.
\item For the electric operators
\[\prod_{j=1}^\s e^{\i\alpha_j\phi_g(z_j)}\prod_{j=1}^\t e^{\i\frac{\eta_j}{2}\phi_g(x_j)},\]
it is necessary to regularize, since $\phi_g$ is not defined pointwise. For $\varepsilon>0$, we introduce the averaging operator
\[(T_\varepsilon f)(z)=\frac{1}{\ell_g(\pa D_g(z,\varepsilon))}\int_{\pa D_g(z,\varepsilon)}f\,\d\ell_g,\]
where $D_g(z,r)$ denotes the geodesic disk of radius $r$ centered at $z$ with respect to $g$. For the GFF $X_g$, we write $X_{g,\varepsilon}=T_\varepsilon X_g$. At each $z_j$, for $\alpha_j\in\Rb$, we define
\[V_{\alpha_j,g,\varepsilon}(v_j)=\varepsilon^{-\frac{\alpha_j^2}{2}}e^{\i\alpha_jc}\,e^{\i\alpha_jI_{x_0}(2\pi R\omega)(v_j)}\,e^{\i\alpha_jX_{g,\varepsilon}(z_j)},\]
where $I_{x_0}(2\pi R\omega)(v_j)=\int_{\gamma_j}2\pi R\omega$, $\gamma_j$ is a path on $\Sigma$ from $x_0$ to $z_j$ such that $\gamma_j^\circ\subset\Sigma\setminus\zbf$ and the derivative of $\gamma_j$ at $z_j$ is $v_j$. If $\alpha_jR\in\Zb$, this is well-defined up to $o(1)$ as $\varepsilon\to0$. Indeed, changing the representative $\omega$ to $\omega+\frac{1}{2\pi R}df$ amounts to translating $X_g$ by $h=f-\frac{1}{\vol_g\Sigma}\int_\Sigma f\,\d v_g$ (as in the proof of \cref{measure1}), which gives a multiplicative factor $e^{\i\alpha_j(h(z_j)-T_\varepsilon h(z_j))}=1+o(\varepsilon)$ as $\varepsilon\to0$. Thus it will be clear that the limit as $\varepsilon\to0$ is well-defined. Similarly, at each $x_j$, for $\eta_j\in\Rb$, we define
\[V_{\eta_j,g,\varepsilon}(x_j)=\varepsilon^{-\frac{\eta_j^2}{4}}e^{\i\frac{\eta_j}{2}c}\,e^{\i\frac{\eta_j}{2}I_{x_0}(2\pi R\omega)(x_j)}\,e^{\i\frac{\eta_j}{2}X_{g,\varepsilon}(x_j)}.\]
If $\eta_jR\in2\Zb$, this is well-defined up to $o(1)$ as $\varepsilon\to0$. To simplify the notation, we write
\[V_{g,\varepsilon}(\phi_g)=\prod_{j=1}^\s V_{\alpha_j,g,\varepsilon}(v_j)\prod_{j=1}^\t V_{\eta_j,g,\varepsilon}(x_j).\]
\end{itemize}

\begin{rem*}
It is possible to define $V_{g,\varepsilon}$ as an exactly well-defined (i.e., not just up to $o(1)$) function on $\Dc'^\N(\Sigma\setminus\zbf,\TR)$, so that the integral below is well-defined for a fixed $\varepsilon$. One way to do this is as follows. Fix a separating family $\deltabf$ of $\Sigma$ with respect to $\vbf$ that is disjoint from $x_0$ and $\xbf$. We define
\[V_{\alpha_j,g,\varepsilon}(v_j)=\varepsilon^{-\frac{\eta_j^2}{4}}e^{\i\frac{\eta_j}{2}c}\,e^{\i\frac{\eta_j}{2}T_\varepsilon I_{x_0}^\deltabf(2\pi R\omega)(x_j)}\,e^{\i\frac{\eta_j}{2}X_{g,\varepsilon}(x_j)},\]
where $I_{x_0}^\deltabf(2\pi R\omega):\Sigma\setminus\deltabf\to\Rb$ is as in \cref{curv-term-section}, and similarly for $V_{\eta_j,g,\varepsilon}(x_j)$. Then $V_{g,\varepsilon}$ is really invariant under changes of the representatives $\omega$. Moreover, this definition has the advantage that $V_{g,\varepsilon}$ is well-defined as long as $\sum\alpha_j+\sum\frac{\eta_j}{2}\in\Zb/R$ (rather than requiring that $\alpha_j,\frac{\eta_j}{2}\in\Zb/R$ for each $j$). Indeed, with this definition, the only term in $V_{g,\varepsilon}$ that is only defined up to $2\pi R\Zb$ is $e^{\i\big(\sum\alpha_j+\sum\frac{\eta_j}{2}\big)c}$, and this makes sense if $\sum\alpha_j+\sum\frac{\eta_j}{2}\in\Zb/R$. Intuitively, the reason is that the choice of a separating family allows us to relate the value of $I_{x_0}(\omega)$ at different points, so that we have a global ambiguity in $2\pi R\Zb$ instead of one for each $z_j$.

This definition of $V_{g,\varepsilon}$ depends on the choice of a separating family $\deltabf$. However, since $(T_\varepsilon I_{x_0}^\deltabf(2\pi R\omega))(z_j)\to I_{x_0}(2\pi R\omega)(v_j)$ (resp.\ $(T_\varepsilon I_{x_0}^\deltabf(2\pi R\omega))(x_j)\to I_{x_0}(2\pi R\omega)(x_j)$), the limit as $\varepsilon\to0$ is well-defined and coincides with the limit for the previous definition of $V_{g,\varepsilon}$. In practice, it does not matter which definition we use after taking the limit $\varepsilon\to0$. Note that here our convention in \cref{sep-family-extended} on the direction of the $d_j$ is important. In the convention of \cite{CILT}, it is not generally true that $(T_\varepsilon I_{x_0}^\deltabf(2\pi R\omega))(z_j)\to I_{x_0}(2\pi R\omega)(v_j)$.
\end{rem*}

\begin{defn}\label{def:correlation_function}
For a measurable function $F:\Dc_\mbf'^\N(\Sigma\setminus\zbf,\TR)\to\Cb$, the \textit{BCILT correlation function} of $F$ is
\[\la FV_{\alphabf,\etabf,\mbf}(\vbf,\xbf)\ra_{\Sigma,g}=
\lim_{\varepsilon\to0}\int_{\Dc'^\N_\mbf(\Sigma\setminus\zbf,\TR)}
F(\phi_g)\,V_{g,\varepsilon}(\phi_g)\,K_{\Sigma,g}(\phi_g)^{-1}\,e^{-M_{\Sigma,g}(\phi_g)}\,\d\phi_g.\]
\end{defn}

We define the following function class on $\Dc_\mbf'^\N(\Sigma\setminus\zbf,\TR)$. Let $\Ec(\Dc_\mbf'^\N(\Sigma\setminus\zbf,\TR))$ be the linear space spanned by functions $F:\Dc_\mbf'^\N(\Sigma\setminus\zbf,\TR)\to\Cb$ whose pullback via the homeomorphism in \cref{measure1} has the form
\begin{equation}\label{F1}
F([(X_g,c+\ol{I_{x_0}(2\pi R\omega)})])=e^{\i nc/R}\,P(\la X_g,h_1\ra,\ldots,\la X_g,h_m\ra)\,G([\omega]),
\end{equation}
where $n\in\Zb$, $m\in\Nb$, $P$ is a complex polynomial, $h_1,\ldots,h_m\in\Dc(\Sigma,\Rb)$, $G:H_\mbf^1(\Sigma\setminus\zbf)\to\Cb$ is a bounded function. This space is well-defined, i.e., it does not depend on the choices of $x_0,\omega$. Note that $1\in\Ec(\Dc_\mbf'^\N(\Sigma\setminus\zbf,\TR))$.

\begin{thm}\label{thm_exist_corr}
Suppose $\fr j$, $\alpha_j,\eta_j>Q$. For $F\in\Ec(\Dc_\mbf'^\N(\Sigma\setminus\zbf,\TR))$, $\la FV_{\alphabf,\etabf,\mbf}(\vbf,\xbf)\ra_{\Sigma,g}$ exists.
\end{thm}
\begin{proof}
Unraveling the definitions, it is
\begin{equation}\label{eq_corr_def}
    \begin{split}
&\sqrt{\frac{\vol_g\Sigma}{\det'\Delta_{g}}}
\sum_{[\omega]\in H^1_\mbf(\Sigma\setminus\zbf)}
e^{-\pi R^2\int_\Sigma^\reg|\omega|_g^2\,\d v_g}\,
e^{\i(\sum\alpha_jI_{x_0}(2\pi R\omega)(v_j)+\sum\frac{\eta_j}{2}I_{x_0}(2\pi R\omega)(x_j)-QRK_{\Sigma,g,x_0}^\deltabf(\omega))}\\
&\quad\lim_{\varepsilon\to0}\varepsilon^{-\sum\frac{\alpha_j^2}2-\sum\frac{\eta_j^2}4}\int_0^{2\pi R}e^{\i(\sum\alpha_j+\sum\frac{\eta_j}{2}-Q\chi(\Sigma))c}\,
\Eb\bigg[
e^{-R\int_\Sigma X_g\,d^*\omega\,\d v_g}\,
F(\phi_g)\,e^{\i(\sum\alpha_jX_{g,\varepsilon}(z_j)+\sum\frac{\eta_j}{2}X_{g,\varepsilon}(x_j))}\\
&\hspace{11em} e^{-\frac{\i Q}{4\pi}\int_\Sigma K_gX_g\,\d v_g}\,
e^{-\mu\int_\Sigma e^{\i\beta(c+I_{x_0}(2\pi R\omega))}M_\beta^g(X_g,\d v_g)-\int_{\pa\Sigma}\mu_\pa\,e^{\i\frac{\beta}{2}(c+I_{x_0}(2\pi R\omega))}L_\beta^g(X_g,\d\ell_g)}\bigg]\d c,
\end{split}
\end{equation}
where the expectation $\Eb$ is over the GFF $X_g$. By (1) of \cref{harmonic}, we may assume $d^*\omega=0$.

According to Proposition \ref{prop_im_CM}, the expectation in \eqref{eq_corr_def} can be written as
\begin{equation}\label{eq_prop_corr_CM_1}
    \begin{split}
        &\mathbb E\left[F(\phi_g+u_\varepsilon)e^{-\mu M_g(\phi_g+u_\varepsilon,\Sigma)-\mu_\partial L_g(\phi_g+u_\varepsilon,\partial\Sigma)}\right]\\
        &\times e^{-\frac12\mathbb E[(\sum_j\alpha_jX_{g,\varepsilon}(z_j)+\sum_j\frac{\eta_j}{2}X_{g,\varepsilon}(x_j)-\frac{ Q}{4\pi}\int_\Sigma K_g X_g\,\d v_g-\frac{Q}{2\pi}\int_{\pa\Sigma}k_gX_g\,\d\ell_g)^2]},
    \end{split}
    \end{equation}
    where 
    \begin{equation*}
    \begin{split}
        u_\varepsilon(x):=&\i\sum_j\alpha_j\mathbb E[X_g(x)X_{g,\varepsilon}(z_j)]+\i\sum_j\frac{\eta_j}{2}\mathbb E[X_g(x)X_{g,\varepsilon}(x_j)]\\
        -&\frac{\i Q}{4\pi}\int_\Sigma K_g(y) E[X_g(x)X_{g}(y)]\,\d v_g(y)-\frac{\i Q}{2\pi}\int_{\pa\Sigma}k_gE[X_g(x)X_{g}(y)]\,\d\ell_g(y).
    \end{split}
    \end{equation*}
    Let $X_{g,0}:=X_g$, we have $u_\varepsilon(x)$ converge to $u_0(x)$ pointwisely on $\Sigma\setminus(\zbf\cup\xbf)$.
    We first study the $L^p$ limit of the exponential of GMC.
\begin{lem}\label{lem_exp_conv_gmc}
Suppose $0>\alpha_j>Q,0>\eta_j>Q$ and consider $\phi_g,u_\varepsilon$ defined as above. Then for all $p>1$,
    \begin{equation*}
        \lim_{\varepsilon\to0}\mathbb E[|e^{-\mu M_g(\phi_g+u_\varepsilon,\Sigma)-\mu_\partial L_g(\phi_g+u_\varepsilon,\partial\Sigma)}-e^{-\mu M_g(\phi_g+u_0,\Sigma)-\mu_\partial L_g(\phi_g+u_0,\partial\Sigma)}|^p]=0,
    \end{equation*}
    uniformly over all $\omega$.
\end{lem}
\begin{proof}
    First, by Cauchy-Schwarz inequality, the problem is reduced to
    \begin{equation*}
    \begin{split}
        &\mathbb E[|e^{-\mu M_g(\phi_g+u_\varepsilon,\Sigma)-\mu_\partial L_g(\phi_g+u_\varepsilon,\partial\Sigma)}-e^{-\mu M_g(\phi_g+u_0,\Sigma)-\mu_\partial L_g(\phi_g+u_0,\partial\Sigma)}|^p]\\
        \leq&\mathbb E[|e^{-\mu M_g(\phi_g+u_0,\Sigma)-\mu_\partial L_g(\phi_g+u_0,\partial\Sigma)}|^{2p}]^\frac12\mathbb E[|e^{-\mu (M_g(\phi_g+u_\varepsilon,\Sigma)-M_g(\phi_g+u_0,\Sigma))-\mu_\partial (L_g(\phi_g+u_\varepsilon,\partial\Sigma)-L_g(\phi_g+u_0,\partial\Sigma))}-1|^{2p}]^\frac12.
    \end{split} 
    \end{equation*}
By using $|e^z-1|\leq |e^{|z|}-1|$ and $(e^u-1)^p\leq C(e^{pu}-1),u\geq0,p>1$, it is enough to study for $\lambda>0$, 
\begin{equation}\label{eq_prop_exp_conv_1}
   \mathbb E[e^{\lambda |M_g(\phi_g+u_\varepsilon,\Sigma)-M_g(\phi_g+u_0,\Sigma)|+\lambda |L_g(\phi_g+u_\varepsilon,\partial\Sigma)-L_g(\phi_g+u_0,\partial\Sigma)|}]
\end{equation}
goes to zero. According to Corollary \ref{coro_expo_momen}, \eqref{eq_prop_exp_conv_1} is smaller than
\begin{equation*}
\begin{split}
    e^{C V(\varrho_1)}&\sqrt[3]{\left(1+C\sqrt{U_1(D,\varrho_2)}e^{CU_1(D,\varrho_2)}\right)\left(1+C\sqrt{U_1(\Sigma\setminus D,\varrho_2)}e^{CU_1(\Sigma\setminus D,\varrho_2)}\right)}\\
        &\sqrt[3]{\left(1+C\sqrt{U_2(\paN\Sigma,\varrho_3)}e^{CU_2(\paN\Sigma,\varrho_3)}\right)},
\end{split}
\end{equation*}
where $U_1,U_2,V$ are all defined in \ref{eq_defn_U1U2V} and $\varrho_2:=e^{\i\beta u_\epsilon}-e^{-\i\beta u_0},\varrho_3:=e^{\i\beta u_\epsilon/2}-e^{-\i\beta u_0/2}$ and $\varrho_1:=\max(\varrho_2,\varrho_3)$.

Now it is plain to see our bound has nothing to do with $\omega$. Thus we just need to prove $U_1,U_2$ and $V$ all converge to zero, and we will dominate all $\i u_\varepsilon$ as follows. We could forget about the curvature terms in $u_\varepsilon$ since they do not depend on $\varepsilon$.

For the green function we have the following inequality
\begin{equation*}
    G_g(x,y)\leq -\frac1{2\pi}\log|d_g(x,y)|-\frac1{2\pi}\log|d_g(x,\tau(y))|+C,
\end{equation*}
where $C$ is a constant independent of $x,y\in\bar\Sigma$. Then we have a bound for $\i u_\varepsilon$ uniformly over all $x\in\bar\Sigma$ and $\varepsilon$ small (where the $C$ might change)
\begin{equation*}
   \i u_\varepsilon(x)\leq \sum_j\alpha_j(\log|d_g(x,z_j)|+\log|d_g(x,\tau(z_j))|)+\sum_j\eta_j\log|d_g(x,x_j)|+C=:\i u'(x).
\end{equation*}
Therefore, to prove $U_1,U_2$ and $V$ converge to zero we only need to show that $U_1(\cdot,e^{\i\beta u'}),U_2(\cdot,e^{\i\beta u'/2}),\\V(\cdot,\max(e^{\i\beta u'},e^{\i\beta u'/2}))$ are finite and thus dominate $U_1(\cdot,\varrho_2),U_2(\cdot,\varrho_3),V(\cdot,\varrho_1)$.
The evaluations of $U_1,U_2$ are determined by the following basic integrals 
\begin{equation*}
    \iint_{\mathbb D^2}\frac{|x|^{\beta\alpha}|y|^{\beta\alpha}\mathrm{d}x\mathrm{d}y}{|x-y|^{\beta^2}},\iint_{(\mathbb D\cap\mathbb H)^2}\frac{|x|^{\beta\eta}|y|^{\beta\eta}}{|x-y|^{\beta^2}}+\frac{|x|^{\beta\eta}|y|^{\beta\eta}}{|x-\bar y|^{\beta^2}}\mathrm{d}x\mathrm{d}y,\iint_{[-1,1]^2}\frac{|x|^{\beta\eta/2}|y|^{\beta\eta/2}\mathrm{d}x\mathrm{d}y}{|x-y|^{\beta^2/2}}.
\end{equation*}
which are finite if $\beta\alpha-\beta^2/2>-2,\beta\eta-\beta^2/2>-2,$ and $\beta\eta/2-\beta^2/4>-1$, i.e. $\alpha>Q$ and $\eta>Q$. For $V$, we need
\begin{equation*}
    \int_{\mathbb D\cap\mathbb H}|x|^{\beta\eta}d(x,[-1,1])^{-\beta^2/2}\ \d x=\int_{0}^{\pi}\int_{0}^1r^{\beta\eta}|r\sin\theta|^{-\beta^2/2}r\ \d r\d \theta,
\end{equation*}
which is again $\eta>Q$.
\end{proof}
The $L^2$ convergence of $M_g(\phi_g+u_\varepsilon,\Sigma)$ ($L_g(\phi_g+u_\varepsilon,\partial\Sigma)$ resp.) can be proved similarly. Let $W(x):=\lim_{\varepsilon\to0}\mathbb E[X_{g,\varepsilon}^2(x)]+\log\varepsilon$. We point out that 
\begin{equation*}
    \mathbb E[|M_g(\phi_g+u_0,\Sigma)|^2]=\iint_{\Sigma^2} e^{\i\beta(u_{0}(x)-\bar u_{0}(y))}e^{\beta^2G(x,y)-\beta^2/2(W(x)+W(y))+\i\beta (I_{x_0}(2\pi R\omega)(x)-I_{x_0}(2\pi R\omega)(y))}\d v_g(x)\d v_g(y)
\end{equation*}
and
\begin{equation*}
    \mathbb E[|L_g(\phi_g+u_0,\partial\Sigma)|^2]=\iint_{\partial\Sigma^2} e^{\i\frac\beta2(u_{0}(x)-\bar u_{0}(y))}e^{\frac{\beta^2}4G(x,y)-\frac{\beta^2}8(W(x)+W(y))+\i\frac\beta2 (I_{x_0}(2\pi R\omega)(x)-I_{x_0}(2\pi R\omega)(y))}\d \ell_g(x)\d \ell_g(y).
\end{equation*}
are finite. The control of these second moments is equivalent to the control of the following integrals
\begin{equation*}
    \iint_{\mathbb D^2}\frac{|x|^{\beta\alpha}|y|^{\beta\alpha}\mathrm{d}x\mathrm{d}y}{|x-y|^{\beta^2}},\iint_{(\mathbb D\cap\mathbb H)^2}\frac{|x|^{\beta\eta}|y|^{\beta\eta}d(x,\bar x)^{\beta^2/2}d(y,\bar y)^{\beta^2/2}\mathrm{d}x\mathrm{d}y}{|x-y|^{\beta^2}|x-\bar y|^{\beta^2}},\iint_{[-1,1]^2}\frac{|x|^{\beta\eta/2}|y|^{\beta\eta/2}\mathrm{d}x\mathrm{d}y}{|x-y|^{\beta^2/2}},
\end{equation*}
which are finite if $\beta\alpha-\beta^2/2>-2,\beta\eta-\beta^2/2>-2,$ and $\beta\eta/2-\beta^2/4>-1$, i.e. $\alpha>Q$ and $\eta>Q$. Note that $\sqrt{d(x,\tau (x))d(y,\tau (y))}\leq d(x,\tau(y))$.
Now we sketch the proof of convergence of 
\begin{equation*}
    \begin{split}
&\sum_{[\omega]\in H^1_\mbf(\Sigma\setminus\zbf)}
e^{-\pi R^2\int_\Sigma^\reg|\omega|_g^2\,\d v_g}\,
e^{\i(\sum\alpha_jI_{x_0}(2\pi R\omega)(v_j)+\sum\frac{\eta_j}{2}I_{x_0}(2\pi R\omega)(x_j)-QRK_{\Sigma,g,x_0}^\deltabf(\omega))}\\
&\quad\lim_{\varepsilon\to0}\int_0^{2\pi R}e^{\i(\sum\alpha_j+\sum\frac{\eta_j}{2}-Q\chi(\Sigma))c}\,
\mathbb E\left[F(\phi_g+u_\varepsilon)e^{-\mu M_g(\phi_g+u_\varepsilon,\Sigma)-\mu_\partial L_g(\phi_g+u_\varepsilon,\partial\Sigma)}\right]\d c,
\end{split}
\end{equation*}
where $\varepsilon^{-\sum\alpha_j^2/2-\sum\eta_j^2/8}e^{-\frac12\mathbb E[(\sum_j\alpha_jX_{g,\varepsilon}(z_j)+\sum_j\frac{\eta_j}{2}X_{g,\varepsilon}(x_j)-\frac{ Q}{4\pi}\int_\Sigma K_g X_g\,\d v_g-\frac{Q}{2\pi}\int_{\pa\Sigma}k_gX_g\,\d\ell_g)^2]}$ converges to a constant so we leave it out.

Similar to \cite[Theorem 6.11]{CILT}, for $p>1$ and $1/p+1/q=1$, we need to control two parts:
\begin{equation*}
   \left( \int_0^{2\pi R}
\mathbb E\left[\left |F(\phi_g+u_\varepsilon)-F(\phi_g+u_0)\right|^p\right]\d c\right)^{1/p}\left(\int_0^{2\pi R}
\mathbb E\left[\left |e^{-\mu M_g(\phi_g+u_\varepsilon,\Sigma)-\mu_\partial L_g(\phi_g+u_\varepsilon,\partial\Sigma)}\right|^q\right]\d c\right)^{1/q},
\end{equation*}
and
\begin{equation*}
   \left( \int_0^{2\pi R}
\mathbb E\left[\left |F(\phi_g+u_0)\right|^p\right]\d c\right)^{1/p}\left(\int_0^{2\pi R}
\mathbb E\left[\left |e^{-\mu M_g(\phi_g+u_\varepsilon,\Sigma)-\mu_\partial L_g(\phi_g+u_\varepsilon,\partial\Sigma)}-e^{-\mu M_g(\phi_g+u_0,\Sigma)-\mu_\partial L_g(\phi_g+u_0,\partial\Sigma)}\right|^q\right]\d c\right)^{1/q}.
\end{equation*}
And their convergences are independent of $\omega$.

We only need to consider $F$ of the form \eqref{F1}, then $\left |F(\phi_g+u_\varepsilon)-F(\phi_g+u_0)\right|$ is determined up to some polynomial of $\la u_\varepsilon-u_0,f_{[\omega],1}\ra,\cdots$ . The second one is already proved in Lemma \ref{lem_exp_conv_gmc}.

Finally,
\begin{equation*}
\sum_{[\omega]\in H^1_\mbf(\Sigma\setminus\zbf)}
e^{-\pi R^2\int_\Sigma^\reg|\omega|_g^2\,\d v_g}\,
e^{\i(\sum\alpha_jI_{x_0}(2\pi R\omega)(v_j)+\sum\frac{\eta_j}{2}I_{x_0}(2\pi R\omega)(x_j)-QRK_{\Sigma,g,x_0}^\deltabf(\omega))}
\end{equation*}
which converges absolutely by \cref{cohomology-summable}.
\end{proof}

Suppose $F$ has \textit{pure degree} $n$ in the sense that $F(\phi+c)=e^{\i nc/R}F(\phi)$ for $c\in\TR$. Expanding the exponentials in the GMCs and collecting the terms involving the zero mode $c$, we see that the path integral has the form
\[\sum_{p,q\geq0}\,a_{p,q}\exp\left(\i\left(\frac{n}{R}+\sum_{j=1}^\s\alpha_j+\sum_{j=1}^\t\frac{\eta_j}{2}-Q\chi(\Sigma)+p\beta+q\frac{\beta}2\right)c\right),\]
where the coefficients $a_{p,q}$ do not depend on $c$. Integrating over $c$, the terms that do not vanish must satisfy the \textit{neutrality condition}
\begin{equation}\label{neutrality-condition}
\frac{n}{R}+\sum_{j=1}^\s\alpha_j+\sum_{j=1}^\t\frac{\eta_j}{2}-Q\chi(\Sigma)+p\beta+q\beta/2=0.
\end{equation}
There are finitely many such terms. In particular, $\la V_{\alphabf,\etabf,\mbf}(\vbf,\xbf)\ra_{\Sigma,g}=0$ unless
\[\sum_{j=1}^\s\alpha_j+\sum_{j=1}^\t\frac{\eta_j}{2}\leq Q\chi(\Sigma).\]
This is analogous to the second Seiberg bound for real LCFT.

\subsection{Segal's amplitudes}\label{bcilt-amplitude}

Let $\Sigma$ be an extended surface with $\paD\Sigma\neq\varnothing$ and $g$ a Neumann extendible conformal metric on $\Sigma$. Fix $\vbf$ as usual. Similarly to the previous section, we define the following functionals on the Liouville field $\phi_g=[(X_g,\ol{P_\Sigma\varphibf}+\ol{2\pi Rf}+I_{x_0}(2\pi R\omega)+c_0)]$ as in \cref{measure2}.
\begin{itemize}
\item For the curvature terms, we define
\[K_{\Sigma,g}(\phi_g)=e^{\i Q\chi(\Sigma)c_0}\,e^{\i QRK_{\Sigma,g,x_0}^\deltabf(\omega)}\exp\left(\frac{\i Q}{4\pi}\int_\Sigma K_g(X_g+P_\Sigma\varphibf+2\pi Rf)\,\d v_g+\frac{\i Q}{2\pi}\int_{\paD\Sigma}k_g(\varphibf+2\pi Rf)\,\d\ell_g\right)
,\]
where $\deltabf$ is any separating family of $\Sigma$ with respect to $\vbf$ disjoint from $x_0$. If $QR\in4\Zb$ (or $2\Zb$ if $\Sigma$ has no corners), this is well-defined.
\item For the potential terms, we define
\begin{align*}
M_g^\beta(\phi_g,\d v_g)&=e^{\i\beta(c_0+I_{x_0}(2\pi R\omega))}e^{\i\beta(P_\Sigma\varphibf+2\pi Rf)}M_g^\beta(X_g,\d v_g),\\
L_g^\beta(\phi_g,\d\ell_g)&=e^{\i\frac{\beta}{2}(c_0+I_{x_0}(2\pi R\omega))}e^{\i\frac{\beta}{2}(P_\Sigma\varphibf+2\pi Rf)}L_g^\beta(X_g,\d\ell_g),
\end{align*}
and
\[M_{\Sigma,g}(\phi_g)=\mu\int_\Sigma M_g^\beta(\phi_g,\d v_g)+\int_{\paN\Sigma}\mu_\pa\,L_g^\beta(\phi_g,\d\ell_g).\]
If $\beta R\in2\Zb$ (or $\Zb$ if $\mu_\pa=0$), this is well-defined.
\item For the electric operators, we define for $\varepsilon>0$
\begin{align*}
V_{\alpha_j,g,\varepsilon}(v_j)&=\varepsilon^{-\frac{\alpha_j^2}{2}}e^{\i\alpha_jP_\Sigma\varphibf(z_j)}\,e^{\i\alpha_jc_0}\,e^{\i\alpha_jI_{x_0}(2\pi R\omega)(v_j)}\,e^{\i\alpha_j2\pi Rf(z_j)}\,e^{\i\alpha_jX_{g,\varepsilon}(z_j)},\\
V_{\eta_j,g,\varepsilon}(x_j)&=\varepsilon^{-\frac{\eta_j^2}{4}}e^{\i\frac{\eta_j}{2}P_\Sigma\varphibf(x_j)}\,e^{\i\frac{\eta_j}{2}c_0}\,e^{\i\frac{\eta_j}{2}I_{x_0}(2\pi R\omega)(x_j)}\,e^{\i\frac{\eta_j}{2}2\pi Rf(x_j)}\,e^{\i\frac{\eta_j}{2}X_{g,\varepsilon}(x_j)},
\end{align*}
and
\[V_{g,\varepsilon}(\phi_g)=\prod_{j=1}^\s V_{\alpha_j,g,\varepsilon}(v_j)\prod_{j=1}^\t V_{\eta_j,g,\varepsilon}(x_j).\]
The same remark applies concerning the well-definedness of this functional.
\end{itemize}

\begin{defn}\label{def:amplitude}
For a measurable function $F:\Dc_\mbf'^\M(\Sigma\setminus\zbf,\TR)\to\Cb$, the \textit{BCILT amplitude} of $F$ is the function $\Ac_{\Sigma,g,\alphabf,\etabf,\vbf,\xbf,\mbf,\zetabf}(F):\Dc'(\paS,\TR)\to\Cb$ defined by
\[\Ac_{\Sigma,g,\alphabf,\etabf,\vbf,\xbf,\mbf,\zetabf}(F,\bt^\kbf)=\lim_{\varepsilon\to0}
\int_{\Dc'^{\M,\bt^\kbf}_\mbf(\Sigma\setminus\zbf,\TR)}
F(\phi_g)\,V_{g,\varepsilon}(\phi_g)\,
K_{\Sigma,g}(\phi_g)^{-1}\,e^{-M_{\Sigma,g}(\phi_g)}\,\d\phi_g.\]
\end{defn}

We define the following function class on $\Dc_\mbf'^\M(\Sigma\setminus\zbf,\TR)$. The idea is that it should be a polynomial in $\phi_g$. Recall that we have a continuous surjection
\[\setlength\arraycolsep{1pt}
\begin{array}[t]{ccccccc}
\Dc'(\Sigma,\Rb) & \times & \Dc_0'(\paS,\Rb) & \times & C_\mbf^{\M,\const}(\Sigma\setminus\zbf,\TR) & \to & \Dc_\mbf'^\M(\Sigma\setminus\zbf,\TR)\\
X_g&&\varphibf&&\psi&\mapsto&[(X_g,\ol{P_\Sigma\varphibf}+\psi)]
\end{array}\]
Fix a representative in each connected component of $C_\mbf^{\M,\const}(\Sigma\setminus\zbf,\TR)$. If $\psi_0$ is a representative, then any element $\psi$ in its connected component $[\psi_0]$ has the form $\psi_0+\ol{f}$ where $f\in C^\infty(\Sigma,\Rb)$ is Neumann extendible and locally constant near $\paD\Sigma$. Here $f$ is well-defined up to $2\pi R\Zb$, and $\ol{f}|_{\paD\Sigma}$ can be regarded as an element in $\TR^{\bD}$. Let $\Ec(\Dc_\mbf'^\M(\Sigma\setminus\zbf,\TR))$ be the linear space spanned by functions $F:\Dc_\mbf'^\M(\Sigma\setminus\zbf,\TR)\to\Cb$ whose pullback via the surjection above has the form
\begin{equation}\label{F2}
F([(X_g,\ol{P_\Sigma\varphibf}+(\psi_0+\ol{f}))])=e^{\i\frac{\nbf}{R}\cdot \ol{f}|_{\paD\Sigma}}\,P(\la X_g+P_\Sigma\varphibf+f,h_1\ra,\ldots,\la X_g+P_\Sigma\varphibf+f,h_m\ra)\,G([\psi_0]),
\end{equation}
where $\psi_0$ is a chosen representative, $\nbf\in\Zb^{\bD}$, $m\in\Nb$, $P$ is a complex polynomial, $h_1,\ldots,h_m\in\Dc(\Sigma,\Rb)$ with zero integral over $\Sigma$, $G:\pi_0(C_\mbf^{\M,\const}(\Sigma\setminus\zbf,\TR))\cong H^1(\Sigma\setminus\zbf)\to\Cb$ is bounded. Here $\frac{\nbf}{R}\cdot \ol{f}|_{\paD\Sigma}=\sum_{j=1}^{\bD}\ol{\frac{n_j}{R}f|_{c_j^\D}}$ is well-defined up to $2\pi\Zb$, $X_g+f$ is well-defined up to $2\pi R\Zb$ (recall that the equivalence relation on $\Dc_\mbf'^\M(\Sigma\setminus\zbf,\TR)$ allows only changing the representatives by a function vanishing on the Dirichlet boundary), so this is well-defined. It is easy to check that this function space does not depend on the choices of the representatives $\psi_0$. Note that $1\in\Ec(\Dc_\mbf'^\M(\Sigma\setminus\zbf,\TR))$.

\begin{thm}\label{thm_exist_amplitude}
Suppose $\fr j$, $\alpha_j,\eta_j>Q$. For $F\in\Ec(\Dc_\mbf'^\M(\Sigma\setminus\zbf,\TR))$, $\Ac_{\Sigma,g,\alphabf,\etabf,\vbf,\xbf,\mbf,\zetabf}(F)$ exists as a limit in $L^2(\Dc'(\paS,\TR))$, \blue{where $\Dc'(\paS,\TR)$ is equipped with the measure defined in \cref{section:measure-space1}}.
\end{thm}

We postpone the proof to \cref{section:gluing-bcilt}.

One should think of $\Ac_{\Sigma,g,\alphabf,\etabf,\vbf,\xbf,\mbf,\zetabf}(F)$ as being antilinear (resp.\ linear) on the outgoing (resp.\ incoming) Dirichlet (semi)circles, so that it defines a bounded linear map $\Hc^{\otimes\bD^-}\otimes\Hc_+^{\otimes\bMD^-}\to\Hc^{\otimes\bD^+}\otimes\Hc_+^{\otimes\bMD^+}$, where $\bD^\pm$, $\bMD^\pm$ denote the number of outgoing or incoming circles or semicircles, respectively. See \cite[Section 7]{Segal} for a detailed discussion.

For an extended surface $\Sigma$ with $\paD\Sigma=\varnothing$, identifying functions $\Dc'(\paS,\TR)=\{*\}\to\Cb$ with points in $\Cb$, we define $\Ac_{\Sigma,g,\alphabf,\etabf,\vbf,\xbf,\mbf,*}(F)=\la FV_{\alphabf,\etabf,\mbf}(\vbf,\xbf)\ra_{\Sigma,g}$, in accordance with the convention in \cref{gluing-free} for unifying the cases $\paD\Sigma=\varnothing$ and $\paD\Sigma\neq\varnothing$. We view correlation functions as a special case of amplitudes in this sense. One should bear in mind, however, that the definitions are fundamentally different.

\subsection{Weyl anomaly}

\begin{thm}[\textbf{Weyl anomaly}]\label{Weyl-anomaly}
For $\rho\in C^\infty(\Sigma,\Rb)$ Neumann extendible with $\rho|_{\paD\Sigma}=0$ if $\paD\Sigma\neq\varnothing$, $F\in\Ec(\Dc_\mbf'^\M(\Sigma\setminus\zbf,\TR))$, we have
\begin{align*}
\Ac_{\Sigma,e^\rho g,\alphabf,\etabf,\vbf,\xbf,\mbf,\zetabf}(F)&=\Ac_{\Sigma,g,\alphabf,\etabf,\vbf,\xbf,\mbf,\zetabf}\big(F\big(\cdot-\tfrac{\i Q}{2}\rho\big)\big)\\
&\qquad\exp\left(\frac{c_\L}{96\pi}\Big(\int_\Sigma\big(|d\rho|_g^2+2K_g\rho\big)\,\d v_g\Big)-\sum_{j=1}^\s\Delta_{\alpha_j,m_j}\rho(z_j)-\frac{1}{2}\sum_{j=1}^\t\Delta_{\eta_j}\rho(x_j)\right),
\end{align*}
where $c_\L$, $\Delta_{\alpha_j,m_j}$, $\Delta_{\eta_j}$ were defined in \cref{main-theorem}.
\end{thm}

Here $F(\cdot-\tfrac{\i Q}{2}\rho)$ is interpreted similarly to \cref{imaginary-girsanov}. If $\paD\Sigma=\varnothing$ and $F$ is as in \cref{F1}, then $F(\phi_g-\tfrac{\i Q}{2}\rho)=e^{\i n(c-\frac{\i Q}{2}m_g(\rho))/R}\,P(\la X_g-\frac{\i Q}{2}(\rho-m_g(\rho)),h_1\ra,\ldots\la X_g-\frac{\i Q}{2}(\rho-m_g(\rho)),h_m\ra)\,G([\omega])$. If $\paD\Sigma\neq\varnothing$ and $F$ is as in \cref{F2}, then $F(\phi_g-\tfrac{\i Q}{2}\rho)=e^{\i\frac{\nbf}{R}\cdot \ol{f}|_{\paD\Sigma}}\,P(\la X_g+P_\Sigma\varphibf-\frac{\i Q}{2}\rho,h_1\ra,\ldots,\la X_g+P_\Sigma\varphibf-\frac{\i Q}{2}\rho,h_m\ra)\,G([\psi_0])$.

\begin{proof}
Here is how each term changes under a change of the conformal metric:
\begin{itemize}
\item The change of the measure on $\Dc'^{\M,\bt^\kbf}(\Sigma\setminus\zbf,\TR)$ was calculated in \cref{measure-conf1,measure-conf2}.
\item By (4) of \cref{curv-term-calculations},
\[K_{\Sigma,e^\rho g}(\phi_g)=K_{\Sigma,g}(\phi_g)\,\exp\left(\frac{\i Q}{4\pi}\int_\Sigma\la d\rho,d\phi_g\ra_g\,\d v_g\right).\]
\item We have $M_{e^\rho g}^\beta(\phi_g,\d v_{e^\rho g})=e^{-\frac{\beta}{2}Q\rho}M_g^\beta(\phi_g,\d v_g)$, $L_{e^\rho g}^\beta(\phi_g,\d \ell_{e^\rho g})=e^{-\frac{\beta}{2}Q\rho}L_g^\beta(\phi_g,\d\ell_g)$, so formally
\[M_{\Sigma,e^\rho g}(\phi_g)=M_{\Sigma,g}\big(\phi_g+\tfrac{\i Q}{2}\rho\big).\]
\item As $\varepsilon\to0$,
\begin{align*}
V_{\alpha_j,e^\rho g,\varepsilon}(z_j)&=V_{\alpha_j,g,\varepsilon}(z_j)\,e^{-\frac{\alpha_j^2}{4}\rho(z)+o(1)},\\
V_{\eta_j,e^\rho g,\varepsilon}(x_j)&=V_{\eta_j,g,\varepsilon}(x_j)\,e^{-\frac{\eta_j^2}{8}\rho(z)+o(1)}.
\end{align*}
\end{itemize}

First suppose either $\paD\Sigma=\varnothing$ and $\int_\Sigma\rho\,\d v_g=0$ or $\paD\Sigma\neq\varnothing$ and $\rho|_{\paD\Sigma}=0$. Omitting the integral sign $\int_{\Dc_\mbf'^{\M,\bt^\kbf}(\Sigma\setminus\zbf,\TR)}$, the integrand is
\begin{align*}
&F(\phi_{e^\rho g})\,V_{e^\rho g,\varepsilon}(\phi_{e^\rho g})\,K_{\Sigma,e^\rho g}(\phi_{e^\rho g})^{-1}\,e^{-M_{\Sigma,e^\rho g}(\phi_{e^\rho g})}\,\d\phi_{e^\rho g}=F(\phi_g)\,V_{g,\varepsilon}(\phi_g)\,K_{\Sigma,e^\rho g}(\phi_g)^{-1}\,e^{-M_{\Sigma,e^\rho g}(\phi_g)}\,\d\phi_g\\
&\hspace{10em}\exp\left(\frac{1}{96\pi}\int_\Sigma\big(|d\rho|_g^2+2K_g\rho\big)\,\d v_g-\sum_{j=1}^\s\bigg(\frac{\alpha_j^2}{4}+\frac{m_j^2R^2}{4}\bigg)\rho(z_j)-\sum_{j=1}^\t\frac{\eta_j^2}{8}\rho(x_j)+o(1)\right).
\end{align*}
Applying the imaginary Girsanov transform to the GFF $X_g$ with $Y=-\frac{Q}{4\pi}\int_\Sigma\la d\rho,dX_g\ra_g\,\d v_g$, $\Eb[YX_g]=-\frac{Q}{2}\rho$, $\Eb[Y^2]=\frac{Q^2}{8\pi}\int_\Sigma|d\rho|_g^2\,\d v_g$, we get
\begin{align*}
&F(\phi_g)\,V_{g,\varepsilon}(\phi_g)\,K_{\Sigma,e^\rho g}(\phi_g)^{-1}\,e^{-M_{\Sigma,e^\rho g}(\phi_g)}\,\d\phi_g=F\big(\phi_g-\tfrac{\i Q}{2}\rho\big)\,V_{g,\varepsilon}(\phi_g)\,K_{\Sigma,g}(\phi_g)^{-1}\,e^{-M_{\Sigma,g}(\phi_g)}\,\d\phi_g\\
&\hspace{15em}\exp\left(-\frac{Q^2}{16\pi}\int_\Sigma\big(|d\rho|_g^2+2K_g\rho\big)\,\d v_g+Q\sum_{j=1}^\s\frac{\alpha_j}{2}\rho(z_j)+Q\sum_{j=1}^\t\frac{\eta_j}{4}\rho(x_j)\right).
\end{align*}
Summing up, we get the desired formula.

It remains to treat the case $\paD\Sigma=\varnothing$ and $\rho$ constant. Suppose $F$ is as in \cref{F1}. As discussed at the end of \cref{bcilt-corr}, $\la FV_{\alphabf,\etabf,\mbf}(\vbf,\xbf)\ra_{\Sigma,g}=\sum_{p,q}a_{p,q}$ where the sum is over $(p,q)\in\Nb^2$ satisfying the neutrality condition \cref{neutrality-condition}. Under the change from $g$ to $e^\rho g$ for $\rho\in\Rb$ constant, the change of $a_{p,q}$ is
\[\exp\left(-\left(p\beta+q\frac{\beta}{2}\right)\frac{Q}{2}\rho-\sum_{j=1}^\s\left(\frac{\alpha_j^2}{4}+\frac{m_j^2R^2}{4}\right)\rho-\sum_{j=1}^\t\frac{\eta_j^2}{8}\rho\right).\]
By the neutrality condition, this is equal to
\[\exp\left(\frac{n}{R}\frac{Q}{2}\rho-\chi(\Sigma)\frac{Q^2}{2}\rho-\sum_{j=1}^\s\left(\frac{\alpha_j}{2}\left(\frac{\alpha_j}{2}-Q\right)+\frac{m_j^2R^2}{4}\right)\rho-\sum_{j=1}^\t\frac{\eta_j}{4}\left(\frac{\eta_j}{2}-Q\right)\rho\right),\]
as desired.
\end{proof}

\subsection{Spin}

\begin{thm}[\textbf{Spin}]\label{spin}
For $\thetabf\in\Rb^\s$, we have
\[\Ac_{\Sigma,g,\alphabf,\etabf,e^{\i\thetabf}\vbf,\xbf,\mbf,\zetabf}(F)=\Ac_{\Sigma,g,\alphabf,\etabf,\vbf,\xbf,\mbf,\zetabf}(F)\,e^{\i R\sum_{j=1}^\s(\alpha_j-Q)m_j\theta_j}.\]
\end{thm}
\begin{proof}
The terms that depend on $\vbf$ are $K_{\Sigma,g,x_0}^\deltabf(\omega)$ in $K_{\Sigma,g}(\phi_g)$ and $e^{\i\alpha_jI_{x_0}(2\pi R\omega)(v_j)}$ in $V_{\alpha_j,g,\varepsilon}(v_j)$. The change of the former was calculated in \cref{curv-term-spin}. The latter changes by $e^{\i\alpha_jRm_j\theta_j}$, as is easy to check.
\end{proof}

\subsection{Gluing}\label{section:gluing-bcilt}

Using the conventions and notations in \cref{gluing-free}, the gluing theorem for BCILT can be simply stated as
\[\Ac_{\Sigma^\#,g,\alphabf,\etabf,\vbf,\xbf,\mbf,\zetabf_0}(F^\#,\bt_0^{\kbf_0})=C\int_{\Dc'(\zeta^*\Cc,\TR)}\Ac_{\Sigma,g|_\Sigma,\alphabf,\etabf,\vbf,\xbf,\mbf,\zetabf}(F,\wt{\varphi}^k\times\wt{\varphi}^k\times\bt_0^{\kbf_0})\,\d\wt{\varphi}^k,\]
with $C$ as in \cref{gluing}, where $F^\#(\phi)=F(\phi|_\Sigma)$. For the reader's convenience, we spell this out.

\begin{thm}[\textbf{Gluing}]\label{gluing-bcilt}
We resume the notations in \cref{gluing-top}.
\begin{itemize}
\item In the case of gluing two surfaces, we have
\begin{align*}
&\Ac_{\Sigma\#\Sigma',g,\alphabf\times\alphabf',\etabf\times\etabf',\vbf\times\vbf',\xbf\times\xbf',\mbf\times\mbf',\zetabf^\c\times\zetabf'^\c}(F\# F',(\bt^\kbf)^\c\times(\bt'^{\kbf'})^\c)\\
&\qquad=C\int_{\Dc'(\zeta^*\Cc,\TR)}\Ac_{\Sigma,g|_\Sigma,\alphabf,\etabf,\vbf,\xbf,\mbf,\zetabf}(F,\wt{\varphi}^k\times(\bt^\kbf)^\c)\,\Ac_{\Sigma',g|_{\Sigma'},\alphabf',\etabf',\vbf',\xbf',\mbf',\zetabf'}(F',\wt{\varphi}^k\times(\bt'^{\kbf'})^\c)\,\d\wt{\varphi}^k,
\end{align*}
where $(F\#F')(\phi)=F(\phi|_\Sigma)F'(\phi|_{\Sigma'})$, or more precisely, $(F\#F')([(Y,\psi)])=F([(Y|_\Sigma-P_\Sigma(\zeta^*Y|_\Cc\times\mathbf{0}),\psi|_\Sigma+\ol{P_\Sigma(\zeta^*Y|_\Cc\times\mathbf{0})})])\,F'([(Y|_{\Sigma'}-P_{\Sigma'}(\zeta^*Y|_\Cc\times\mathbf{0}),\psi|_{\Sigma'}+\ol{P_{\Sigma'}(\zeta^*Y|_\Cc\times\mathbf{0})})])$.
\item In the case of self-gluing a surface, we have
\[\Ac_{\#\Sigma,g,\alphabf,\etabf,\vbf,\xbf,\mbf,\zetabf^{\c\c}}(F^\#,(\bt^\kbf)^{\c\c})=C\int_{\Dc'(\zeta^*\Cc,\TR)}\Ac_{\Sigma,g|_\Sigma,\alphabf,\etabf,\vbf,\xbf,\mbf,\zetabf}(F,\wt{\varphi}^k\times\wt{\varphi}^k\times(\bt^\kbf)^{\c\c})\,\d\wt{\varphi}^k,\]
where $(\#F)(\phi)=F(\phi|_\Sigma)$, or more precisely, $(\#F)([(Y,\psi)])=F([(Y|_\Sigma-P_\Sigma(\zeta^*Y|_\Cc\times\zeta^*Y|_\Cc\times\mathbf{0}),\psi|_\Sigma+\ol{P_\Sigma(\zeta^*Y|_\Cc\times\zeta^*Y|_\Cc\times\mathbf{0})})])$.
\end{itemize}
Here the constant $C$ is as in \cref{gluing}.
\end{thm}

Here we assume that all functionals are in the function space $\Ec$, so that the amplitudes exist. It is straightforward to check that the function space $\Ec$ is invariant under gluing.

Viewing amplitudes as operators, the case of gluing two surfaces amounts to composing their amplitudes, and the case of self-gluing a surface amounts to taking the trace of its amplitude.

\begin{proof}[Proof of \cref{thm_exist_amplitude,gluing-bcilt}]
Let us denote by $\Ac_{\Sigma,g,\alphabf,\etabf,\vbf,\xbf,\mbf,\zetabf}^\varepsilon(F)$ the regularized amplitude before taking the limit $\varepsilon\to0$, whose existence can be proved using the \textcolor{black}{exponential moment of the Dirichlet GMC}. For a fixed $\varepsilon$, gluing for $\Ac^\varepsilon$ follows from \cref{gluing}. Indeed, it suffices to check gluing for the functionals $K_{\Sigma,g}$, $e^{M_{\Sigma,g}}$, $V_{g,\varepsilon}$. Gluing for $K_{\Sigma,g}$ follows from \cref{curv-gluing}. Gluing for the latter two is obvious.

We shall prove \cref{thm_exist_amplitude} by showing that $\Ac_{\Sigma,g,\alphabf,\etabf,\vbf,\xbf,\mbf,\zetabf}^\varepsilon(F)\in L^2(\Dc'(\paS,\TR))$ for $\varepsilon>0$ and it converges in $L^2$ as $\varepsilon\to0$. Then \cref{gluing-bcilt} follows from gluing for $\Ac^\varepsilon$ by taking the limit $\varepsilon\to0$.

\blue{The idea is to use gluing to write the $L^2$ inner product of amplitudes on $\Sigma$ as a modified correlation function on the \textit{Dirichlet double} $\Sigma^{\D\#2}$, defined as the extended surface obtained by gluing $\Sigma$ and its oppositely oriented copy $\ol{\Sigma}$ along their Dirichlet boundaries via the identity map $\paD\Sigma\xleftrightarrow{\sim}\paD\ol{\Sigma}$ on the Dirichlet boundary (compare \cref{doubling}). Clearly $\pa_\D\Sigma^{\D\#2}=\varnothing$. Again, we denote by $\tau$ the \textit{reflection} on $\Sigma^{\D\#2}$ with respect to $\paD\Sigma=\paD\ol{\Sigma}$, so that $\ol{\Sigma}=\tau(\Sigma)$. By \cref{Weyl-anomaly}, we may assume that the metric $g$ on $\Sigma$ is \textit{Dirichlet extendible} in the sense that it doubles to a metric on $\Sigma^{\D\#2}$, denoted in the same way. For $F:\Dc'^\M(\Sigma\setminus\zbf,\TR)\to\Cb$, we have
\begin{equation*}
\ol{\Ac^{\varepsilon}_{\Sigma,g,\alphabf,\etabf,\vbf,\xbf,\mbf,\zetabf}(F,\bt^\kbf)}=\int_{\Dc'^{\M,\bt^\kbf}_{-\mbf}(\ol{\Sigma}\setminus\tau(\zbf)),\TR)}
\ol{F(\phi_g)}\,V^{-\alphabf,-\etabf}_{g,\varepsilon}(\phi_g)\,
\ol{K_{\ol{\Sigma},g}^{-1}(\phi_g)}\,e^{-\ol{M_{\ol{\Sigma},g}(\phi_g)}}\,\d\phi_g,
\end{equation*}
where $\phi_g$ is now the Liouville field on $\ol{\Sigma}$ with respect to the metric $g|_{\ol{\Sigma}}=\tau_*(g|_\Sigma)$, and we have made explicit the dependence of $V_{g,\varepsilon}$ on the electric charges.
For $\varepsilon,\varepsilon'>0$, by \cref{gluing},
\begin{align*}
&\la\Ac^\varepsilon_{\Sigma,g,\alphabf,\etabf,\vbf,\xbf,\mbf,\zetabf}(F),\Ac^{\varepsilon'}_{\Sigma,g,\alphabf,\etabf,\vbf,\xbf,\mbf,\zetabf}(F)\ra_{L^2(\Dc'(\paS,\TR))}\\
&\hspace{1em}=C\int_{\Dc_{\mbf\times(-\mbf)}'^\N(\Sigma^{\D\#2}\setminus\zbf\cup\tau(\zbf),\TR)}F(\phi_g|_\Sigma)\ol{F(\phi_g|_{\ol{\Sigma}}\circ\tau)}\,V^{\alphabf\times(-\alphabf),\etabf\times(-\etabf)}_{g,\varepsilon}(\phi_g)\\
&\hspace{15em}K_{\Sigma,g}^{-1}(\phi_g|_\Sigma)\ol{K_{\Sigma,g}^{-1}(\phi_g|_\Sigma)}\,e^{-M_{\Sigma,g}(\phi_g|_\Sigma)-\ol{M_{\ol{\Sigma},g}(\phi_g|_{\ol{\Sigma}})}}\,\d\phi_g,
\end{align*}
where $C$ is an absolute constant. While this is not a correlation function per se, one can show that this integral exists and converges as $\varepsilon,\varepsilon'\to0$, similarly to the proof of \cref{thm_exist_corr}. Note that taking $\varepsilon'=\varepsilon$ shows that $\Ac_{\Sigma,g,\alphabf,\etabf,\vbf,\xbf,\mbf,\zetabf}^\varepsilon(F)$ is $L^2$. Thus $\|\Ac^\varepsilon_{\Sigma,g,\alphabf,\etabf,\vbf,\xbf,\mbf,\zetabf}(F)-\Ac^{\varepsilon'}_{\Sigma,g,\alphabf,\etabf,\vbf,\xbf,\mbf,\zetabf}(F)\|^2_{L^2(\Dc'(\paS,\TR))}\to0$ as $\varepsilon,\varepsilon'\to0$.}
\end{proof}

\section{Structure constants}\label{sec_open}

An important step in the understanding of BCILT is the computation of its structure constants.

Suppose now all magnetic charges are zero. Our computations show that electric charges $\alphabf$ and $\etabf$ must satisfy the neutrality condition
\[\sum_{i=1}^\s\alpha_i+\sum_{i=1}^\t\frac{\eta_i}2-Q\chi(\Sigma)+p\beta+q\frac\beta2=0,p,q\in\{0,1,2,\cdots\},\]
where $\chi(\Sigma)$ is the Euler characteristic of $\Sigma$. Hence, for a reasonable $\sum_{i=1}^\s\alpha_i$, we have a finite number of $(p,q)$ satisfying the neutrality condition. The correlation function can be written as the sum of several integrals
\begin{equation*}
\begin{split}
    &\la e^{\i\alphabf\cdot\phi(\zbf)}e^{\i\frac\etabf2\cdot\phi(\xbf)}\ra\\
    =&\sum_{p\geq0}\frac{(-1)^{p+q}\mu^p}{p!q!}I(p,q,\mu_\partial,\alphabf,\zbf,\etabf,\xbf).
\end{split}
\end{equation*}

Now we consider a flat disk $\mathbb D$ ($\chi(\mathbb D)=1$) and $\mu_\partial$ is a constant, then $I(p,q,\mu_\partial,\alphabf,\zbf,\etabf,\xbf)=\mu_{\partial}^qI(p,q,\alphabf,\zbf,\etabf,\xbf)$ is proportional to
\begin{equation*}
\begin{split}
   \int_{\mathbb D^p}\int_{\partial\mathbb D^q} \exp&\bigg(2\pi(\sum_{i<j} \alpha_i\alpha_jG(z_i,z_j)+\sum_i\sum_j \alpha_i\frac{\eta_j} 2G(z_i,x_j)+\sum_i\sum_j\alpha_i\beta G(z_i,w_j)+\sum_i\sum_j\alpha_i\frac\beta 2 G(z_i,y_j)\\
    &\sum_{i<j} \frac{\eta_i} 2\frac{\eta_j} 2G(x_i,x_j)+\sum_i\sum_j\frac{\eta_i} 2\beta G(x_i,w_j)+\sum_i\sum_j\frac{\eta_i} 2\frac\beta 2 G(x_i,y_j)\\
    &\sum_{i<j}\beta^2G(w_i,w_j)+\sum_i\sum_j\frac{\beta^2} 2 G(w_i,y_j)+\sum_{i<j}\frac{\beta^2}4G(y_i,y_j))\bigg)\prod_{i}(1-|w_i|^2)^{\beta^2/2}\d y\ \d w ,
\end{split}
\end{equation*}
where $G(z,z')=-\frac{1}{2\pi}\log|z-z'||1-z\bar z'|$ is the Neumann Green function on $\mathbb D$. (In fact, the above integral is the expectation of ``$p$-th moment of bulk GMC times $q$-th moment of boundary GMC'' with Neumann GFF and $\alphabf,\etabf$ insertions).

Suppose $\mu=0$, which forces $p=0$, and we only have one boundary insertion $\eta$ at $x=1$. Then the above integral is proportional to
\begin{equation*}
    \int_{\partial\mathbb D^q}\prod_{k=1}^q|1-y_k|^{\eta\beta/2}\prod_{k<k'}|y_k-y_{k'}|^{\beta^2/2}\d y.
\end{equation*}
According to \cite[page 7, 1.17]{SelbergInt}, it is equal to $M(\eta\beta/4,\eta\beta/4,\beta^2/4)$, which is
\begin{equation*}
    M(\eta\beta/4,\eta\beta/4,\beta^2/4)=\prod_{j=0}^{q-1}\frac{\Gamma(1+\eta\beta/2+j\beta^2/4)\Gamma(1+(j+1)\beta^2/4)}{\Gamma(1+\eta\beta/4+j\beta^2/4)^2\Gamma(1+\beta^2/4)}.
\end{equation*}
When $\eta=0$, it becomes the Fyodorov--Bouchaud formula
\begin{equation*}
    \frac{\Gamma(1+q\beta^2/4)}{\Gamma(1+\beta^2/4)^p}.
\end{equation*}

\paragraph{Question 1:}When $\mu\not=0$ is a constant, the easiest case is one $\alpha$ insertion at $z=0$ and one $\eta$ at $x=1$. Then we need to consider all non-negative integers $p,q$ s.t. $2p+q=2Q-2\alpha-\eta$, which is
\begin{equation}\label{eq_1B_1b}
\begin{split}
    \int_{\mathbb D^p}\int_{\partial\mathbb D^q}&\prod_{j=1}^p|w_j|^{\alpha\beta}|1-w_j|^{\eta\beta}(1-|w_j|^2)^{\beta^2/2}\prod_{j<j'}|w_j-w_{j'}|^{\beta^2}|1-w_j\bar w_{j'}|^{\beta^2}\\
     &\prod_{k=1}^q|1-y_k|^{\eta\beta/2}\prod_{k<k'}|y_k-y_{k'}|^{\beta^2/2}\prod_{j,k}|w_j-y_{k}|^{\beta^2}\d y \d w.
\end{split}     
\end{equation}
It seems like a combination of Dotsenko-Fateev and Selberg integrals, and to our knowledge, little literature studies this type of integral. There is a formula for the integral \cite[page 8, 5.6]{H3+}, but with only one insertion in the bulk (see our Question 2). Besides, \cite[Corollary 1.5]{FZZ} provides a formula for only a bulk integral with term $(1-|w_j|^2)^{\beta^2/2}$. Still, we wonder whether it would be possible to have a \textbf{formula} for \eqref{eq_1B_1b}.
\paragraph{Question 2:} Suppose we only have one insertion $\alpha$ at $z=0$. Because we need $\alpha>Q$ and $\alpha-Q+p\beta+q\beta/2=0$ at the same time, then there is no non-negative $p,q$ and the bulk one-point function is zero. Thus it seems unclear to us what kind of integral would correspond to a bulk one-point structure constant. The boundary two-point correlation has the same issue. A physics literature \cite{bdytimelike} provides these structure constants for boundary time-like Liouville theory, we wonder whether there would be a connection between our structure constants and theirs.

{
  \let\c\origcedilla      
  \bibliographystyle{plain} 
  \bibliography{new}
}

\end{document}